%% file: ck2.tex
\documentclass{amsart}
\input{aux/style}

\input{aux/packages} 
\input{aux/usualcmds}

\input{aux/commands} 
\input{Tikz_Graphs.tex} 

\title{Operadic Calculus for Higher Colour-Kinematics Duality}

\author{Anibal~M.~Medina-Mardones}
\address{Department of Mathematics, Western University, Ontario, Canada}
\email{\href{mailto:anibal.medina.mardones@uwo.ca}{anibal.medina.mardones@uwo.ca}}

\author{Bruno Vallette}
\address{Laboratoire Analyse, G\'eom\'etrie et Applications, Universit\'e Paris Nord 13, Sorbonne Paris Cit\'e, CNRS, UMR 7539, 93430 Villetaneuse, France.}
\email{\href{vallette@math.univ-paris13.fr}{vallette@math.univ-paris13.fr}}



\date{\today}

\begin{document}
	\input{sec/abstract}
	\maketitle
	\tableofcontents
	\input{sec/introduction}
	\input{sec/acknowledgements}
	\input{sec/algebras}
	\input{sec/operads}
	\appendix
	\input{sec/dual_algebras}
	\sloppy
	\printbibliography
\end{document}

%% file: aux/style.tex
\usepackage[
	backend=biber,
	style=alphabetic, 
	sorting=none,   
	sortcites=true,  
	backref=true,
	url=false,
	doi=false,
	isbn=false,
	eprint=false,
	maxnames=99,
	minnames=99]{biblatex}

\renewbibmacro{in:}{} 
\AtEveryBibitem{\clearfield{pages}} 

\DeclareFieldFormat{title}{\myhref{\mkbibemph{#1}}}
\DeclareFieldFormat[book,article,inbook,incollection,inproceedings,patent,thesis,unpublished]{title}{\myhref{#1\isdot}}

\newcommand{\doiorurl}{%
	\iffieldundef{doi}
	{%
		\iffieldundef{url}
		{%
			\iffieldundef{eprint}
			{}
			{https://arxiv.org/abs/\strfield{eprint}}%
		}
		{\strfield{url}}%
	}
	{https://doi.org/\strfield{doi}}%
}

\newcommand{\myhref}[1]{%
	\ifboolexpr{%
		test {\ifhyperref}
		and
		not test {\iftoggle{bbx:eprint}}
		and
		not test {\iftoggle{bbx:url}}
	}
	{\href{\doiorurl}{#1}}
	{#1}%
}

\usepackage[dvipsnames]{xcolor} 
\definecolor{colorlinks}{rgb}{0.36, 0.2, 0.09}

\usepackage{hyperref}

\hypersetup{
	bookmarks=true,
	linktocpage=true,
	breaklinks=true,
	pdfstartview=FitH,
	plainpages=false,
	naturalnames=true,
	colorlinks=true,
	citecolor=colorlinks,
	filecolor=colorlinks,
	linkcolor=colorlinks,
	urlcolor=colorlinks
}

\usepackage[capitalize, noabbrev]{cleveref}
\crefname{subsection}{\S\!\!}{subsections}

\calclayout

\makeatletter
\@namedef{subjclassname@2020}{%
	\textup{2020} Mathematics Subject Classification}
\makeatother

\let\oldtocsection=\tocsection
\let\oldtocsubsection=\tocsubsection

\renewcommand{\tocsection}[2]{\hspace{0em}\vspace{0.1em}\rule{0pt}{14pt}\oldtocsection{#1}{#2}\bf}
\renewcommand{\tocsubsection}[2]{\hspace{2em}\oldtocsubsection{#1}{#2}}

%% file: aux/packages.tex
\usepackage{todo} 

\usepackage[T1]{fontenc}
\usepackage[lf]{Baskervaldx} 

\usepackage{latexsym}
\usepackage{amssymb}
\usepackage{amsthm}
\usepackage{amscd}
\usepackage{graphicx}
\usepackage[all]{xy}
\usepackage{a4wide}
\usepackage{multirow}
\usepackage{tikz}
\usepackage{tikz-cd}

\usepackage{lmodern} 
\usepackage{microtype} 

\usepackage{mathtools} 

\usepackage{enumerate} 
\usepackage{enumitem} 

\usepackage{tabularx}
\usepackage{booktabs} 

%% file: aux/usualcmds.tex
\DeclareMathOperator{\sign}{sign}
\newcommand{\ot}{\otimes}

\newcommand{\Z}{\mathbb{Z}}
\newcommand{\sym}{\mathbb{S}}

\DeclarePairedDelimiter\bars{\lvert}{\rvert}

\DeclarePairedDelimiter\angles{\langle}{\rangle}
\DeclarePairedDelimiter\set{\{}{\}}

\renewcommand{\th}{\mathrm{th}}

\newcommand{\xra}[1]{\xrightarrow{#1}}
\newcommand{\defeq}{\coloneqq}
\newcommand{\defn}[1]{{\color{NavyBlue} \emph{#1}}}

\newcommand{\rA}{\mathrm{A}}

\newcommand{\rC}{\mathrm{C}}

\newcommand{\rL}{\mathrm{L}}
\newcommand{\rM}{\mathrm{M}}

\newcommand{\rO}{\mathrm{O}}

\newcommand{\rQ}{\mathrm{Q}}

\newcommand{\cC}{\mathcal{C}}

\newcommand{\cK}{\mathcal{K}}
\newcommand{\cL}{\mathcal{L}}

\newcommand{\cO}{\mathcal{O}}
\newcommand{\cP}{\mathcal{P}}

\newcommand{\cT}{\mathcal{T}}

\newcommand{\cZ}{\mathcal{Z}}



%% file: aux/commands.tex
\definecolor{BleuTresFonce}{rgb}{0.215, 0.215, 0.36}

\newcommand{\Gerst}{\mathrm{Gerst}}

\newcommand{\Lie}{\mathrm{Lie}}

\newcommand{\Ass}{\mathrm{Ass}}

\renewcommand{\c}{\mathrm{c}}
\renewcommand{\b}{\mathrm{b}}
\newcommand{\m}{\mathrm{m}}
\newcommand{\Com}{\mathrm{Com}}
\newcommand{\R}{\mathbb{R}}
\newcommand{\Sh}{\mathrm{Sh}}
\newcommand{\ac}{{\scriptstyle \text{\rm !`}}}
\newcommand{\BV}{\mathrm{BV}}

\newcommand{\eBV}{\mathrm{eBV}}
\newcommand{\cBV}{\mathrm{cBV}}

\def\Sy{\mathbb{S}}

\def\KK{\mathbb{K}}

\DeclareMathOperator{\id}{id}
\newcommand{\N}{\mathbb{N}}
\newcommand{\dif}{\partial}
\newcommand{\q}{\mathrm{q}}

\DeclareMathOperator{\End}{End}

\DeclareMathOperator{\Cobar}{\Omega}

\newcommand{\g}{\ensuremath{\mathfrak{g}}}
\newcommand{\Hom}{\ensuremath{\mathrm{Hom}}}
\def\M{\mathcal{M}}
\renewcommand{\P}{\ensuremath{\mathcal{P}}}

\def\C{\mathcal{C}}

\theoremstyle{plain}
\newtheorem{theorem}{Theorem}
\newtheorem{proposition}[theorem]{Proposition}
\newtheorem{corollary}[theorem]{Corollary}
\newtheorem{lemma}[theorem]{Lemma}

\theoremstyle{definition}
\newtheorem{definition}[theorem]{Definition}

\newtheorem{example}[theorem]{Example}

\newtheoremstyle{remarkitalic}
{ }{ }                
{\normalfont}         
{ }                   
{\itshape}            
{.}                   
{ }                   
{}                   
\theoremstyle{remarkitalic}
\newtheorem*{remark}{Remark}

%% file: Tikz_Graphs.tex
\def\BBoxR#1#2{
	\vcenter{\hbox{\begin{tikzpicture}[scale=0.20]
	\draw[thick] (0,2) --(2,0) -- (4,2);
	\draw[thick] (0,2) --(2,0) -- (2, -2);	
	\draw[thick] (4,2)--(4,4);
	\draw[thick] (4,6) -- (4,4);	
	\draw (0,2) node[above] {$\scriptstyle{#1}$};
	\draw (4,6) node[above] {$\scriptstyle{#2}$};	
	\draw[fill=white] (4,4) circle (10pt);
	\draw[fill=black] (2,0) circle (10pt);						
	\end{tikzpicture}}}}

\def\CDR#1#2{
	\vcenter{\hbox{\begin{tikzpicture}[scale=0.20]
	\draw[thick] (0,2) --(2,0) -- (4,2);
	\draw[thick] (0,2) --(2,0) -- (2, -2);	
	\draw[thick] (4,2)--(4,4);
	\draw[thick] (4,6) -- (4,4);	
	\draw (0,2) node[above] {$\scriptstyle{#1}$};
	\draw (4,6) node[above] {$\scriptstyle{#2}$};	
	\draw[fill=black] (4,4) circle (10pt);
	\draw[fill=white] (2,0) circle (10pt);						
	\end{tikzpicture}}}}

\def\BBoxL#1#2{
	\vcenter{\hbox{\begin{tikzpicture}[scale=0.20]
	\draw[thick] (0,2) --(2,0) -- (4,2);
	\draw[thick] (0,2) --(2,0) -- (2, -2);	
	\draw[thick] (0,2)--(0,6);
	\draw (0,6) node[above] {$\scriptstyle{#1}$};
	\draw (4,2) node[above] {$\scriptstyle{#2}$};	
	\draw[fill=white] (0,4) circle (10pt);
	\draw[fill=black] (2,0) circle (10pt);			
	\end{tikzpicture}}}}

\def\BBoxLprime#1#2{
	\vcenter{\hbox{\begin{tikzpicture}[scale=0.20]
	\draw[thick] (0,2) --(2,0) -- (4,2);
	\draw[thick] (0,2) --(2,0) -- (2, -2);	
	\draw[thick] (0,2)--(0,6);
	\draw (0,6) node[above] {$\scriptstyle{#1}$};
	\draw (4,2) node[above] {$\scriptstyle{#2}$};	
	\draw[thin, fill=black] (1.5, 0) -- (2, 0.5) -- (2.5,0) -- (2, -0.5) -- cycle;	
	\draw[thin, fill=white] (-0.5, 4) -- (0, 4.5) -- (0.5,4) -- (0, 3.5) -- cycle;
	\end{tikzpicture}}}}

\def\CDL#1#2{
	\vcenter{\hbox{\begin{tikzpicture}[scale=0.20]
	\draw[thick] (0,2) --(2,0) -- (4,2);
	\draw[thick] (0,2) --(2,0) -- (2, -2);	
	\draw[thick] (0,2)--(0,6);
	\draw (0,6) node[above] {$\scriptstyle{#1}$};
	\draw (4,2) node[above] {$\scriptstyle{#2}$};	
	\draw[fill=black] (0,4) circle (10pt);
	\draw[fill=white] (2,0) circle (10pt);			
	\end{tikzpicture}}}}
	
\def\CDLprime#1#2{
	\vcenter{\hbox{\begin{tikzpicture}[scale=0.20]
	\draw[thick] (0,2) --(2,0) -- (4,2);
	\draw[thick] (0,2) --(2,0) -- (2, -2);	
	\draw[thick] (0,2)--(0,6);
	\draw (0,6) node[above] {$\scriptstyle{#1}$};
	\draw (4,2) node[above] {$\scriptstyle{#2}$};	
	\draw[thin, fill=white] (1.5, 0) -- (2, 0.5) -- (2.5,0) -- (2, -0.5) -- cycle;	
	\draw[thin, fill=black] (-0.5, 4) -- (0, 4.5) -- (0.5,4) -- (0, 3.5) -- cycle;
	\end{tikzpicture}}}}	

\def\DC#1#2{
	\vcenter{\hbox{\begin{tikzpicture}[scale=0.20]
	\draw[thick] (0,2) --(2,0) -- (4,2);
	\draw[thick] (0,2) --(2,0) -- (2, -4);	
	\draw (0,2) node[above] {$\scriptstyle{#1}$};
	\draw (4,2) node[above] {$\scriptstyle{#2}$};	
	\draw[fill=black] (2,-2) circle (10pt);		
	\draw[fill=white] (2,0) circle (10pt);			
	\end{tikzpicture}}}}
	
\def\DCprime#1#2{
	\vcenter{\hbox{\begin{tikzpicture}[scale=0.20]
	\draw[thick] (0,2) --(2,0) -- (4,2);
	\draw[thick] (0,2) --(2,0) -- (2, -4);	
	\draw (0,2) node[above] {$\scriptstyle{#1}$};
	\draw (4,2) node[above] {$\scriptstyle{#2}$};	
	\draw[thin, fill=white] (1.5, 0) -- (2, 0.5) -- (2.5,0) -- (2, -0.5) -- cycle;	
	\draw[thin, fill=black] (1.5, -2) -- (2, -1.5) -- (2.5,-2) -- (2, -2.5) -- cycle;	
	\end{tikzpicture}}}}	

\def\BoxB#1#2{
	\vcenter{\hbox{\begin{tikzpicture}[scale=0.20]
	\draw[thick] (0,2) --(2,0) -- (4,2);
	\draw[thick] (0,2) --(2,0) -- (2, -4);	
	\draw (0,2) node[above] {$\scriptstyle{#1}$};
	\draw (4,2) node[above] {$\scriptstyle{#2}$};	
	\draw[fill=white] (2,-2) circle (10pt);		
	\draw[fill=black] (2,0) circle (10pt);			
	\end{tikzpicture}}}}

\def\BoxBprime#1#2{
	\vcenter{\hbox{\begin{tikzpicture}[scale=0.20]
	\draw[thick] (0,2) --(2,0) -- (4,2);
	\draw[thick] (0,2) --(2,0) -- (2, -4);	
	\draw (0,2) node[above] {$\scriptstyle{#1}$};
	\draw (4,2) node[above] {$\scriptstyle{#2}$};	
	\draw[thin, fill=black] (1.5, 0) -- (2, 0.5) -- (2.5,0) -- (2, -0.5) -- cycle;	
	\draw[thin, fill=white] (1.5, -2) -- (2, -1.5) -- (2.5,-2) -- (2, -2.5) -- cycle;	
	\end{tikzpicture}}}}

\def\CB#1#2#3{
	\vcenter{\hbox{\begin{tikzpicture}[scale=0.2]
	\draw[thick] (0,6) -- (4,2) -- (4,0);
	\draw[thick] (4,6) -- (2,4)--(4,2)--(6,4);
	\draw (0,6) node[above] {$\scriptstyle{#1}$};
	\draw (4,6) node[above] {$\scriptstyle{#2}$};	
	\draw (6,4) node[above] {$\scriptstyle{#3}$};				
	\draw[fill=white] (4,2) circle (10pt);	
	\draw[fill=black] (2,4) circle (10pt);	
	\end{tikzpicture}}}}

\def\CBprime#1#2#3{
	\vcenter{\hbox{\begin{tikzpicture}[scale=0.2]
	\draw[thick] (0,6) -- (4,2) -- (4,0);
	\draw[thick] (4,6) -- (2,4)--(4,2)--(6,4);
	\draw (0,6) node[above] {$\scriptstyle{#1}$};
	\draw (4,6) node[above] {$\scriptstyle{#2}$};	
	\draw (6,4) node[above] {$\scriptstyle{#3}$};				
	\draw[thin, fill=white] (3.5, 2) -- (4, 2.5) -- (4.5,2) -- (4, 1.5) -- cycle;		
	\draw[thin, fill=black] (1.5, 4) -- (2, 4.5) -- (2.5,4) -- (2, 3.5) -- cycle;			
	\end{tikzpicture}}}}

\def\BC#1#2#3{
	\vcenter{\hbox{\begin{tikzpicture}[scale=0.2]
	\draw[thick] (0,6) -- (4,2) -- (4,0);
	\draw[thick] (4,6) -- (2,4)--(4,2)--(6,4);
	\draw (0,6) node[above] {$\scriptstyle{#1}$};
	\draw (4,6) node[above] {$\scriptstyle{#2}$};	
	\draw (6,4) node[above] {$\scriptstyle{#3}$};				
	\draw[fill=black] (4,2) circle (10pt);	
	\draw[fill=white] (2,4) circle (10pt);	
	\end{tikzpicture}}}}

\def\BCprime#1#2#3{
	\vcenter{\hbox{\begin{tikzpicture}[scale=0.2]
	\draw[thick] (0,6) -- (4,2) -- (4,0);
	\draw[thick] (4,6) -- (2,4)--(4,2)--(6,4);
	\draw (0,6) node[above] {$\scriptstyle{#1}$};
	\draw (4,6) node[above] {$\scriptstyle{#2}$};	
	\draw (6,4) node[above] {$\scriptstyle{#3}$};				
	\draw[thin, fill=black] (3.5, 2) -- (4, 2.5) -- (4.5,2) -- (4, 1.5) -- cycle;		
	\draw[thin, fill=white] (1.5, 4) -- (2, 4.5) -- (2.5,4) -- (2, 3.5) -- cycle;			
	\end{tikzpicture}}}}

\def\BoxD{
	\vcenter{\hbox{\begin{tikzpicture}[scale=0.2]
	\draw[thick] (2,4) -- (2,6);
	\draw[thick] (2,2) -- (2,4);
	\draw[thick] (2,2) -- (2,0);
	\draw[fill=black] (2,4) circle (10pt);	
	\draw[fill=white] (2,2) circle (10pt);	
	\end{tikzpicture}}}}

\def\BoxDprime{
	\vcenter{\hbox{\begin{tikzpicture}[scale=0.2]
	\draw[thick] (2,4) -- (2,6);
	\draw[thick] (2,2) -- (2,4);
	\draw[thick] (2,2) -- (2,0);
	\draw[thin, fill=white] (1.5, 2) -- (2, 2.5) -- (2.5,2) -- (2, 1.5) -- cycle;	
	\draw[thin, fill=black] (1.5, 4) -- (2, 4.5) -- (2.5,4) -- (2, 3.5) -- cycle;	
	\end{tikzpicture}}}}

\def\DBox{
	\vcenter{\hbox{\begin{tikzpicture}[scale=0.2]
	\draw[thick] (2,4) -- (2,6);
	\draw[thick] (2,2) -- (2,4);
	\draw[thick] (2,2) -- (2,0);
	\draw[fill=white] (2,4) circle (10pt);	
	\draw[fill=black] (2,2) circle (10pt);	
	\end{tikzpicture}}}}

\def\DBoxprime{
	\vcenter{\hbox{\begin{tikzpicture}[scale=0.2]
	\draw[thick] (2,4) -- (2,6);
	\draw[thick] (2,2) -- (2,4);
	\draw[thick] (2,2) -- (2,0);
	\draw[thin, fill=black] (1.5, 2) -- (2, 2.5) -- (2.5,2) -- (2, 1.5) -- cycle;	
	\draw[thin, fill=white] (1.5, 4) -- (2, 4.5) -- (2.5,4) -- (2, 3.5) -- cycle;	
	\end{tikzpicture}}}}

\def\RMBox#1#2#3{
	\vcenter{\hbox{\begin{tikzpicture}[scale=0.2]
	\draw[thick] (6,6) -- (2,2) -- (2,0);
	\draw[thick] (2,6) -- (4,4)--(2,2)--(0,4);
	\draw (2,6) node[above] {$\scriptstyle{#2}$};
	\draw (6,6) node[above] {$\scriptstyle{#3}$};	
	\draw (0,4) node[above] {$\scriptstyle{#1}$};	
	\draw[fill=white] (4,4) circle (10pt);					
	\end{tikzpicture}}}}

\def\LMBox#1#2#3{
	\vcenter{\hbox{\begin{tikzpicture}[scale=0.2]
	\draw[thick] (0,6) -- (4,2) -- (4,0);
	\draw[thick] (4,6) -- (2,4)--(4,2)--(6,4);
	\draw (0,6) node[above] {$\scriptstyle{#1}$};
	\draw (4,6) node[above] {$\scriptstyle{#2}$};	
	\draw (6,4) node[above] {$\scriptstyle{#3}$};				
	\draw[fill=white] (2,4) circle (10pt);		
	\end{tikzpicture}}}}

\def\LMBoxprime#1#2#3{
	\vcenter{\hbox{\begin{tikzpicture}[scale=0.2]
	\draw[thick] (0,6) -- (4,2) -- (4,0);
	\draw[thick] (4,6) -- (2,4)--(4,2)--(6,4);
	\draw (0,6) node[above] {$\scriptstyle{#1}$};
	\draw (4,6) node[above] {$\scriptstyle{#2}$};	
	\draw (6,4) node[above] {$\scriptstyle{#3}$};	
	\draw[thin] (3.5, 2) -- (4, 2.5) -- (4.5,2) -- (4, 1.5) -- cycle;					
	\draw[thin, fill=white] (1.5, 4) -- (2, 4.5) -- (2.5,4) -- (2, 3.5) -- cycle;												
	\end{tikzpicture}}}}

\def\LCCprime#1#2#3{
	\vcenter{\hbox{\begin{tikzpicture}[scale=0.2]
	\draw[thick] (0,6) -- (4,2) -- (4,0);
	\draw[thick] (4,6) -- (2,4)--(4,2)--(6,4);
	\draw (0,6) node[above] {$\scriptstyle{#1}$};
	\draw (4,6) node[above] {$\scriptstyle{#2}$};	
	\draw (6,4) node[above] {$\scriptstyle{#3}$};	
	\draw[thin, fill=white] (3.5, 2) -- (4, 2.5) -- (4.5,2) -- (4, 1.5) -- cycle;					
	\draw[thin, fill=white] (1.5, 4) -- (2, 4.5) -- (2.5,4) -- (2, 3.5) -- cycle;												
	\end{tikzpicture}}}}

\def\LBoxM#1#2#3{
	\vcenter{\hbox{\begin{tikzpicture}[scale=0.2]
	\draw[thick] (0,6) -- (4,2) -- (4,0);
	\draw[thick] (4,6) -- (2,4)--(4,2)--(6,4);
	\draw (0,6) node[above] {$\scriptstyle{#1}$};
	\draw (4,6) node[above] {$\scriptstyle{#2}$};	
	\draw (6,4) node[above] {$\scriptstyle{#3}$};				
	\draw[fill=white] (4,2) circle (10pt);
	\end{tikzpicture}}}}

\def\LBoxMprime#1#2#3{
	\vcenter{\hbox{\begin{tikzpicture}[scale=0.2]
	\draw[thick] (0,6) -- (4,2) -- (4,0);
	\draw[thick] (4,6) -- (2,4)--(4,2)--(6,4);
	\draw (0,6) node[above] {$\scriptstyle{#1}$};
	\draw (4,6) node[above] {$\scriptstyle{#2}$};	
	\draw (6,4) node[above] {$\scriptstyle{#3}$};				
	\draw[thin, fill=white] (3.5, 2) -- (4, 2.5) -- (4.5,2) -- (4, 1.5) -- cycle;					
	\draw[thin] (1.5, 4) -- (2, 4.5) -- (2.5,4) -- (2, 3.5) -- cycle;												
	\end{tikzpicture}}}}
	
\def\MBoxR#1#2{
	\vcenter{\hbox{\begin{tikzpicture}[scale=0.20]
	\draw[thick] (0,2) --(2,0) -- (4,2);
	\draw[thick] (0,2) --(2,0) -- (2, -2);	
	\draw[thick] (4,2)--(4,4);
	\draw[thick] (4,6) -- (4,4);	
	\draw (0,2) node[above] {$\scriptstyle{#1}$};
	\draw (4,6) node[above] {$\scriptstyle{#2}$};	
	\draw[fill=white] (4,4) circle (10pt);		
	\end{tikzpicture}}}}

\def\MBoxBox#1#2{
	\vcenter{\hbox{\begin{tikzpicture}[scale=0.20]
	\draw[thick] (0,2) --(2,0) -- (4,2);
	\draw[thick] (0,2) --(2,0) -- (2, -2);	
	\draw[thick] (4,2)--(4,6);
	\draw[thick] (0,2)--(0,6);	
	\draw (0,6) node[above] {$\scriptstyle{#1}$};
	\draw (4,6) node[above] {$\scriptstyle{#2}$};	
	\draw[fill=white] (4,4) circle (10pt);	
	\draw[fill=white] (0,4) circle (10pt);			
	\end{tikzpicture}}}}

\def\CBoxR#1#2{
	\vcenter{\hbox{\begin{tikzpicture}[scale=0.20]
	\draw[thick] (0,2) --(2,0) -- (4,2);
	\draw[thick] (0,2) --(2,0) -- (2, -2);	
	\draw[thick] (4,2)--(4,4);
	\draw[thick] (4,6) -- (4,4);	
	\draw (0,2) node[above] {$\scriptstyle{#1}$};
	\draw (4,6) node[above] {$\scriptstyle{#2}$};	
	\draw[fill=white] (4,4) circle (10pt);	
	\draw[fill=white] (2,0) circle (10pt);			
	\end{tikzpicture}}}}

\def\MBoxBoxR#1#2{
	\vcenter{\hbox{\begin{tikzpicture}[scale=0.20]
	\draw[thick] (0,2) --(2,0) -- (4,2);
	\draw[thick] (0,2) --(2,0) -- (2, -2);	
	\draw[thick] (4,2)--(4,4);
	\draw[thick] (4,8) -- (4,4);	
	\draw (0,2) node[above] {$\scriptstyle{#1}$};
	\draw (4,8) node[above] {$\scriptstyle{#2}$};	
	\draw[fill=white] (4,4) circle (10pt);	
	\draw[fill=white] (4,6) circle (10pt);				
	\end{tikzpicture}}}}

\def\MBoxBoxL#1#2{
	\vcenter{\hbox{\begin{tikzpicture}[scale=0.20]
	\draw[thick] (0,2) --(2,0) -- (4,2);
	\draw[thick] (0,2) --(2,0) -- (2, -2);	
	\draw[thick] (0,2)--(0,8);
	\draw (0,8) node[above] {$\scriptstyle{#1}$};	
	\draw (4,2) node[above] {$\scriptstyle{#2}$};
	\draw[fill=white] (0,4) circle (10pt);	
	\draw[fill=white] (0,6) circle (10pt);				
	\end{tikzpicture}}}}

\def\MBoxL#1#2{
	\vcenter{\hbox{\begin{tikzpicture}[scale=0.20]
	\draw[thick] (0,2) --(2,0) -- (4,2);
	\draw[thick] (0,2) --(2,0) -- (2, -2);	
	\draw[thick] (0,2)--(0,6);
	\draw (0,6) node[above] {$\scriptstyle{#1}$};
	\draw (4,2) node[above] {$\scriptstyle{#2}$};	
	\draw[fill=white] (0,4) circle (10pt);		
	\end{tikzpicture}}}}

\def\CBoxL#1#2{
	\vcenter{\hbox{\begin{tikzpicture}[scale=0.20]
	\draw[thick] (0,2) --(2,0) -- (4,2);
	\draw[thick] (0,2) --(2,0) -- (2, -2);	
	\draw[thick] (0,2)--(0,6);
	\draw (0,6) node[above] {$\scriptstyle{#1}$};
	\draw (4,2) node[above] {$\scriptstyle{#2}$};	
	\draw[fill=white] (0,4) circle (10pt);		
	\draw[fill=white] (2,0) circle (10pt);			
	\end{tikzpicture}}}}

\def\MBoxLprime#1#2{
	\vcenter{\hbox{\begin{tikzpicture}[scale=0.20]
	\draw[thick] (0,2) --(2,0) -- (4,2);
	\draw[thick] (0,2) --(2,0) -- (2, -2);	
	\draw[thick] (0,2)--(0,6);
	\draw (0,6) node[above] {$\scriptstyle{#1}$};
	\draw (4,2) node[above] {$\scriptstyle{#2}$};	
	\draw[thin] (1.5, 0) -- (2, 0.5) -- (2.5,0) -- (2, -0.5) -- cycle;	
	\draw[thin, fill=white] (-0.5, 4) -- (0, 4.5) -- (0.5,4) -- (0, 3.5) -- cycle;	
	\end{tikzpicture}}}}

\def\CBoxLprime#1#2{
	\vcenter{\hbox{\begin{tikzpicture}[scale=0.20]
	\draw[thick] (0,2) --(2,0) -- (4,2);
	\draw[thick] (0,2) --(2,0) -- (2, -2);	
	\draw[thick] (0,2)--(0,6);
	\draw (0,6) node[above] {$\scriptstyle{#1}$};
	\draw (4,2) node[above] {$\scriptstyle{#2}$};	
	\draw[thin, fill=white] (1.5, 0) -- (2, 0.5) -- (2.5,0) -- (2, -0.5) -- cycle;	
	\draw[thin, fill=white] (-0.5, 4) -- (0, 4.5) -- (0.5,4) -- (0, 3.5) -- cycle;	
	\end{tikzpicture}}}}

\def\BoxM#1#2{
	\vcenter{\hbox{\begin{tikzpicture}[scale=0.20]
	\draw[thick] (0,2) --(2,0) -- (4,2);
	\draw[thick] (0,2) --(2,0) -- (2, -4);	
	\draw (0,2) node[above] {$\scriptstyle{#1}$};
	\draw (4,2) node[above] {$\scriptstyle{#2}$};	
	\draw[fill=white] (2,-2) circle (10pt);		
	\end{tikzpicture}}}}

\def\BoxC#1#2{
	\vcenter{\hbox{\begin{tikzpicture}[scale=0.20]
	\draw[thick] (0,2) --(2,0) -- (4,2);
	\draw[thick] (0,2) --(2,0) -- (2, -4);	
	\draw (0,2) node[above] {$\scriptstyle{#1}$};
	\draw (4,2) node[above] {$\scriptstyle{#2}$};	
	\draw[fill=white] (2,-2) circle (10pt);	
	\draw[fill=white] (2,0) circle (10pt);			
	\end{tikzpicture}}}}

\def\BoxBoxM#1#2{
	\vcenter{\hbox{\begin{tikzpicture}[scale=0.20]
	\draw[thick] (0,2) --(2,0) -- (4,2);
	\draw[thick] (0,2) --(2,0) -- (2, -6);	
	\draw (0,2) node[above] {$\scriptstyle{#1}$};
	\draw (4,2) node[above] {$\scriptstyle{#2}$};	
	\draw[fill=white] (2,-2) circle (10pt);	
	\draw[fill=white] (2,-4) circle (10pt);	
	\end{tikzpicture}}}}

\def\BoxMBoxL#1#2{
	\vcenter{\hbox{\begin{tikzpicture}[scale=0.20]
	\draw[thick] (0,2) --(2,0) -- (4,2);
	\draw[thick] (0,2) --(2,0) -- (2, -4);	
	\draw[thick] (0,2)--(0,6);	
	\draw (0,6) node[above] {$\scriptstyle{#1}$};
	\draw (4,2) node[above] {$\scriptstyle{#2}$};	
	\draw[fill=white] (2,-2) circle (10pt);		
	\draw[fill=white] (0,4) circle (10pt);			
	\end{tikzpicture}}}}

\def\BoxMBoxR#1#2{
	\vcenter{\hbox{\begin{tikzpicture}[scale=0.20]
	\draw[thick] (0,2) --(2,0) -- (4,2);
	\draw[thick] (0,2) --(2,0) -- (2, -4);	
	\draw[thick] (4,2)--(4,6);	
	\draw (4,6) node[above] {$\scriptstyle{#2}$};
	\draw (0,2) node[above] {$\scriptstyle{#1}$};	
	\draw[fill=white] (2,-2) circle (10pt);		
	\draw[fill=white] (4,4) circle (10pt);			
	\end{tikzpicture}}}}

\def\BoxMprime#1#2{
	\vcenter{\hbox{\begin{tikzpicture}[scale=0.20]
	\draw[thick] (0,2) --(2,0) -- (4,2);
	\draw[thick] (0,2) --(2,0) -- (2, -4);	
	\draw (0,2) node[above] {$\scriptstyle{#1}$};
	\draw (4,2) node[above] {$\scriptstyle{#2}$};	
	\draw[thin] (1.5, 0) -- (2, 0.5) -- (2.5,0) -- (2, -0.5) -- cycle;	
	\draw[thin, fill=white] (1.5, -2) -- (2, -1.5) -- (2.5,-2) -- (2, -2.5) -- cycle;		
	\end{tikzpicture}}}}

\def\BoxCprime#1#2{
	\vcenter{\hbox{\begin{tikzpicture}[scale=0.20]
	\draw[thick] (0,2) --(2,0) -- (4,2);
	\draw[thick] (0,2) --(2,0) -- (2, -4);	
	\draw (0,2) node[above] {$\scriptstyle{#1}$};
	\draw (4,2) node[above] {$\scriptstyle{#2}$};	
	\draw[thin, fill=white] (1.5, 0) -- (2, 0.5) -- (2.5,0) -- (2, -0.5) -- cycle;	
	\draw[thin, fill=white] (1.5, -2) -- (2, -1.5) -- (2.5,-2) -- (2, -2.5) -- cycle;		
	\end{tikzpicture}}}}

\def\CC#1#2{
	\vcenter{\hbox{\begin{tikzpicture}[scale=0.2]
	\draw[thick] (0,4) -- (2,2);
	\draw[thick] (4,4) -- (2,2);
	\draw[thick] (2,2) -- (2,0);
	\draw[fill=white] (2,2) circle (10pt);
	\draw (0,4) node[above] {$\scriptstyle{#1}$};
	\draw (4,4) node[above] {$\scriptstyle{#2}$};		
	\end{tikzpicture}}}}

\def\BOX{
	\vcenter{\hbox{\begin{tikzpicture}[scale=0.15]
	\draw[thick] (2,2) -- (2,4);
	\draw[thick] (2,2) -- (2,0);
	\draw[fill=white] (2,2) circle (10pt);	
	\end{tikzpicture}}}}

\def\BOXprime{
	\vcenter{\hbox{\begin{tikzpicture}[scale=0.15]
	\draw[thick] (2,2) -- (2,4);
	\draw[thick] (2,2) -- (2,0);
	\draw[thin, fill=white] (1.5, 2) -- (2, 2.5) -- (2.5,2) -- (2, 1.5) -- cycle;	
	\end{tikzpicture}}}}

\def\Cprime{
	\vcenter{\hbox{\begin{tikzpicture}[scale=0.15]
	\draw[thick] (0,4) -- (2,2);
	\draw[thick] (4,4) -- (2,2);
	\draw[thick] (2,2) -- (2,0);
	\draw[thin, fill=white] (1.5, 2) -- (2, 2.5) -- (2.5,2) -- (2, 1.5) -- cycle;	
	\end{tikzpicture}}}}

\def\C{
	\vcenter{\hbox{\begin{tikzpicture}[scale=0.15]
	\draw[thick] (0,4) -- (2,2);
	\draw[thick] (4,4) -- (2,2);
	\draw[thick] (2,2) -- (2,0);
	\draw[fill=white] (2,2) circle (10pt);
	\end{tikzpicture}}}}

\def\BDR#1#2{
	\vcenter{\hbox{\begin{tikzpicture}[scale=0.20]
	\draw[thick] (0,2) --(2,0) -- (4,2);
	\draw[thick] (0,2) --(2,0) -- (2, -2);	
	\draw[thick] (4,2)--(4,4);
	\draw[thick] (4,6) -- (4,4);	
	\draw (0,2) node[above] {$\scriptstyle{#1}$};
	\draw (4,6) node[above] {$\scriptstyle{#2}$};	
	\draw[fill=black] (4,4) circle (10pt);
	\draw[fill=black] (2,0) circle (10pt);						
	\end{tikzpicture}}}}

\def\BDL#1#2{
	\vcenter{\hbox{\begin{tikzpicture}[scale=0.20]
	\draw[thick] (0,2) --(2,0) -- (4,2);
	\draw[thick] (0,2) --(2,0) -- (2, -2);	
	\draw[thick] (0,2)--(0,6);
	\draw (0,6) node[above] {$\scriptstyle{#1}$};
	\draw (4,2) node[above] {$\scriptstyle{#2}$};	
	\draw[fill=black] (0,4) circle (10pt);
	\draw[fill=black] (2,0) circle (10pt);			
	\end{tikzpicture}}}}

\def\BDLprime#1#2{
	\vcenter{\hbox{\begin{tikzpicture}[scale=0.20]
	\draw[thick] (0,2) --(2,0) -- (4,2);
	\draw[thick] (0,2) --(2,0) -- (2, -2);	
	\draw[thick] (0,2)--(0,6);
	\draw (0,6) node[above] {$\scriptstyle{#1}$};
	\draw (4,2) node[above] {$\scriptstyle{#2}$};	
	\draw[thin, fill=black] (1.5, 0) -- (2, 0.5) -- (2.5,0) -- (2, -0.5) -- cycle;	
	\draw[thin, fill=black] (-0.5, 4) -- (0, 4.5) -- (0.5,4) -- (0, 3.5) -- cycle;
	\end{tikzpicture}}}}

\def\DB#1#2{
	\vcenter{\hbox{\begin{tikzpicture}[scale=0.20]
	\draw[thick] (0,2) --(2,0) -- (4,2);
	\draw[thick] (0,2) --(2,0) -- (2, -4);	
	\draw (0,2) node[above] {$\scriptstyle{#1}$};
	\draw (4,2) node[above] {$\scriptstyle{#2}$};	
	\draw[fill=black] (2,-2) circle (10pt);		
	\draw[fill=black] (2,0) circle (10pt);			
	\end{tikzpicture}}}}

\def\DBprime#1#2{
	\vcenter{\hbox{\begin{tikzpicture}[scale=0.20]
	\draw[thick] (0,2) --(2,0) -- (4,2);
	\draw[thick] (0,2) --(2,0) -- (2, -4);	
	\draw (0,2) node[above] {$\scriptstyle{#1}$};
	\draw (4,2) node[above] {$\scriptstyle{#2}$};	
	\draw[thin, fill=black] (1.5, 0) -- (2, 0.5) -- (2.5,0) -- (2, -0.5) -- cycle;	
	\draw[thin, fill=black] (1.5, -2) -- (2, -1.5) -- (2.5,-2) -- (2, -2.5) -- cycle;	
	\end{tikzpicture}}}}

\def\RMB#1#2#3{
	\vcenter{\hbox{\begin{tikzpicture}[scale=0.2]
	\draw[thick] (6,6) -- (2,2) -- (2,0);
	\draw[thick] (2,6) -- (4,4)--(2,2)--(0,4);
	\draw (2,6) node[above] {$\scriptstyle{#2}$};
	\draw (6,6) node[above] {$\scriptstyle{#3}$};	
	\draw (0,4) node[above] {$\scriptstyle{#1}$};	
	\draw[fill=black] (4,4) circle (10pt);					
	\end{tikzpicture}}}}

\def\LMB#1#2#3{
	\vcenter{\hbox{\begin{tikzpicture}[scale=0.2]
	\draw[thick] (0,6) -- (4,2) -- (4,0);
	\draw[thick] (4,6) -- (2,4)--(4,2)--(6,4);
	\draw (0,6) node[above] {$\scriptstyle{#1}$};
	\draw (4,6) node[above] {$\scriptstyle{#2}$};	
	\draw (6,4) node[above] {$\scriptstyle{#3}$};				
	\draw[fill=black] (2,4) circle (10pt);		
	\end{tikzpicture}}}}

\def\LMBprime#1#2#3{
	\vcenter{\hbox{\begin{tikzpicture}[scale=0.2]
	\draw[thick] (0,6) -- (4,2) -- (4,0);
	\draw[thick] (4,6) -- (2,4)--(4,2)--(6,4);
	\draw (0,6) node[above] {$\scriptstyle{#1}$};
	\draw (4,6) node[above] {$\scriptstyle{#2}$};	
	\draw (6,4) node[above] {$\scriptstyle{#3}$};				
	\draw[thin] (3.5, 2) -- (4, 2.5) -- (4.5,2) -- (4, 1.5) -- cycle;					
	\draw[thin, fill=black] (1.5, 4) -- (2, 4.5) -- (2.5,4) -- (2, 3.5) -- cycle;							
	\end{tikzpicture}}}}

\def\LBM#1#2#3{
	\vcenter{\hbox{\begin{tikzpicture}[scale=0.2]
	\draw[thick] (0,6) -- (4,2) -- (4,0);
	\draw[thick] (4,6) -- (2,4)--(4,2)--(6,4);
	\draw (0,6) node[above] {$\scriptstyle{#1}$};
	\draw (4,6) node[above] {$\scriptstyle{#2}$};	
	\draw (6,4) node[above] {$\scriptstyle{#3}$};				
	\draw[fill=black] (4,2) circle (10pt);
	\end{tikzpicture}}}}
	
\def\LMMB#1#2#3{
	\vcenter{\hbox{\begin{tikzpicture}[scale=0.2]
	\draw[thick] (0,10) -- (0,6) -- (4,2) -- (4,0);
	\draw[thick] (4,6) -- (2,4)--(4,2)--(6,4);
	\draw (0,10) node[above] {$\scriptstyle{#1}$};
	\draw (4,6) node[above] {$\scriptstyle{#2}$};	
	\draw (6,4) node[above] {$\scriptstyle{#3}$};				
	\draw[fill=black] (0,8) circle (10pt);
	\end{tikzpicture}}}}

\def\LMMBB#1#2#3{
	\vcenter{\hbox{\begin{tikzpicture}[scale=0.2]
	\draw[thick] (4,10) -- (4,6);
	\draw[thick] ((0,6) -- (4,2) -- (4,0);
	\draw[thick] (4,6) -- (2,4)--(4,2)--(6,4);
	\draw (0,6) node[above] {$\scriptstyle{#1}$};
	\draw (4,10) node[above] {$\scriptstyle{#2}$};	
	\draw (6,4) node[above] {$\scriptstyle{#3}$};				
	\draw[fill=black] (4,8) circle (10pt);
	\end{tikzpicture}}}}

\def\LMMBBB#1#2#3{
	\vcenter{\hbox{\begin{tikzpicture}[scale=0.2]
	\draw[thick] (6,10) -- (6,4);
	\draw[thick] ((0,6) -- (4,2) -- (4,0);
	\draw[thick] (4,6) -- (2,4)--(4,2)--(6,4);
	\draw (0,6) node[above] {$\scriptstyle{#1}$};
	\draw (4,6) node[above] {$\scriptstyle{#2}$};	
	\draw (6,10) node[above] {$\scriptstyle{#3}$};				
	\draw[fill=black] (6,8) circle (10pt);
	\end{tikzpicture}}}}
	
\def\DLMM#1#2#3{
	\vcenter{\hbox{\begin{tikzpicture}[scale=0.2]
	\draw[thick] (0,6) -- (4,2) -- (4,0);
	\draw[thick] (4,6) -- (2,4)--(4,2)--(6,4);
	\draw[thick] (4,2)--(4,-2);	
	\draw (0,6) node[above] {$\scriptstyle{#1}$};
	\draw (4,6) node[above] {$\scriptstyle{#2}$};	
	\draw (6,4) node[above] {$\scriptstyle{#3}$};				
	\draw[fill=black] (4,0) circle (10pt);
	\end{tikzpicture}}}}

\def\LMDM#1#2#3{
	\vcenter{\hbox{\begin{tikzpicture}[scale=0.2]
	\draw[thick] (4,0) -- (4,2)  -- (2,4) --(2,8) -- (0,10) ;
	\draw[thick] (4,2)  -- (6,4);
	\draw[thick] (2,8) -- (4,10) ;
	\draw (0,10) node[above] {$\scriptstyle{#1}$};
	\draw (4,10) node[above] {$\scriptstyle{#2}$};	
	\draw (6,4) node[above] {$\scriptstyle{#3}$};				
	\draw[fill=black] (2,6) circle (10pt);
	\end{tikzpicture}}}}
	
\def\LBMprime#1#2#3{
	\vcenter{\hbox{\begin{tikzpicture}[scale=0.2]
	\draw[thick] (0,6) -- (4,2) -- (4,0);
	\draw[thick] (4,6) -- (2,4)--(4,2)--(6,4);
	\draw (0,6) node[above] {$\scriptstyle{#1}$};
	\draw (4,6) node[above] {$\scriptstyle{#2}$};	
	\draw (6,4) node[above] {$\scriptstyle{#3}$};				
	\draw[thin, fill=black] (3.5, 2) -- (4, 2.5) -- (4.5,2) -- (4, 1.5) -- cycle;					
	\draw[thin] (1.5, 4) -- (2, 4.5) -- (2.5,4) -- (2, 3.5) -- cycle;							
	\end{tikzpicture}}}}	

\def\BBB#1#2#3{
	\vcenter{\hbox{\begin{tikzpicture}[scale=0.2]
	\draw[thick] (0,6) -- (4,2) -- (4,0);
	\draw[thick] (4,6) -- (2,4)--(4,2)--(6,4);
	\draw (0,6) node[above] {$\scriptstyle{#1}$};
	\draw (4,6) node[above] {$\scriptstyle{#2}$};	
	\draw (6,4) node[above] {$\scriptstyle{#3}$};				
	\draw[fill=black] (4,2) circle (10pt);	
	\draw[fill=black] (2,4) circle (10pt);	
	\end{tikzpicture}}}}
	
\def\BBBprime#1#2#3{
	\vcenter{\hbox{\begin{tikzpicture}[scale=0.2]
	\draw[thick] (0,6) -- (4,2) -- (4,0);
	\draw[thick] (4,6) -- (2,4)--(4,2)--(6,4);
	\draw (0,6) node[above] {$\scriptstyle{#1}$};
	\draw (4,6) node[above] {$\scriptstyle{#2}$};	
	\draw (6,4) node[above] {$\scriptstyle{#3}$};				
	\draw[thin, fill=black] (3.5, 2) -- (4, 2.5) -- (4.5,2) -- (4, 1.5) -- cycle;		
	\draw[thin, fill=black] (1.5, 4) -- (2, 4.5) -- (2.5,4) -- (2, 3.5) -- cycle;			
	\end{tikzpicture}}}}	

\def\MDR#1#2{
	\vcenter{\hbox{\begin{tikzpicture}[scale=0.20]
	\draw[thick] (0,2) --(2,0) -- (4,2);
	\draw[thick] (0,2) --(2,0) -- (2, -2);	
	\draw[thick] (4,2)--(4,4);
	\draw[thick] (4,6) -- (4,4);	
	\draw (0,2) node[above] {$\scriptstyle{#1}$};
	\draw (4,6) node[above] {$\scriptstyle{#2}$};	
	\draw[fill=black] (4,4) circle (10pt);		
	\end{tikzpicture}}}}

\def\MDL#1#2{
	\vcenter{\hbox{\begin{tikzpicture}[scale=0.20]
	\draw[thick] (0,2) --(2,0) -- (4,2);
	\draw[thick] (0,2) --(2,0) -- (2, -2);	
	\draw[thick] (0,2)--(0,6);
	\draw (0,6) node[above] {$\scriptstyle{#1}$};
	\draw (4,2) node[above] {$\scriptstyle{#2}$};	
	\draw[fill=black] (0,4) circle (10pt);		
	\end{tikzpicture}}}}

\def\MDLprime#1#2{
	\vcenter{\hbox{\begin{tikzpicture}[scale=0.20]
	\draw[thick] (0,2) --(2,0) -- (4,2);
	\draw[thick] (0,2) --(2,0) -- (2, -2);	
	\draw[thick] (0,2)--(0,6);
	\draw (0,6) node[above] {$\scriptstyle{#1}$};
	\draw (4,2) node[above] {$\scriptstyle{#2}$};	
	\draw[thin] (1.5, 0) -- (2, 0.5) -- (2.5,0) -- (2, -0.5) -- cycle;	
	\draw[thin, fill=black] (-0.5, 4) -- (0, 4.5) -- (0.5,4) -- (0, 3.5) -- cycle;
	\end{tikzpicture}}}}

\def\DM#1#2{
	\vcenter{\hbox{\begin{tikzpicture}[scale=0.20]
	\draw[thick] (0,2) --(2,0) -- (4,2);
	\draw[thick] (0,2) --(2,0) -- (2, -4);	
	\draw (0,2) node[above] {$\scriptstyle{#1}$};
	\draw (4,2) node[above] {$\scriptstyle{#2}$};	
	\draw[fill=black] (2,-2) circle (10pt);		
	\end{tikzpicture}}}}

\def\DMprime#1#2{
	\vcenter{\hbox{\begin{tikzpicture}[scale=0.20]
	\draw[thick] (0,2) --(2,0) -- (4,2);
	\draw[thick] (0,2) --(2,0) -- (2, -4);	
	\draw (0,2) node[above] {$\scriptstyle{#1}$};
	\draw (4,2) node[above] {$\scriptstyle{#2}$};	
	\draw[thin] (1.5, 0) -- (2, 0.5) -- (2.5,0) -- (2, -0.5) -- cycle;	
	\draw[thin, fill=black] (1.5, -2) -- (2, -1.5) -- (2.5,-2) -- (2, -2.5) -- cycle;	
	\end{tikzpicture}}}}

\def\BB#1#2{
	\vcenter{\hbox{\begin{tikzpicture}[scale=0.2]
	\draw[thick] (0,4) -- (2,2);
	\draw[thick] (4,4) -- (2,2);
	\draw[thick] (2,2) -- (2,0);
	\draw[fill=black] (2,2) circle (10pt);
	\draw (0,4) node[above] {$\scriptstyle{#1}$};
	\draw (4,4) node[above] {$\scriptstyle{#2}$};		
	\end{tikzpicture}}}}

\def\DD{
	\vcenter{\hbox{\begin{tikzpicture}[scale=0.2]
	\draw[thick] (2,4) -- (2,6);
	\draw[thick] (2,2) -- (2,4);
	\draw[thick] (2,2) -- (2,0);
	\draw[fill=black] (2,4) circle (10pt);	
	\draw[fill=black] (2,2) circle (10pt);	
	\end{tikzpicture}}}}
	
\def\BoxBoxprime{
	\vcenter{\hbox{\begin{tikzpicture}[scale=0.2]
	\draw[thick] (2,4) -- (2,6);
	\draw[thick] (2,2) -- (2,4);
	\draw[thick] (2,2) -- (2,0);
	\draw[thin, fill=white] (1.5, 2) -- (2, 2.5) -- (2.5,2) -- (2, 1.5) -- cycle;	
	\draw[thin, fill=white] (1.5, 4) -- (2, 4.5) -- (2.5,4) -- (2, 3.5) -- cycle;	
	\end{tikzpicture}}}}	

\def\RR#1#2#3{
	\vcenter{\hbox{\begin{tikzpicture}[scale=0.2]
	\draw[thick] (6,6) -- (2,2) -- (2,0);
	\draw[thick] (2,6) -- (4,4)--(2,2)--(0,4);
	\draw (2,6) node[above] {$\scriptstyle{#2}$};
	\draw (6,6) node[above] {$\scriptstyle{#3}$};	
	\draw (0,4) node[above] {$\scriptstyle{#1}$};				
	\end{tikzpicture}}}}

\def\BBRprime#1#2#3{
	\vcenter{\hbox{\begin{tikzpicture}[scale=0.2]
	\draw[thick] (6,6) -- (2,2) -- (2,0);
	\draw[thick] (2,6) -- (4,4)--(2,2)--(0,4);
	\draw (2,6) node[above] {$\scriptstyle{#2}$};
	\draw (6,6) node[above] {$\scriptstyle{#3}$};	
	\draw (0,4) node[above] {$\scriptstyle{#1}$};	
	\draw[thin, fill=black] (3.5, 4) -- (4, 4.5) -- (4.5,4) -- (4, 3.5) -- cycle;		
	\draw[thin, fill=black] (1.5, 2) -- (2, 2.5) -- (2.5,2) -- (2, 1.5) -- cycle;							
	\end{tikzpicture}}}}

\def\LL#1#2#3{
	\vcenter{\hbox{\begin{tikzpicture}[scale=0.2]
	\draw[thick] (0,6) -- (4,2) -- (4,0);
	\draw[thick] (4,6) -- (2,4)--(4,2)--(6,4);
	\draw (0,6) node[above] {$\scriptstyle{#1}$};
	\draw (4,6) node[above] {$\scriptstyle{#2}$};	
	\draw (6,4) node[above] {$\scriptstyle{#3}$};				
	\end{tikzpicture}}}}

\def\LLprime#1#2#3{
	\vcenter{\hbox{\begin{tikzpicture}[scale=0.2]
	\draw[thick] (0,6) -- (4,2) -- (4,0);
	\draw[thick] (4,6) -- (2,4)--(4,2)--(6,4);
	\draw (0,6) node[above] {$\scriptstyle{#1}$};
	\draw (4,6) node[above] {$\scriptstyle{#2}$};	
	\draw (6,4) node[above] {$\scriptstyle{#3}$};	
	\draw[thin] (3.5, 2) -- (4, 2.5) -- (4.5,2) -- (4, 1.5) -- cycle;					
	\draw[thin] (1.5, 4) -- (2, 4.5) -- (2.5,4) -- (2, 3.5) -- cycle;						
	\end{tikzpicture}}}}
	
\def\Dprime{
	\vcenter{\hbox{\begin{tikzpicture}[scale=0.15]
	\draw[thick] (2,2) -- (2,4);
	\draw[thick] (2,2) -- (2,0);
	\draw[thin, fill=black] (1.5, 2) -- (2, 2.5) -- (2.5,2) -- (2, 1.5) -- cycle;		
	\end{tikzpicture}}}}

\def\D{
	\vcenter{\hbox{\begin{tikzpicture}[scale=0.15]
	\draw[thick] (2,2) -- (2,4);
	\draw[thick] (2,2) -- (2,0);
	\draw[fill=black] (2,2) circle (10pt);	
	\end{tikzpicture}}}}

\def\Bprime{
	\vcenter{\hbox{\begin{tikzpicture}[scale=0.15]
	\draw[thick] (0,4) -- (2,2);
	\draw[thick] (4,4) -- (2,2);
	\draw[thick] (2,2) -- (2,0);
	\draw[thin, fill=black] (1.5, 2) -- (2, 2.5) -- (2.5,2) -- (2, 1.5) -- cycle;		
	\end{tikzpicture}}}}

\def\B{
	\vcenter{\hbox{\begin{tikzpicture}[scale=0.15]
	\draw[thick] (0,4) -- (2,2);
	\draw[thick] (4,4) -- (2,2);
	\draw[thick] (2,2) -- (2,0);
	\draw[fill=black] (2,2) circle (10pt);
	\end{tikzpicture}}}}

\def\Mprime{
	\vcenter{\hbox{\begin{tikzpicture}[scale=0.15]
	\draw[thick] (0,4) -- (2,2);
	\draw[thick] (4,4) -- (2,2);
	\draw[thick] (2,2) -- (2,0);
	\draw[thin] (1.5, 2) -- (2, 2.5) -- (2.5,2) -- (2, 1.5) -- cycle;	
	\end{tikzpicture}}}}	

\def\M{
	\vcenter{\hbox{\begin{tikzpicture}[scale=0.15]
	\draw[thick] (0,4) -- (2,2);
	\draw[thick] (4,4) -- (2,2);
	\draw[thick] (2,2) -- (2,0);
	\end{tikzpicture}}}}	
	

%% file: sec/abstract.tex
\begin{abstract}
	The search for algebraic foundations of colour-kinematics duality and the double-copy construction has brought into focus a generalization of Batalin--Vilkovisky algebras, referred to here as coexact BV-algebras and as \(\textrm{BV}^\square\)-algebras in other sources.
	While these structures capture the cubic sector, they fail to encode higher-valence phenomena, for which a homotopy-theoretic extension becomes necessary.
	This work introduces a conceptual notion of homotopy coexact BV-algebra, defined through the homotopical interplay of commutative and BV structures, and provides a concrete model in terms of generators and relations, obtained through an extension the theory of Koszul duality for operads. The resulting framework enables the systematic use of homotopical tools---\(\infty\)-morphisms, homotopy transfer, rectification, and deformation theory---and naturally accommodates the quartic-level structures recently identified in Yang--Mills theory.
\end{abstract}

%% file: sec/introduction.tex

\newpage
\section{Introduction}

This work is motivated by gauge field theory, particularly by \textit{colour--kinematics duality} and the \textit{double copy construction}.
From the algebro-homotopical perspective on perturbative field theory, surveyed for example in \cite{Hohm2017LInfinityAlgebrasGauge, Jurco2020PerturbativeQFTandHomotopyAlgebras, Jurco2019LInfinityBVReview}, a classical field theory is encoded as a cyclic \(\rL_\infty\)-algebra \(\big(\cL, \set{\ell_n}, \langle\cdot, \cdot \rangle \big)\)---a model for homotopy Lie algebras with a compatible non-degenerate inner product---whose Maurer--Cartan functional
\[
S(a) \defeq \sum_{n \geqslant 1} \frac{1}{(n+1)!}\,\big\langle a, \ell_n(a^{\otimes n})\big\rangle
\]
defines the action.
For many gauge theories of interest, including Chern--Simons and Yang--Mills, their \(\rL_\infty\)-algebra splits as a tensor product
\[
\cL \cong \cK \ot \g,
\]
where \(\g\) is a Lie algebra, the \textit{colour Lie algebra}, and \(\cK\) is a homotopy commutative algebra, the \textit{kinematic algebra}.
This \emph{colour-splitting} provides the algebraic framework to understand the duality between colour and kinematics, a phenomenon first observed by Bern--Carrasco--Johansson while studying amplitudes in Yang--Mills theory \cite{Bern2008NewRelationsGaugeTheory}.
There it was found that kinematic numerators obey the same antisymmetry and Jacobi relations as the colour factors, suggesting that the kinematic algebra \(\cK\) carries a Lie-type structure \cite{Monteiro2011KinematicAlgebraSelfDual}.

For instance, in Chern--Simons theory one has \(\cK = (\Omega(M), d, \wedge)\), the algebra of differential forms on an oriented 3-manifold.
Equipping \(M\) with a metric introduces the Hodge codifferential \(d^\star\), and the commutator \([d^\star, \wedge]\) endows \(\cK\) with a shifted Lie bracket giving rise to the kinematic Lie algebra of Chern--Simons theory \cite{BenShahar2022OffShellCKDualityCS}.

In this example, \(\cL\) is a dg Lie algebra, i.e., an \(\rL_\infty\)-algebra in which all operations beyond arity 2 vanish, reflecting the fact that Chern--Simons theory involves only cubic interactions.
The algebraic structure on the kinematic algebra \((\cK, d, \wedge, d^\star)\) is itself strict and motivates the notion of \emph{\(\cBV\)-algebra}, a weaker version of dg Batanin--Vilkovisky algebra \cite{Lian1993NewPerspectivesBRST, Getzler1994BVAlgebrasTQFT}, or \(\BV\)-algebra for short.
More precisely, a \(\cBV\)-algebra is a tuple \((A, d_A, m, \triangle)\) where \((A, d_A, m)\) is a dg commutative algebra and \(\triangle\) is a nilpotent operator of degree \(|\triangle| = -|d_A|\) and order at most 2.

\(\BV\)-algebras are special cases of \(\cBV\)-algebras, realized precisely when the \emph{obstruction} \([d_A, \triangle]\) vanishes.
In this case, the shifted Lie bracket \(\b = [\triangle, \m]\) makes \((A, d_A, \m, \b)\) into a dg Gerstenhaber algebra.
For differential forms, the obstruction \([d, d^\star]\) is precisely the failure of the de Rham differential and the Hodge codifferential to commute, yielding the Laplacian in Riemannian signature and the d'Alembertian in Lorentzian.

The concept of \(\cBV\)-algebra, or \textit{coexact \(\BV\)-algebra}, is dual---in a sense explained in \cref{ss:exact}---to the notion of \textit{exact \(\BV\)-algebra} arising in Poisson geometry \cite{DotsenkoShadrinVallette15, GuanMuro23}.
In the physics literature, \(\cBV\)-algebras are referred to as \(\BV^\square\)-algebras, a term coined by M.~Reiterer in his influential preprint \cite{Reiterer2020HomotopyBVYMCK}.
These algebras have been employed to define colour--kinematics duality for theories with at most cubic interactions and to construct double copies of these \cite{Borsten2021DoubleCopyHomotopyAlgebras, Borsten2023KinematicLieAlgebrasTwistor, Borsten2023DoubleCopySDYM, Borsten2023TreeLevelCKPureSpinor, Bonezzi2024DoubleCopy3DCSKodairaSpencer, Borsten2025DoubleCopyFromTensorBV, BenShahar2025OffShellDoubleCopyBV}.
To move beyond the cubic case and incorporate higher-order interactions, as for instance up to quartic order in \cite{Bonezzi2022GaugeStructureDoubleField, Bonezzi2023GaugeInvariantDoubleCopyQuartic, Bonezzi2024WeaklyConstrainedDoubleField, Bonezzi2023GaugeIndependentKinematicSDYM}, M.~Reiterer also pioneered the use of a homotopy version of \(\BV^\square\)-algebras, which he introduced through an ingenious albeit ad hoc definition.
The goal of the present paper is to provide a conceptual definition of homotopy \(\cBV\)-algebras and a concrete model for them, which we ground in the operadic calculus of \cite{LodayVallette12}.
As reviewed in \cref{ss:tools}, this enables the systematic application of key homotopical techniques to the resulting \(\cBV_{\!\infty}\)-algebras, including their homotopy transfer, rectification, \(\infty\)-morphisms, and deformation theory.

To facilitate the use of \(\cBV_{\!\infty}\)-algebras and the above techniques in physics, we give explicit descriptions of the generating operations and relations defining \(\cBV_{\!\infty}\)-algebras, relate them to \(\rC_\infty\)- and \(\BV_{\!\infty}\)-algebras, and show explicitly how the quartic-level structure constructed in \cite{Bonezzi2023GaugeInvariantDoubleCopyQuartic, Bonezzi2024WeaklyConstrainedDoubleField} on the kinematic algebra of Yang--Mills theory fits naturally into our framework.

\medskip

Conceptually, it is useful to situate \(\cBV_{\!\infty}\)-algebras in a more abstract context.
Within the relevant model category \cite{Hinich97}, the operad \(\cBV_{\!\infty}\) cofibrantly resolves the operad \(\cBV\), with this resolution fitting in the following commutative diagram, where the horizontal compositions are the canonical inclusions:
\begin{equation}\label{eq:diagram_main}
	\begin{tikzcd}[row sep=large, column sep=large]
		\Com_\infty \dar[->>, "\sim"] \arrow[r, >->, "\sim"] &
		\cBV_{\!\infty} \dar[->>, "\sim"] \arrow[r,->>] &
		\BV_{\!\infty} \dar[->>, "\sim"] \\
		\Com \arrow[r, "\sim"] &
		\cBV \arrow[r,->>] &
		\BV
	\end{tikzcd}
\end{equation}
Here \(\rightarrowtail\) denotes a cofibration, \(\twoheadrightarrow\) a fibration, and \(\xra{\sim}\) a weak equivalence.
The operads
\begin{equation*}\label{eq:cofibrant_operads_intro}
	\Com_\infty \defeq \Cobar\Com^\ac,
	\qquad
	\BV_{\!\infty} \defeq \Cobar\BV^\ac,
	\qquad
	\cBV_{\!\infty} \defeq \Cobar\cBV^\ac,
\end{equation*}
arise from applying the cobar construction to the Koszul dual cooperad of the operad governing the corresponding strict algebras.

\medskip\noindent
Two model-independent consequences of Diagram \eqref{eq:diagram_main} are:
\begin{enumerate}
	\item Any homotopy commutative algebra extends, up to homotopy, to a homotopy \(\cBV\)-algebra.
	\item A homotopy \(\cBV\)-algebra descends, up to homotopy, to a homotopy \(\BV\)-algebra if and only if its restriction to the homotopy fibre of the inclusion of \(\Com\) into \(\BV\) is nullhomotopic.
\end{enumerate}
In addition we have the following model-dependent algebraic counterparts:
\begin{enumerate}
	\item[(a1)] Any \(\rC_\infty\)-algebra defines a \(\cBV_{\!\infty}\)-algebra after arbitrarily choosing certain \textit{extra generating maps} described in \cref{ss:generating_maps}.
	\item[(a2)] A \(\cBV_{\!\infty}\)-algebra defines a \(\BV_{\!\infty}\)-algebra if and only if certain explicitly defined \textit{obstruction maps}, described in \cref{ss:obstruction_maps}, vanish.
\end{enumerate}
Moreover, the sets of extra generating maps and of obstruction maps are in a canonical bijection.

The picture that emerges in practice is that of a controlled extension process.
The kinematic data, presented as a \(\rC_\infty\)-algebra may always be promoted to a \(\cBV_{\!\infty}\)-algebra, but the promotion is not necessarily canonical: it depends on a choice of extra generating maps.
Each such new operation gives rise to an obstruction map, all of which vanish precisely when the resulting structure is that of a \(\BV_{\!\infty}\)-algebra.
In applications, the additional generating maps are tuned so as to control the obstructions in a physically meaningful way.

\medskip

An important use of the duality between colour and kinematics, which was pioneered in \cite{Bern2010PerturbativeQGDoubleCopy} at the scattering amplitude level, is the \textit{double copy construction}, in which the colour Lie algebra is formally replaced by the kinematic Lie algebra.
At cubic order, this construction has been understood off-shell using the tensor product of \(\cBV\)-algebras \cite{Borsten2025DoubleCopyFromTensorBV}.
For this tensor product to exist---like in the associative case and unlike the Lie case---it is necessary and sufficient that the operad \(\cBV\) admits a \emph{diagonal}, that is a morphism of operads
\[
\Delta \colon \cBV \to \cBV \ot_H \cBV,
\]
where \(\ot_H\) denotes the Hadamard (or aritywise) tensor product of operads.
We describe a canonical diagonal for \(\cBV\) in \cref{ss:hopf} recovering the tensor product defined in \cite{Borsten2025DoubleCopyFromTensorBV}.

Since \(\cBV_{\!\infty}\) is a cofibrant resolution of \(\cBV\), the lifting property in the model category of operads ensures the existence of a diagonal for \(\cBV_{\!\infty}\):
\[
\begin{tikzcd}
	&[-5pt] &[-10pt] \cBV_{\!\infty} \otimes_H \cBV_{\!\infty} \hspace*{3pt} \ar[d, ->> ,"\sim", shift right=5pt]\\
	\cBV_{\!\infty} \ar[r, ->>, "\sim \ "] \arrow[urr, dashed, out=45, in=180] & \cBV \ar[r, "\Delta"] & \cBV \otimes_H \cBV.
\end{tikzcd}
\]
Consequently, one obtains the existence of a universal formula for the tensor product of \(\cBV_{\!\infty}\)-algebras, which in principle allows for a formulation of the double copy extending beyond cubic interactions.
The construction, however, depends on a specific choice of diagonal, a task that falls outside the scope of this work.

\medskip

Mathematically, the central technical contribution of this work is the extension of Koszul duality to quadratic-linear presentations with non-trivial differentials.
This refinement allows the operad \(\cBV\) to be treated within the established framework of operadic calculus and, in particular, enables the explicit construction of its Koszul replacement.
As explained in \cref{ss:reiterer}, the resulting notion of \(\cBV_{\!\infty}\)-algebras, a model of homotopy \(\cBV\)-algebras, is more general and better behaved than the \(\BV^\square_\infty\)-algebras introduced in \cite{Reiterer2020HomotopyBVYMCK}, which we are unsure if defines a model for homotopy.

\medskip\noindent\textbf{Outline.}
In \cref{sec:algebras} we work at the level of algebras, avoiding the operadic constructions underlying our definitions.
We describe \(\cBV_{\!\infty}\)-algebras in terms of generating maps and relations, and clarify their relationships with \(\rC_\infty\)- and \(\BV_{\!\infty}\)-algebras.
We then introduce a weight filtration that organizes the hierarchy of homotopies and yields stricter forms of \(\cBV_{\!\infty}\)-algebras.
Finally, we illustrate this algebraic structure using the kinematic algebra of Yang--Mills theory as developed in \cite{Bonezzi2023GaugeInvariantDoubleCopyQuartic,Bonezzi2024WeaklyConstrainedDoubleField}.

\cref{sec:operads} constructs Diagram~\eqref{eq:diagram_main} and develops the homotopical theory of algebras over the operad \(\cBV_{\!\infty}\).
We begin in \cref{ss:strict_operads} by reviewing the operads \(\Com\), \(\BV\), and \(\cBV\) governing the strict algebraic structures of interest.
In \cref{ss:exact}, an independent subsection, we compare \(\cBV\) with the operad \(\eBV\) of exact \(\BV\)-algebras.
Section~\cref{ss:koszul_duals} extends Koszul duality to differential graded inhomogeneous presentations, and \cref{ss:cbv} applies this framework to \(\cBV\) to obtain an explicit model for its Koszul replacement \(\cBV_{\!\infty}\).
In \cref{ss:hopf} we show that tensor products of \(\cBV_{\!\infty}\)-algebras are well defined up to homotopy.
Section~\cref{ss:homotopy_fibre} studies the model-categorical relationship between \(\Com_\infty\), \(\cBV_{\!\infty}\), and \(\BV_{\!\infty}\).
In \cref{ss:homotopy_algebras} we prove that algebras over \(\cBV_{\!\infty}\) agree with the explicitly presented \(\cBV_{\!\infty}\)-algebras of \cref{sec:algebras}.
Finally, \cref{ss:tools} describes the deformation, obstruction, and homotopy theories of \(\cBV_{\!\infty}\)-algebras from the operadic perspective.

In \cref{ss:bv_bd} we record an interesting observation which, although deferred to a separate work for further development, is included in this appendix due to its close technical parallels with the constructions presented here.
It states that algebras over the Koszul dual cooperad of the Batalin--Vilkovisky operad are Belinson--Drinfeld algebras

\medskip\noindent\textbf{Outlook.}
Our focus in this paper is on the non-cyclic aspects of homotopy \(\cBV\)-algebras.
By contrast, the \(\rL_\infty\)-algebras arising in field theory typically carry a cyclic structure.
Accordingly, in examples of colour-splitting the associated kinematic \(\rC_\infty\)-structures are also cyclic.
Extending the present framework to cyclic \(\cBV_{\!\infty}\)-algebras, as well as to their quantum counterparts, constitutes an important direction for future research.

%% file: sec/acknowledgements.tex

\medskip\noindent\textbf{Acknowledgements}
We would like to thank
Roberto Bonezzi,
Christoph Chiaffrino,
Felipe D\'iaz-Jaramillo,
Owen Gwilliam,
Olaf Hohm,
Branislav Jur\v{c}o,
Victor Roca i Lucio,
and Dennis Sullivan
for interesting discussions which led to the present paper.
We also thank the IH\'ES for providing excellent working conditions during the development of this project.

%% file: sec/algebras.tex

\section{A model for homotopy \(\cBV\)-algebras}\label{sec:algebras}

In this section we introduce \(\cBV_{\!\infty}\)-algebras in concrete terms.
We begin by recalling the conventions underlying our constructions and review \(\Com\)-algebras, \(\cBV\)-algebras, and \(\BV\)-algebras.
We then introduce \(\cBV_{\!\infty}\)-algebras via generating maps, show how every \(\rC_\infty\)-algebra extends canonically to such a structure, and define obstruction maps measuring the failure of a \(\cBV_{\!\infty}\)-algebra to descend to \(\BV_{\!\infty}\)-algebra.
A weight filtration is introduced to organize the hierarchy of homotopies present in \(\cBV_{\!\infty}\)-algebras and define \(\cBV_{\!(n)}\)-algebras for \(n \in \N\).
Finally, we illustrate the utility of this framework with the kinematic algebra of Yang--Mills theory \cite{Bonezzi2023GaugeInvariantDoubleCopyQuartic, Bonezzi2024WeaklyConstrainedDoubleField}.

\subsection{Conventions}

\subsubsection{The underlying category}\label{ss:dgVec}

We work throughout in the closed symmetric monoidal category of cochain complexes over a field \(\KK\) of characteristic \(0\).
We spell this out for the convenience of the reader.

\medskip\noindent
A \textit{cochain complex} is a pair \((V, d)\) consisting of a graded vector space \(V = \bigoplus_{n \in \Z} V^n\) and a \textit{differential}, that is, a linear map \(d \colon V \to V\) such that \(d(V^n) \subseteq V^{n+1}\) for all \(n \in \Z\) and \(d \circ d = 0\).

\medskip\noindent
Saying that the category is \emph{monoidal} means that it is equipped with a product and a unit object satisfying the usual associativity and unitality constraints.
In this case, we have the tensor product
\[
(V \otimes W)^n \;=\; \bigoplus_{p+q=n} V^p \otimes_\KK W^q,
\]
with differential
\[
d_{V \otimes W}(v \otimes w) \;\defeq\; d(v) \otimes w \;+\; (-1)^{|v|} v \otimes d_W(w),
\]
and the unit object being \(\KK\) concentrated in degree \(0\) with trivial differential.

\medskip\noindent
Saying that the category is \emph{closed} means that we have an \textit{internal Hom}: for any cochain complex \(V,W\), the collection of linear maps from \(V\) to \(W\) itself forms a cochain complex
\[
\Hom(V,W) \;=\; \bigoplus_{k \in \Z} \Hom^k(V,W),
\]
where \(\Hom^k(V,W)\) consists of all linear maps \(f \colon V \to W\) of \textit{degree \(k\)}, i.e.\ satisfying \(f(V^n) \subseteq W^{n+k}\) for all \(n\) in \(\Z\).
The differential on \(\Hom(V,W)\) is given by
\[
d_{\Hom}(f) \;\defeq\; d_W \circ f \;-\; (-1)^{|f|} f \circ d .
\]
This internal Hom gives rise to a canonical \textit{dg adjunction}:
\[
\Hom(U \otimes V, W) \;\cong\; \Hom\big(U,\, \Hom(V,W)\big)
\]
for any \(U,V,W\).

\medskip\noindent
Saying that this category is \emph{symmetric} asserts the existence of the isomorphism
\begin{equation}\label{eq:symmetry_isomorphism}
	\begin{tikzcd}[row sep=0, column sep=small]
		V \otimes W \arrow[r, "\tau_{V,W}", shift left=0.5ex] & W \otimes V \\
		v \otimes w \arrow[mapsto, r] & (-1)^{|v||w|}\, w \otimes v,
	\end{tikzcd}
\end{equation}
for all homogeneous elements \(v \in V\) and \(w \in W\).

\subsubsection{Symmetric group action}

The \textit{symmetry isomorphism} \eqref{eq:symmetry_isomorphism} defines a left action of the symmetric group \(\Sy_n\) on \(V^{\otimes n}\) by
\[
\sigma (v_1 \otimes \dotsb \otimes v_n)
\defeq
\pm\, v_{\sigma^{-1}(1)} \otimes \dotsb \otimes v_{\sigma^{-1}(n)},
\]
for all \(\sigma \in \Sy_n\) and homogeneous \(v_i \in V\),
where the sign is determined by the \emph{Koszul convention}.

This induces a natural right \(\Sy_n\)-action on \(\Hom(V^{\otimes n},V)\) given by precomposition.
Explicitly, for \(f \in \Hom(V^{\otimes n},V)\) and \(\sigma \in \Sy_n\)
\[
(f \circ \sigma^{-1})(v_1 \otimes \dotsb \otimes v_n)
\defeq
\pm\, f \big(v_{\sigma(1)} \otimes \dotsb \otimes v_{\sigma(n)}\big).
\]

\subsubsection{Partial compositions and the insertion bracket}

Given \(f \in \Hom(V^{\ot n}, V)\), \(g \in \Hom(V^{\ot m}, V)\), and \(1 \leqslant i \leqslant n\), the \defn{partial composition}
\[
f \circ_i g \;\in\; \Hom\!\big(V^{\ot (n+m-1)}, V\big)
\]
is defined on homogeneous inputs \(v_1 \ot\dotsb\ot v_{n+m-1}\) by
\[
\begin{split}
	(-1)^{|g| \cdot (|v_1| + \dotsb + |v_{i-1}|)}
	f\big(v_1 \ot\dotsb\ot v_{i-1} \ot g(v_i \ot\dotsb\ot v_{i+m-1})
	\ot v_{i+m} \ot\dotsb\ot v_{n+m-1}\big).
\end{split}
\]

\medskip\noindent
From these partial compositions we define the \defn{pre-Lie product}
\begin{equation}\label{eq:preLie_product}
	f \star g \;\defeq\; \sum_{i=1}^{n} f \circ_i g ,
\end{equation}
and its antisymmetrization, the \defn{insertion bracket},
\[
[f,g] \;\defeq\; f \star g \;-\; (-1)^{|f||g|}\, g \star f.
\]
This bracket makes
\[
\End_V \;\defeq\; \prod_{n \geqslant 1} \Hom(V^{\ot n},V)
\]
into a \emph{dg Lie algebra} extending the (graded) commutator Lie algebra on \(\Hom(V,V)\).

\subsection{Strict algebras}

A differential graded commutative algebra, or more compactly a \defn{\(\Com\)-algebra}, is a tuple \((A,d,m)\), where \((A,d)\) is a cochain complex and \(m \colon A \ot A \to A\) is a (graded) symmetric bilinear map of degree \(0\), i.e. \(m \circ \big(\id - (12)\big) = 0\), such that:
\[
m \circ_1 m - m \circ_2 m = 0, \qquad
[d, m] = 0.
\]
The first identity states the associativity of \(m\) and the second that \(d\) is a \textit{derivation} of \(m\).

\begin{definition}\label{def:cbv_algebra}
	A \defn{\(\cBV\)-algebra} (\emph{coexact Batalin--Vilkovisky algebra}) is a tuple \((A,d,m,\triangle)\), where \((A,d,m)\) is a \(\Com\)-algebra and \(\triangle \in \Hom(A, A)\) is a degree \(-1\) map satisfying \(\triangle \circ \triangle = 0\) and the \textit{second-order relation}:
	\begin{equation}\label{eq:second-order}
		\begin{split}
			\triangle(abc)
			&= \triangle(ab)c
			+ (-1)^{|a|}a\,\triangle(bc)
			+ (-1)^{(|a|+1)|b|}b\,\triangle(ac) \\[4pt]
			&\quad
			- \triangle(a)bc
			- (-1)^{|a|}a\,\triangle(b)c
			- (-1)^{|a|+|b|}ab\,\triangle(c),
		\end{split}
	\end{equation}
	where \(m(a \otimes b)\) is denoted \(ab\).
\end{definition}

\medskip\noindent A (differential graded) \defn{\(\BV\)-algebra}, short for (dg) \textit{Batalin--Vilkovisky algebra}, is a \(\cBV\)-algebra \((A, d, m, \triangle)\) where the \defn{obstruction}
\[
n \defeq [d, \triangle]
\]
is identically \(0\).

\begin{remark}
	Both \(\cBV\) and \(\BV\) algebras may also be defined in the absence of a differential, in which case the two notions coincide.
	Throughout, we restrict attention to their differential graded versions.
\end{remark}

\begin{example}
	For a smooth manifold \(M\), the \(\Com\)-algebra \((\Omega(M), d, \wedge)\) of differential forms extends to a \(\cBV\)-algebra once a non-degenerate metric is chosen, with the Hodge codifferential \(d^\star\) playing the role of \(\triangle\).
	In the Lorentzian case, the obstruction \([d, d^\star]\) is the d'Alembertian, typically denoted \(\square\).
	This is the origin of the alternative terminology \(\BV^\square\)-algebras for \(\cBV\)-algebras \cite{Reiterer2020HomotopyBVYMCK}.
\end{example}

In a \(\cBV\)-algebra, the failure of \(\triangle\) to be a derivation of the product \(\m\) defines a degree \(-1\) symmetric bracket \(b \defeq [\triangle,\m]\) satisfying the \defn{Jacobi relation}
\[
(b \,\circ_1 b) \circ \big(\id + (123) + (132)\big) = 0
\]
and the \defn{Leibniz relation}
\[
b \,\circ_1 m = m \,\circ_1 b + (m \,\circ_2 b) \circ (213).
\]
The latter is equivalent to the second-order relation~\eqref{eq:second-order}.
Thus \((A,d,\m,b)\) is almost a \textit{Gerstenhaber algebra} except that \([d,b] = [d,[\triangle,\m]] = [[d,\triangle],\m] = [n,\m]\) need not vanish if the obstruction \(n \neq 0\).

\subsection{\(\cBV_{\!\infty}\)-algebras}\label{ss:generating_maps}

Recall that a \defn{\(\rC_\infty\)-algebra} is a graded vector space \(A\) together with maps
\[
m_n \colon A^{\ot n} \to A, \qquad n \geqslant 1,
\]
of degree \(2-n\) with the following symmetries holding for every \(j\):
\[
\sum_{\mathclap{\sigma \in \Sh(j,\, n - j)}} \ \sign(\sigma)\,
m_n \! \circ \sigma = 0.
\]
They satisfy the relation:
\[
\sum_{r+s+t = n} (-1)^{r + st}\, m_{r+1+t} \circ (\id^{\ot r} \ot \, m_s \ot \id^{\ot t}) = 0
\]
for any \(n \geqslant 1\).
If \(m_k = 0\) for all \(k > n+1\) we say that this is a \defn{\(\rC_{(n)}\)-algebras}.

\medskip

Let us fix a graded vector space \(A\).
A set of \defn{generating maps} on \(A\) consists of a collection of linear maps
\begin{align*}
	m_{p_1,\dots,p_k}^t &\colon
	A^{\ot p_1} \ot\dotsb\ot A^{\ot p_k}
	\longrightarrow A,
\end{align*}
for each \(	t \geqslant 0, \ k \geqslant 1, \ p_1, \dots, p_k \geqslant 1\)\,, of degree
\[
\bars{m_{p_1, \dots, p_k}^t} = 3 - 2t - p_1 - \cdots - p_k - k,
\]
with the following \defn{block} and \defn{shuffle symmetries}:

\medskip\noindent(1)
For any permutation \(\sigma \in \sym_k\):
\[
m_{p_{\sigma^{-1}(1)},\dots,p_{\sigma^{-1}(k)}}^t \!= m_{p_1, \dots, p_k}^t \! \circ \overline\sigma,
\]
where \(\overline\sigma\) is the image of \(\sigma\) via the block inclusion \(\sym_k \to \sym_{p_1+\dots+p_k}\).

\medskip\noindent(2)
For any \(i \in \set{1,\dots,k}\) and \(j \in \set{1, \dots, p_i-1}\):
\[
\sum_{\mathclap{\ \sigma \in \Sh(j,\, p_i - j)}} \ \sign(\sigma) \;
m_{p_1,\dots,p_k}^t \! \circ \underline\sigma = 0,
\]
where \(\underline\sigma\) is the image of \(\sigma\) in via the inclusion \(\sym_{p_i} \to \sym_{p_1+\dots+p_k}\) induced by the \(i^\th\)-block.

\medskip In \cref{label} we will prove that the following definition coincides with that of an algebra over a cofibrant resolution of the operad controlling \(\cBV\)-algebras.

\begin{definition}\label{def:cBV-algebra}
	A \defn{\(\cBV_{\!\infty}\)-algebra} is a graded vector space \(A\) together with a set of generating maps \(\set{m_{p_1, \dots, p_k}^t}\) such that \((A, m^0_1, m^0_2, m^0_3, \dots )\) is a \(\rC_\infty\)-algebra.
\end{definition}

\noindent A straightforward consequence of this characterization is the following.

\begin{theorem}\label{t:extension from C_infty}
	Any \(\rC_\infty\)-algebra \(\big(A, m_1, m_2, m_3, \dots\big)\) can be extended to a \(\cBV_{\!\infty}\)-algebra by completing \(\set{m^0_n \defeq m_n}_{n \geqslant 1}\) to a set of generating maps.
\end{theorem}

\noindent A canonical, although not very interesting choice, is to set all the additional maps to \(0\).

\subsection{Obstruction maps and \(\BV_{\!\infty}\)-algebras}\label{ss:obstruction_maps}

\subsubsection{Straight shuffles}

To understand explicitly the relationship between \(\cBV_{\!\infty}\)-algebras and \(\BV_{\!\infty}\)-algebras, we will need the combinatorial notion of \textit{straight shuffle}, which we illustrate with the following example:
\begin{equation}\label{ex:straight_shuffle}
	1\;\underline{2}\;3\;4\;|\;\underline{5\;6}\;7\;8\;|\;9\;\underline{10\;11}\;|\;\underline{12\;13}
	\ \longrightarrow\
	9\;1\;\;\underline{2\;5\;6\;10\;11\;12\;13}\;\;7\;3\;8\;4.
\end{equation}

\medskip\noindent Forgetting all decoration in \eqref{ex:straight_shuffle}, we have the permutation
\[
\sigma =
\left[
\begin{array}{ccccccccccccc}
	1&2&3&4&5&6&7&8&9&10&11&12&13\\
	2&3&11&13&4&5&10&12&1&6&7&8&9
\end{array}
\right]
\in \sym_{13}.
\]
The segmentation induced by the vertical bars in \eqref{ex:straight_shuffle} is encoded in the \(k\)-tuple
\[
\bar p \defeq (4,4,3,2)
\]
with \(k = 4\).
The subdivision of each segment of \eqref{ex:straight_shuffle} induced by the underline is encoded in the three \(k\)-tuples
\[
\bar l = (1,0,1,0), \quad
\bar q = (1,2,2,2), \quad
\bar r = (2,2,0,0).
\]
By definition, the permutation \(\sigma\) is a \defn{ \((\bar l, \bar q, \bar r, \bar p)\)-shuffle} if for all \(1 \leqslant i \leqslant k\),
\[
(p_1+\dots+p_{i-1} + l_i) + [1, q_i ] \xmapsto{\sigma}
(l_1 + \dotsb + l_k) + (q_1 + \dotsb + q_{i-1}) + [1, q_i].
\]
Explicitly,
\[
2 \xmapsto{\sigma} 3, \quad
[5, 6] \xmapsto{\sigma} [4, 5], \quad
[10, 11] \xmapsto{\sigma} [6, 7], \quad
[12, 13] \xmapsto{\sigma} [8, 9].
\]

\medskip

We will be interested in the set \(\mathrm{StSh}(\bar q, \bar p)\) of all straight shuffles for a fixed pair
\[
\bar q = (q_1, \dots, q_k) \leqslant (p_1, \dots, p_k) = \bar p.
\]

\subsubsection{Straight shuffle extensions}

Let us now introduce some notation for the action of straight shuffles on tensor product elements of a graded vector space.
Using the straight shuffle \(\sigma\) from example \eqref{ex:straight_shuffle}, we have that if \(x = x_1 \ot\dotsb\ot x_{13}\) then
\[
\sigma x =
\pm\, \underbrace{x_9 \ot x_1}_{\sigma x^{(1)}} \ot
\underbrace{x_2 \ot x_5 \ot x_6 \ot x_{10} \ot x_{11} \ot x_{12} \ot x_{13}}_{\sigma x^{(2)}} \ot
\underbrace{x_7 \ot x_3 \ot x_8 \ot x_4}_{\sigma x^{(3)}},
\]
where the sign is given by the permutation of elements under the Koszul sign rule.

\medskip

The \defn{straight shuffle extension} of a map
\[
m_{q_1, \dots, q_k} \colon A^{\ot q_1} \ot\dotsb\ot A^{\ot q_k} \to A
\]
to a map
\[
m_{\frac{p_1,\dots,p_k}{q_1,\dots,q_k}} \colon {A^{\ot p_1}} \ot\dotsb\ot {A^{\ot p_k}} \to A,
\]
where \((q_1, \dots, q_k) \leqslant (p_1, \dots, p_k)\), is defined by
\[
m_{\frac{p_1,\dots,p_k}{q_1,\dots,q_k}}(x)
=
\sum_{\quad \mathclap{\sigma \in \mathrm{StSh}(\bar q, \bar p)}}
\, \pm \,
\sigma x^{(1)} \ot
m_{q_1,\dots,q_k}\big( \sigma x^{(2)} \big) \ot \sigma x^{(3)}.
\]
The sign is given by the product of the following four contributions:
the sign of \(\sigma\),
the Koszul sign given by permutation of the elements, as referenced above,
the Koszul sign arising from the permutation of \(\sigma x^{(1)}\) and \(m_{q_1,\dots,q_k}\),
and the explicit factor:
\[
(-1)^{(q_1+\cdots+q_r+1)(r_1+\cdots+r_k) + (l_2+r_2) + 2(l_3+r_3) + \dotsb + (k-1)(l_k+r_k)}
\]
expressed under the assumption that \(\sigma\) is a \((\bar l, \bar q, \bar r, \bar p)\)-shuffle.

\subsubsection{Obstruction maps}\label{sssec:Obstructions}

For each \(t \geqslant 0, \ k \geqslant 1, \ p_1, \dots, p_k \geqslant 1\) with \(t + k \geqslant 2\), the \defn{obstruction map} \(n^t_{p_1, \dots, p_k}\) is defined on a basis element
\[
\underbrace{(x_1 \ot\dotsb\ot x_{p_1})}_{w_1}
\ot \underbrace{(x_{P_1 + 1} \ot\dotsb\ot x_{P_1 + p_2})}_{w_2}
\ot\dotsb\ot
\underbrace{(x_{P_{k-1} + 1} \ot\dotsb\ot x_{P_k})}_{w_k},
\]
where \(P_i = p_1 + \dots + p_{i-1}\), as
\begin{equation}\label{eq:relation_maps_1}
	\sum_{\mathclap{\begin{array}{c}
			\scriptstyle 0 \leqslant s \leqslant t
	\end{array}}}
	\hspace{25pt}
	\sum_{\hspace*{10pt}\mathclap{\begin{array}{l}
				\scriptstyle I \,\sqcup J = \{1, \dots, k\} \\[-2pt]
				\scriptstyle I = \{i_1, \dots, i_a\} \neq \emptyset \\[-2pt]
				\scriptstyle J = \{j_1, \dots, j_b\}
	\end{array}}}
	\hspace*{35pt} 
	\sum_{\hspace*{0pt}\mathclap{\begin{array}{c}
				\scriptstyle 1 \leqslant q_1 \leqslant p_1 \\[-2pt]
				\scalebox{0.5}{\(\vdots\)} \\[-2pt]
				\scriptstyle 1 \leqslant q_a \leqslant p_a
	\end{array}}}
	\hspace*{3pt}	
	\pm \ m^{s}_{r,\, p_{j_1}, \dots, p_{j_b}}
	\Big(m^{t-s}_{\frac{p_1, \dots, p_{i_a}}{q_1, \dots, q_a}}
	(w_{i_1} \!\ot\dotsb\ot w_{i_a}) \ot w_{j_1} \!\ot\dotsb\ot w_{j_b}\Big),
\end{equation}
where \(r = 1 + p_{i_1} + \dots + p_{i_a} - q_1 - \dots - q_a\)\,; minus the sum
\begin{equation}\label{eq:relation_maps_2}
	\sum_{1 \leqslant i \leqslant k}
	\hspace*{15pt}
	\sum_{\mathclap{1 \leqslant j \leqslant p_i-1}}
	\ (-1)^{p_1+\cdots+p_{i-1}+j-i} \ m^{t-1}_{p_1, \dots, j, p_i-j, \dots, p_k}
	\Big(w_1 \ot\dotsb\ot w_i^{(1)} \ot w_i^{(2)} \ot\dotsb\ot w_k\Big),
\end{equation}
where \(w_i^{(1)} = x_{P_i+1} \ot\dotsb\ot x_{P_i+j}\) and \(w_i^{(2)} = x_{P_i+j+1} \ot\dotsb\ot x_{P_{i+1}}\).

\medskip

The sign in \eqref{eq:relation_maps_1} is given by the product of the following three contributions:
the Koszul sign arising from the permutation of the elements,
the sign of the \((a,b)\)-shuffle corresponding to the partition \(I \sqcup J = \{1,\dots,k\}\)---computed by considering the blocks of sizes \(p_1-1,\dots,p_k-1\) and shuffling them according to the permutation determined by this partition (so, in particular, if all \(p_i=1\), there is no sign)---and the explicit factor
\[
(-1)^{(q_1+\dotsb+q_a - a + 1)\,(p_{j_1}+\cdots+p_{j_b} - b) + r - 1}.
\]

These obstruction maps \(n^t_{p_1, \dots, p_k}\) satisfy the same block and shuffle symmetries as the operations
\(m_{p_1,\dots,p_k}^t\) and the have degree
\[|n^t_{p_1, \dots, p_k}|=4 - 2t - p_1 - \cdots - p_k - k.\]

\subsubsection{\(\BV_{\!\infty}\)-algebras}\label{BVinfinity}

For more details on the notions of \(\BV_{\!\infty}\)-algebras, \(\Gerst_\infty\)-algebras, and shifted \(\rL_\infty\)-algebras, we refer to \cite{GCTV12}.

\begin{theorem}[{\cite[\S 2.3]{GCTV12}}]\label{t:obstruction vanishing}
	Let \(\big(A, \set{m_{p_1, \dots, p_k}^t}\big)\) be a \(\cBV_{\!\infty}\)-algebra.
	\begin{enumerate}
		\item \(\big(A, \set{m_{p_1, \dots, p_k}^t}\big)\) is a \(\BV_{\!\infty}\)-algebra if and only if every \(n^t_{p_1, \dots, p_k}\) vanishes.
		\item \(\big(A, \set{m^0_{p_1, \dots, p_k}}\big)\) is a \(\Gerst_\infty\)-algebra if and only if every \(n^0_{p_1, \dots, p_k}\) vanishes.
		\item \(\big(A, m^0_1, m^0_{1,1}, m^0_{1,1,1}, \dots\big)\) is a shifted \(\rL_\infty\)-algebra if and only if every \(n^0_{1, \dots, 1}\) vanishes.
	\end{enumerate}
\end{theorem}

\subsection{Weight filtration}

\subsubsection{\(\cBV_{\!(n)}\)-algebras}

We assign to each \(m^t_{p_1, \dots, p_k}\) in a set of generating maps the following \defn{weight}:
\[
\mathrm{weight}(m^t_{p_1, \dots, p_k}) \defeq t + p_1 + \dots + p_k - 1.
\]

\begin{definition}
	A \defn{\(\cBV_{\!(n)}\)-algebra} is a \(\cBV_{\!\infty}\)-algebra in which all generating maps of weight greater than \(n\) vanish.
\end{definition}

\noindent As in \cref{t:extension from C_infty}, every \(\rC_{(n)}\)-algebra extends to a \(\cBV_{\!(n)}\)-algebra, and, similarly to \cref{t:obstruction vanishing}, a \(\cBV_{\!(n)}\)-algebra descends to a \(\BV_{(n)}\)-algebra, a \(\Gerst_{(n)}\)-algebra, or a shifted \(\rL_{(n)}\)-algebra if and only if the same sets of obstruction maps vanish.
Among these obstructions, only finitely many can be non-zero, namely those whose weight
\[
\mathrm{weight}(n^t_{p_1, \dots, p_k}) \defeq t + p_1 + \dots + p_k - 1
\]
is at most \(n\).

\subsubsection{Initial stages}\label{ss:explicit_descriptions}

We provide an explicit description of the generating and obstruction maps in the initial stages of the weight filtration.

\smallskip\noindent\textsc{Generators}\nopagebreak

\smallskip\noindent\textit{Weight 0}.
The only generating map of weight \(0\) is \(m^0_1\) of degree \(\bars{m^0_1} = 1\).

\smallskip\noindent\textit{Weight 1}.
The set of generating maps of weight \(1\) is:
\[
m^1_1 \colon A \to A, \qquad
m^0_2 \colon A^{\ot 2} \to A, \qquad
m^0_{1,1} \colon A \ot A \to A,
\]
of degrees
\[
\bars{m^1_1} = -1, \quad
\bars{m^0_2} = 0, \quad
\bars{m^0_{1,1}} = -1.
\]
They are required to have the following symmetric properties:
\[
m^0_{2} = m^0_{2} \circ (12), \qquad
m^0_{1,1} = m^0_{1,1} \circ (12),
\]
or, equivalently, for homogeneous elements \(a,b \in A\),
\[
m^0_{2}(a,b) = (-1)^{|a||b|} m^0_{2}(b,a), \qquad
m^0_{1,1}(a,b) = (-1)^{|a||b|} m^0_{1,1}(b,a).
\]

\smallskip\noindent\textit{Weight 2}.
The subset of independent generating maps of weight \(2\) is:
\begin{align*}
	&m^0_{3} \colon A^{\ot 3} \to A,
	&& m^0_{1,2} \colon A \ot A^{\ot 2} \to A,
	&& m^0_{1,1,1} \colon A \ot A \ot A \to A, \\
	&m^1_{1,1} \colon A \ot A \to A,
	&& m^1_{2} \colon A^{\ot 2} \to A,
	&& m^2_{1} \colon A \to A,
\end{align*}
whose degrees are:
\begin{align*}
	&\bars{m^0_{3}} = -1, && \bars{m^0_{1,2}} = -2, && \bars{m^0_{1,1,1}} = -3, \\
	&\bars{m^1_{1,1}} = -3, && \bars{m^1_{2}} = -2, && \bars{m^2_{1}} = -3.
\end{align*}
They are required to have the following symmetry properties:
\begin{equation*}\label{eq:symmetry_weight2}
	\begin{split}
		m^0_3 \circ \big(\id - (12) + (132)\big) &= 0, \qquad \Big(m^0_3 \circ \big(\id - (23) + (123)\big) = 0,\Big) \\
		m^0_{1,2} \circ \big(\id - (23)\big) &= 0, \qquad \Big(m^0_{2,1} - m^0_{1,2} \circ (123) = 0,\Big) \\
		m^0_{1,1,1} \circ \big(\id - (12)\big) &= 0, \qquad m^0_{1,1,1} \circ \big(\id - (23)\big) = 0, \\
		m^1_{1,1} \circ \big(\id - (12)\big) &= 0, \qquad m^1_{2} \circ \big(\id - (12)\big) = 0.
	\end{split}
\end{equation*}

\smallskip\noindent\textit{Weight 3}.
There are 11 independent generators of weight 3:
\begin{align*}
	&m^0_{4}
	&& m^0_{1,3}
	&& m^0_{2,2}
	&& m^0_{1,1,2}
	&& m^0_{1,1,1,1}
	\\
	&m^1_{3}
	&& m^1_{1,2}
	&& m^1_{1,1,1}
	\\
	&m^2_{2}
	&& m^2_{1,1}
	\\
	&m^3_{1}
\end{align*}

\medskip\noindent\textsc{Obstructions}\nopagebreak

\medskip\noindent\textit{Weight 1}.\par\nopagebreak
\medskip\noindent\input{sec/obstructions_weight1}

\medskip\noindent\textit{Weight 2}.\par\nopagebreak
\medskip\noindent\input{sec/obstructions_weight2}

\medskip\noindent\textit{Weight 3}.\par\nopagebreak
\medskip\noindent\input{sec/obstructions_weight3}

\medskip\noindent We describe explicitly the properties obstructed by the terms of weight at most 2:

\smallskip\noindent
\begin{itemize}[itemsep=3pt, label*=\(\diamond\), leftmargin=12pt]
	\item \(n^1_1\) obstruction to \(m^0_1\) commuting with the operator \(m^1_1\).
	\item \(n^0_{1,1}\) obstruction to \(m^0_1\) acting as a derivation of the shifted bracket \(m^0_{1,1}\).
	\item \(n^0_{1,2}\) obstruction to \(m^0_2\) and \(m^0_{1,1}\) satisfying the Leibniz relation up to the homotopy \(m^0_{1,2}\).
	\item \(n^0_{1,1,1}\) obstruction to \(m^0_{1,1,1}\) serving as a homotopy for the (shifted) Jacobi relation of \(m^0_{1,1}\).
	\item \(n^1_{1,1}\) obstruction to \(m^1_{1,1}\) serving as a homotopy for the commutation relation between \(m^1_1\) and \(m^0_{1,1}\).
	\item \(n^1_2\) obstruction to \(m^1_2\) serving as a homotopy for the relation expressing \(m^0_{1,1}\) as the commutator of \(m^1_1\) with \(m^0_2\).
	\item \(n^2_1\) obstruction to \(m^2_1\) serving as a homotopy for the square-zero relation of \(m^1_1\).
\end{itemize}

\subsection{Yang--Mills kinematic algebra}

We now review an example of a \(\cBV_{\!(2)}\)-algebra arising as the kinematic algebra of Yang--Mills theory.
It was constructed and developed in \cite{Bonezzi2022GaugeStructureDoubleField, Bonezzi2023GaugeInvariantDoubleCopyQuartic, Bonezzi2024WeaklyConstrainedDoubleField}.

\medskip\noindent
Let \(\cO\) denote the vector space of smooth real-valued function on \(\R^d\) with the Minkowski metric \(\eta = \operatorname{diag}(+,-,\dots,-)\).
We use the following short hand notation throughout:
\[
\partial_\mu \lambda = \frac{\partial \lambda}{\partial t} + \sum_{i = 1}^{d-1} \frac{\partial \lambda}{\partial x^i} \,,
\qquad
\partial^\mu \lambda = \eta^{\mu \nu} \partial_\nu \lambda = \frac{\partial \lambda}{\partial t} - \sum_{i = 1}^{d-1} \frac{\partial \lambda}{\partial x^i} \,.
\]
Let \(\cK = \cZ \ot \cO\) where \(\cZ\) be the following graded vector space with a given basis:
\[
\begin{tikzcd}[row sep = 0]
	\cZ_0 & \cZ_1 & \cZ_2 & \cZ_3 \\
	\R\set{\theta_+} & \R\set{\theta_0, \dots, \theta_{d-1}} & \R\set{\theta_-} \\
	& \oplus & \oplus & \\
	& \R\set{s\theta_+} & \R\set{s\theta_0, \dots, s\theta_{d-1}} & \R\set{s\theta_-}. \\
\end{tikzcd}
\]
We will define multilinear maps on \(\cK\) using this basis, specifically, assigning to each basis element of \(\cZ^{\ot r}\) a sum of differential operator \(\cO^{\ot r} \to \cO\) parameterized by \(\cZ\).
More precisely, we will use the natural inclusion
\[
\Hom(\cZ^{\ot r}, \cZ \ot \Hom(\cO^{\ot r}, \cO)) \to \Hom(\cK^{\ot r}, \cK).
\]
Additionally, we will omit operators that are identically \(0\) and leave implicit the operator \(\id \ot\dotsb\ot \id\), which agrees with the product of functions.

\medskip\noindent The map \(m_1\) is defined by:
\begin{align*}
	&\theta_+ \mapsto \theta_\mu \partial^\mu + s\theta_+ \partial^\mu \partial_\mu,
	&&s\theta_+ \mapsto - s\theta_\mu \partial^\mu - \theta_-, \\
	&\theta_\mu \mapsto s\theta_\mu \partial^\mu \partial_\mu + \theta_- \partial_\mu,
	&&s\theta_\mu \mapsto - s\theta_- \partial_\mu, \\
	&\theta_- \mapsto s\theta_- \partial^\mu \partial_\mu,
	&&s\theta_- \mapsto 0.
\end{align*}
The map \(m_2\), which is assumed to satisfy \(m_2 \circ (12) = m_2\), is defined by:

\noindent\hspace*{5pt}
\begin{minipage}{0.2\textwidth}
	\begin{align*}
		&\theta_+ \ot \theta_+
		\mapsto \theta_+ , \\
		&\theta_+ \ot \theta_\mu
		\mapsto \theta_\mu + s\theta_+( \partial_\mu \ot \id + \id \ot \partial_\mu), \\
		&\theta_+ \ot \theta_-
		\mapsto -s\theta_\mu (\id \ot \partial^\mu), \\
		&\theta_+ \ot s\theta_\mu
		\mapsto s\theta_\mu , \\
		&\theta_+ \ot s\theta_-
		\mapsto s\theta_- ,
	\end{align*}
\end{minipage}
\begin{minipage}{0.5\textwidth}
	\begin{align*}
		&\theta_\mu \ot \theta_\nu
		\mapsto
		s\theta_\nu ( \partial_\mu \ot \id + 2\, \id \ot \partial_\mu ) \\
		&\phantom{(\theta_\mu, \theta_\nu) \mapsto}
		- s\theta_\mu ( \id \ot \partial_\nu + 2\, \partial_\nu \ot \id ) \\
		&\phantom{(\theta_\mu, \theta_\nu) \mapsto}
		+ s\theta_\rho \eta_{\mu\nu} ( \partial^\rho \ot \id - \id \ot \partial^\rho ), \\
		&\theta_\mu \ot \theta_-
		\mapsto s\theta_- (\id \ot \partial_\mu), \\
		&\theta_\mu \ot s\theta_\nu
		\mapsto -s\theta_- \eta_{\mu \nu}.
	\end{align*}
\end{minipage}

\medskip\noindent
The map \(m_3\) is defined by
\[
\theta_\mu \ot \theta_\nu \ot \theta_\rho
\mapsto
s\theta_\mu \eta_{\nu\rho} - 2s\theta_\nu \eta_{\mu\rho} + s\theta_\rho \eta_{\nu\nu}.
\]

In \cite{Bonezzi2022GaugeStructureDoubleField}, the authors verify that these maps define a \(\rC_{(2)}\)-algebra structure on \(\cK\).
By \cref{t:extension from C_infty}, this structure can be extended to a \(\cBV_{\!\infty}\)-algebra structure by choosing, for each \(	t \geqslant 0, \ k \geqslant 1, \ p_1, \dots, p_k \geqslant 1\) with \(t+k \geqslant 2\), a generating maps \(m^t_{p_1,\dots,p_k}\).
We now review the choices made by the authors in \cite{Bonezzi2023GaugeInvariantDoubleCopyQuartic, Bonezzi2024WeaklyConstrainedDoubleField}.

\medskip\noindent\textsc{Generators}\nopagebreak

\medskip\noindent\textit{Weight 0}.
We set \(m^0_1 \defeq m_1\).

\medskip\noindent\textit{Weight 1}.
We set \(m^0_2 \defeq m_2\).
The generating map \(m^1_1\) is defined by
\begin{align*}
	s\theta_+ \mapsto \theta_+, \quad
	s\theta_\mu \mapsto \theta_\mu, \quad
	s\theta_- \mapsto \theta_-.
\end{align*}
The final generating map of weight \(1\) is \(m^0_{1,1} \defeq [m^1_1, m^0_2]\).

\medskip\noindent\textit{Weight 2}.
We set \(m^0_3 \defeq m_3\).
The following generating maps are chosen to be \(0\):
\begin{align*}
	m^1_2 &= 0,	& m^2_1 &= 0,	& m^1_{1,1} &= 0.
\end{align*}
The authors explicitly define a map \(\theta_3 \colon \cK^{\ot 3} \to \cK\), found in \cite[Appendix~A]{Bonezzi2024WeaklyConstrainedDoubleField}, which we use to define \(m^0_{1,2} \defeq \theta_3 \circ (13)\).
The final generating map of weight 2 is \(m^0_{1,1,1} \defeq [m^1_1, m^0_{1,2}]\).

\medskip\noindent\textit{Weight \(\geqslant\)\,3}. All other generating maps are chosen to be identically 0.
Therefore, the resulting structure on \(\cK\) is that of a \(\cBV_{\!(2)}\)-algebra.

\medskip\noindent\textsc{Obstructions}\nopagebreak

\medskip\noindent
We use the general list in \cref{ss:explicit_descriptions} to make explicit some of the obstruction maps arising from the above choices of generating maps.

\medskip\noindent\textit{Weight 1}.
The obstruction map \(n^1_1 = [m^0_1, m^1_1]\) is the d'Alambertian operator \(\square\) and \(n^0_{1,1} = [m^0_1, m^0_{1,1}]\).

\medskip\noindent\textit{Weight 2}.
Three obstruction maps are identically 0:
\begin{align*}
	n^1_2 &= 0,	& n^2_1 &= 0,	& n^1_{1,1} &= 0.
\end{align*}
The other two, \(n^0_{1,2}\) and \(n^0_{1,1,1}\), are obstructions to the homotopy Leibniz and Jacobi relations respectively.
They are left uncomputed.

\medskip\noindent\textit{Weight 3}.
There are four obstruction maps that are identically 0:
\[
n^1_{1,2} = 0, \qquad n^2_{1,1} = 0, \qquad n^2_2 = 0, \qquad n^3_1 = 0.
\]
Two others are simplified compared to their general definition:
\[
n^1_3 = [m^1_1, m^0_3] - m^0_{1,2} - m^0_{2,1}
\quad\text{and}\quad
n^1_{1,1,1} = [m^1_1, m^0_{1,1,1}].
\]
The remaining four obstruction maps are left uncomputed.

%% file: sec/obstructions_weight1.tex
\renewcommand{\arraystretch}{1.25}
\begin{tabularx}{\textwidth}{|@{}r@{\hspace{2pt}}X@{}|}
	\hline
	\(n^1_1\) & \( = \displaystyle m^0_1 \circ m^1_1 + m^1_1 \circ m^0_1 \)
	\( = [m^0_1, m^1_1] \) \\

	\hline
	\,\(n^0_{1,1}\) & \( = \displaystyle
	m^0_1 \circ m^0_{1,1}
	- m^0_{1,1} \circ_1 m^0_1
	- m^0_{1,1} \circ_2 m^0_1
	= [m^0_1, m^0_{1,1}] \) \\

	\hline
\end{tabularx}

%% file: sec/obstructions_weight2.tex
\renewcommand{\arraystretch}{1.25}
\begin{tabularx}{\textwidth}{|@{}r@{\hspace{2pt}}X@{}|}
	\hline
	\(n^0_{1,2}\) & \( = [m^0_1, m^0_{1,2}]
	- m^0_{2} \circ_1 m^0_{1,1}
	- (m^0_{2} \circ_2 m^0_{1,1}) \circ(12)
	+ m^0_{1,1} \circ_2 m^0_2 \) \\

	\hline
	\,\(n^0_{1,1,1}\) & \( =
	[m^0_1, m^0_{1,1,1}] - \big(m^0_{1,1} \circ_1 m^0_{1,1}\big) \circ \big(\id + (123) + (132)\big) \) \\

	\hline
	\(n^1_{1,1}\) & \( =
	[m^0_1, m^1_{1,1}]
	- [m^1_1, m^0_{1,1}] \) \\

	\hline
	\(n^1_2\) & \( =
	[m^0_1, m^1_2]
	+ [m^1_1, m_2^0] 
	- m^0_{1,1} \) \\

	\hline
	\(n^2_1\) & \( =
	[m^0_1, m^2_1]
	+ m^1_1 \circ m^1_1 \) \\
	\hline
\end{tabularx}

%% file: sec/obstructions_weight3.tex
\renewcommand{\arraystretch}{1.25}
\begin{tabularx}{\textwidth}{|@{}r@{\hspace{2pt}}X@{}|}
	\hline
	\(n^0_{2,2}\) & \( = \displaystyle
	[m^0_1, m^0_{2,2}]
	+ m^0_{1,2} \circ_1 m^0_2
	- m^0_{2,1} \circ_3 m^0_2
	+ (m^0_2 \circ_2 m^0_{1,2}) \circ (12)
	- m^0_2 \circ_2 m^0_{1,2}
	+ (m_2^0 \circ_2 m^0_{2,1}) \circ (123)
	- m_2^0 \circ_1 m^0_{2,1}
	+ m^0_3 \circ_2 m^0_{1,1}
	- (m^0_3 \circ_3 m^0_{1,1}) \circ (123)
	+ (m^0_3 \circ_3 m^0_{1,1}) \circ (23)
	+ (m^0_3 \circ_2 m^0_{1,1}) \circ (1243)
	- (m^0_3 \circ_1 m^0_{1,1}) \circ (23)
	+ (m^0_3 \circ_1 m^0_{1,1}) \circ (234) \) \\

	\hline
	\(n^0_{1,3}\) & \( = \displaystyle
	[m^0_1, m^0_{1,3}]
	- (m_2^0 \circ_2 m^0_{1,2}) \circ (12)
	+ m_2^0 \circ_1 m^0_{1,2}
	+ m_{1,2} \circ_2 m_2^0
	- m_{1,2} \circ_3 m_2^0 \) \\ & \( \quad
	+ m_{1,1}^0 \circ_2 m^0_3
	+ m_3^0 \circ_1 m^0_{1,1}
	+ (m_3^0 \circ_2 m^0_{1,1}) \circ (12)
	+ (m_3^0 \circ_3 m^0_{1,1}) \circ (321) \) \\

	\hline
	\(n^0_{1,1,2}\) & \( = \displaystyle
	[m^0_1, m^0_{1,1,2}]
	- m^0_{1,2}\circ_1 m^0_{1,1}
	- m^0_{1,2}\circ_2 m^0_{1,1}
	- \left(m^0_{1,2}\circ_3 m^0_{1,1}\right) \circ (23)
	- \left(m^0_{1,2}\circ_2 m^0_{1,1}\right) \circ (12)
	- \left(m^0_{1,2}\circ_3 m^0_{1,1}\right) \circ (321)
	+ m^0_{1,1}\circ_2 m^0_{1,2}
	+ \left(m^0_{1,1}\circ_2 m^0_{1,2}\right) \circ (12)
	- m^0_{2}\circ_1 m^0_{1,1,1}
	- \left(m^0_{2}\circ_1 m^0_{1,1,1}\right) \circ (34)
	+ m^0_{1,1,1}\circ_3 m^0_{2} \) \\

	\hline
	\,\(n^0_{1,1,1,1}\) & \( = \displaystyle
	[m^0_1, m^0_{1,1,1,1}]
	+ (m^0_{1,1} \circ_1 m^0_{1,1,1}) \circ \big(\id + (1234) + (13)(24) + (4321)\big) \) \\ & \( \quad
	+ (m^0_{1,1,1} \circ_1 m^0_{1,1}) \circ \big(\id + (23) + (234) + (123) + (1342) + (13)(24)\big) \) \\

	\hline
	\(n^1_{3}\) & \( = \displaystyle
	[m^0_1, m^1_{3}]
	+ [m_1^1, m_3^0]
	+ m_2^1 \circ_1 m_2^0 - m_2^1 \circ_2 m_2^0
	+ m_2^0 \circ_1 m_2^1 - m_2^0 \circ_2 m_2^1
	- m^0_{1,2} - m^0_{2,1} \) \\

	\hline
	\(n^1_{1,2}\) & \( = \displaystyle
	[m^0_1, m^1_{1,2}]
	+ [m^1_1, m^0_{1,2}]
	- m^0_{2} \circ_1 m^1_{1,1}
	- (m^0_{2} \circ_2 m^1_{1,1}) \circ (12)
	+ m^0_{1,1} \circ_2 m^1_2 \) \\ & \( \quad
	- m^1_{2} \circ_1 m^0_{1,1}
	- (m^1_{2} \circ_2 m^0_{1,1}) \circ (12)
	+ m^0_{1,1} \circ_2 m^1_2
	- m^0_{1,1,1} \) \\

	\hline
	\(n^1_{1,1,1}\) & \( = \displaystyle
	[m^0_1, m^1_{1,1,1}]
	+ [m^1_1, m^0_{1,1,1}]
	+ (m^0_{1,1} \circ_1 m^1_{1,1}) \circ \big(\id + (123) + (321)\big) \) \\ & \( \quad
	+ (m^1_{1,1} \circ_1 m^0_{1,1}) \circ \big(\id + (123) + (321)\big) \) \\

	\hline
	\(n^2_{1,1}\) & \( = \displaystyle
	[m^0_1, m^2_{1,1}]
	+ [m_1^1, m^1_{1,1}]
	+ [m^2_1, m^0_{1,1}] \) \\

	\hline
	\(n^2_{2}\) & \( = \displaystyle
	[m^0_1, m^2_{2}]
	+ [m^1_1, m_2^1]
	+ [m^2_1, m_2^0]
	- m^1_{1,1} \) \\

	\hline
	\(n^3_{1}\) & \( = \displaystyle
	[m^0_1, m^3_{1}]
	+ [m_1^1, m_1^2] \) \\

	\hline
\end{tabularx}

%% file: sec/operads.tex

\section{The operad \(\cBV_{\!\infty}\)}\label{sec:operads}

In this section we lift the discussion from algebras to the operadic level.
We study the operads \(\Com\), \(\BV\), and \(\cBV\), which govern commutative, Batalin--Vilkovisky, and coexact Batalin--Vilkovisky algebras respectively.
We then extend the theory of Koszul duality to \emph{differential graded inhomogeneous quadratic} presentations, a generalization that permits the construction of a cofibrant replacement \(\cBV_{\!\infty}\) of \(\cBV\).
We verify that algebras over \(\cBV_{\!\infty}\) coincide with the \(\cBV_{\!\infty}\)-algebras described in the previous section and establish the existence of a universal formula for their tensor product.
Since \(\cBV_{\!\infty}\) arises as a Koszul resolution, the standard tools from operadic calculus endow \(\cBV_{\!\infty}\)-algebras with an explicit deformation and obstruction theory, \(\infty\)-morphisms, as well as homotopy transfer and rectification theorems.

\subsection{Conventions}

In the previous section we worked with cochain complexes over a field \(\KK\) of characteristic \(0\) (\cref{ss:dgVec}).
Here we work in the category of chain complexes over \(\KK\), where differentials have degree \(-1\).
A cochain complex \((C^\bullet, d)\) can be regarded as a chain complex via \(C_n = C^{-n}\).

We follow the operadic conventions of \cite{LodayVallette12} and work in Hinich's model category of differential non-negatively graded operads \cite{Hinich97}, which we refer to simply as operads.
In this model structure, fibrations are surjective morphisms of operads and weak equivalences are aritywise quasi-isomorphisms.

\input{sec/strict_operads}
\input{sec/exact}
\input{sec/koszul_duals}
\input{sec/cofibrant_operads}
\input{sec/cbv}
\input{sec/hopf}
\input{sec/homotopy_fibre}
\input{sec/homotopy_algebras}
\input{sec/homotopical_tools}

%% file: sec/strict_operads.tex

\subsection{Operads for strict algebras}\label{ss:strict_operads}

An operad \(\mathrm{O}\) is said to be \defn{presented} by data \((E, R)\), where \(E\) is an \(\Sy\)-module and \(R\) is an \(\Sy\)-submodule of the free operad \(\cT(E)\), if
\[
\mathrm{O} \cong \cP(E, R) \defeq \frac{\cT(E)}{(R)},
\]
with \((R)\) the operadic ideal generated by \(R\).

\subsubsection{The operad \(\Com\)}\label{def:com_operad}

We recall the standard quadratic presentation \((V,S)\) of the operad governing (dg) commutative algebras.
Here \defn{quadratic} means that \(V\) is an \(\Sy\)-module concentrated in homological degree \(0\) with trivial differential, and that \(S\) is an \(\Sy\)-submodule of \(\cT(V)^{(2)}\), the part of the free operad \(\cT(V)\) spanned by all compositions involving exactly two generators.
As we review in \cref{ss:koszul_duals}, quadratic presentations are the classical setting where Koszul duality can be applied.

\medskip\noindent The operad \defn{\(\Com\)} is presented by the following quadratic data \((V, S)\), where the symmetric action on the generator is trivial:

\medskip\noindent \(\diamond\) \(V \defeq \KK\set{\m}\) where \(\m \defeq \M\) with degree \(\bars{\m}=0\),

\medskip\noindent \(\diamond\) \(S\) is the \(\Sy\)-module spanned by the relator
\vspace*{-7pt}
\[
\textsc{associativity: } \LL{1}{2}{3} - \RR{1}{2}{3},
\]
that is
\[
\Com \defeq \frac{\mathcal{T}(\m)}{\big(\text{associativity}\big)}.
\]

\subsubsection{The operad \(\BV\)}\label{def:OperadBV}

A straightforward definition of the operad \defn{\(\BV\)} governing (dg) \(\BV\)-algebras is
\[
\BV \defeq \frac{\Com \vee \mathcal{T}(\triangle)}{\big(\triangle \circ \triangle,\ \triangle \ \text{second-order}\big)}
\]
or, explicitly,
\[
\BV \cong \frac{\mathcal{T}(\m, \triangle)}{\big(\text{associativity}, \triangle \circ \triangle,\ \triangle \ \text{second-order}\big)},
\]
where the generators are in degree \(0\) and the differential is trivial.

\medskip Unfortunately, the second-order relation involves the composition of \textit{three} generators, making the associated presentation unfit for Koszul duality techniques.
To find a better suited presentation, recall that in a \(\BV\)-algebra \((A, d, \m, \triangle)\) the failure of \(\triangle\) to be a derivation of the product defines a shifted Lie bracket \(\b \defeq [\triangle, \m]\) for which the Leibniz relation with respect to \(\m\) is equivalent to the second-order relation for \(\triangle\).
This leads to the following alternative presentation, which is \defn{inhomogeneous quadratic}: here the generating \(\Sy\)-module has trivial differential and the relators lie in \(\cT(E)^{(1)} \oplus \cT(E)^{(2)}\), that is, a mixture of linear and quadratic terms in the free operad \(\cT(E)\).

\begin{lemma}\label{lemma:presentationBV}
	The operad \(\BV\) is presented by \((E, R)\), where the symmetric action on all generators is trivial:

	\medskip\noindent \(\diamond\) \(E \defeq \KK\set{\m} \oplus \KK\set{\b} \oplus \KK\set{\triangle}\) where
		\[
		\m \defeq \M,\quad \b \defeq \B,\quad \triangle \defeq \D~,
		\]
		with \(|\m| = 0\) and \(|\triangle| = |\b| = 1\).

	\medskip\noindent \(\diamond\) \(R\) is spanned, as an \(\Sy\)-module, by the following 6 relations:

	\noindent
	\begin{tabular}{ll}
		\textnormal{\textsc{associativity:}} \(\LL{1}{2}{3} - \RR{1}{2}{3}\) &
		\textnormal{\textsc{square-zero:\quad}} \(\DD\) \\
		\textnormal{\textsc{bracket:}} \(\BB{1}{2} - \hspace*{-5pt} \DM{1}{2} \hspace*{-5pt} + \MDL{1}{2} + \MDR{1}{2}\) &
		\textnormal{\textsc{leibniz:}} \(\LBM{1}{2}{3}- \hspace*{-5pt} \LMB{1}{3}{2}-\RMB{1}{2}{3}\)\\
		\textnormal{\textsc{jacobi:}} \(\BBB{1}{2}{3}+ \hspace*{-5pt} \BBB{2}{3}{1}+ \hspace*{-5pt} \BBB{3}{1}{2}\) &
		\textnormal{\textsc{derivation:}} \(\DB{1}{2} - \BDL{1}{2} - \BDR{1}{2}\)~.
	\end{tabular}
\end{lemma}

\begin{proof}
	Leibniz relation for \(\b\) is equivalent to the second-order relation for \(\triangle\), while the Jacobi relation and the derivation relation are direct consequences of the other relations.
\end{proof}

\begin{remark}
	As noted above, the Jacobi and derivation relations follow from the other relations, but we include them explicitly since a presentation of this form is required for applying the theory of Koszul duality to the operad \(\BV\), see \cref{ex:qlCondition}.
\end{remark}

\subsubsection{The operad \(\cBV\)}

The operad \defn{\(\cBV\)} governing \(\cBV\)-algebras is
\[
\BV \defeq \dfrac{\BV \vee \mathcal{T}(\square)}
{\big([\triangle, \square], \ \square \ \text{second-order}\big)}
\]
or, explicitly,
\[
\cBV \defeq \frac{\mathcal{T}(\m, \triangle, \square)}
{\big(\text{associativity}, \triangle \circ \triangle,\ [\triangle, \square],\ \triangle\text{ second-order},\ \square\text{ second-order}\big)},
\]
where the differential is the unique derivation extending
\[
\dif\, \triangle \defeq \square,
\qquad
\dif\, \m \defeq 0.
\]


As for the first presentation given for the operad \(\BV\), this presentation of \(\cBV\) is not well suited to our purposes.
We introduce new generators corresponding to \([\triangle,\m]\) and \([\square,\m]\) and describe the operad \(\cBV\) via a \defn{differential graded inhomogeneous quadratic} presentation, i.e., one in which the generating \(\Sy\)-module \(E\) may carry a non-trivial differential and, as before, the relators lie in \(\cT(E)^{(1)} \oplus \cT(E)^{(2)}\).

\begin{lemma}\label{lemma:presentationcBV}
	The operad \(\cBV\) is presented by \((E_\bullet, R_\bullet)\), where the symmetric action on all generators is trivial:

	\medskip\noindent \(\diamond\) \(E_\bullet \defeq \big(\KK\set{\m} \oplus \KK\set{\b} \oplus \KK\set{\triangle} \oplus \KK\set{\c} \oplus \KK\set{\square},\ \dif\big)\) where
	\[
	\m \defeq \M,\quad
	\b \defeq \B,\quad
	\triangle \defeq \D,\quad
	\c \defeq \C,\quad
	\square \defeq \BOX,
	\]
	with \(|m| = |c| = |\square| = 0\) and \(|b| = |\triangle| = 1\), and differential
	\[
	\dif \colon \M \mapsto 0,\quad
	\B \mapsto \C \mapsto 0,\quad
	\D \mapsto \square \mapsto 0.
	\]

	\medskip\noindent \(\diamond\) \(R_\bullet\) is spanned as an \(\Sy\)-module by the following 11 types of relators.

	\noindent
	\begin{tabular}{ll}
		\textnormal{\textsc{associativity:}} \(\LL{1}{2}{3} - \RR{1}{2}{3}\) &
		\textnormal{\textsc{square-zero:\quad}} \(\DD\) \\[4pt]

		\textnormal{\textsc{bracket \(\triangle\):}} \(\BB{1}{2} - \hspace*{-5pt}\DM{1}{2}\hspace*{-5pt} + \MDL{1}{2} + \MDR{1}{2}\) &
		\textnormal{\textsc{leibniz \(\triangle\):}} \(\hspace*{-5pt}\LBM{1}{2}{3}-\hspace*{-5pt}\LMB{1}{3}{2}-\RMB{1}{2}{3}\) \\[4pt]

		\textnormal{\textsc{bracket \(\square\):}} \(\CC{1}{2} - \hspace*{-5pt}\BoxM{1}{2}\hspace*{-5pt} + \MBoxL{1}{2} + \MBoxR{1}{2}\) &
		\textnormal{\textsc{leibniz \(\square\):}} \(\hspace*{-5pt}\LBoxM{1}{2}{3}-\hspace*{-5pt}\LMBox{1}{3}{2}-\RMBox{1}{2}{3}\) \\[4pt]

		\textnormal{\textsc{jacobi:}} \(\hspace*{-5pt}\BBB{1}{2}{3}+\hspace*{-5pt}\BBB{2}{3}{1}+\hspace*{-5pt}\BBB{3}{1}{2}\) &
		\textnormal{\textsc{derivation:}} \(\hspace*{-5pt}\DB{1}{2} - \BDL{1}{2} - \BDR{1}{2}\)
	\end{tabular}

	\noindent
	\begin{tabular}{l}
		\textnormal{\textsc{compatibility 1:}} \hspace{6pt}\(\DBox - \BoxD\) \\[4pt]
		\textnormal{\textsc{compatibility 2:}} \(\CB{1}{2}{3}+ \hspace*{-5pt}\CB{2}{3}{1}+ \hspace*{-5pt}\CB{3}{1}{2}- \hspace*{-5pt}\BC{1}{2}{3}- \hspace*{-5pt}\BC{2}{3}{1}- \hspace*{-5pt}\BC{3}{1}{2}\) \\[4pt]
		\textnormal{\textsc{compatibility 3:}} \(\BoxB{1}{2} \hspace*{-5pt}-\hspace*{-5pt} \DC{1}{2} \hspace*{-5pt}- \CDL{1}{2} +\BBoxL{1}{2}-\CDR{1}{2}+\BBoxR{1}{2}\)~.
	\end{tabular}
\end{lemma}

\begin{proof}
	First one can see that both \(E_\bullet\) and \(R_\bullet\) are stable under the differential \(\dif\), so this differential graded quadratic-linear data is well defined.
	The rest of the proof is similar to the one for the operad \(\BV\):
	the Leibniz relation for \(\c\) is equivalent to the second-order relation for \(\square\) and the compatibility relations 2 and 3 are direct consequences of the other relations.
\end{proof}

\begin{remark}
	Like for the operad \(\BV\), the Jacobi relation, the derivation relation, and the compatibility relations 2 and 3 are not mandatory to obtain an equivalent presentation for the operad \(\cBV\), but we need to consider all of them in order to get a presentation suitable for the theory of Koszul duality developed in \cref{subsubsec:QLcBV}.
\end{remark}

\begin{remark}
	Notice that the ``bracket'' \(\c\) induced by the second-order operator \(\square\) need not satisfy the Jacobi relation, since \(\square\) does not square to zero.
	For the same reason, \(\square\) is not, in general, a derivation of the bracket \(\c\).
\end{remark}

%% file: sec/exact.tex
\subsection{Exact and coexact \(\BV\)-structures}\label{ss:exact}

This subsection is not part of the logical dependency structure of the paper and can be safely skipped on a first reading.
Its purpose is to explain our choice of terminology by comparing the operad \(\cBV\) to the operad \(\eBV\) controlling \textit{exact} BV-algebras.

\subsubsection{Exact \(\BV\)-algebras and Poisson manifolds}

Recall that in a (dg) \(\BV\)-algebra \((A, d, \cdot, \triangle)\), the operators \(d\) and \(\triangle\) commute; that is, \([d, \triangle] = 0\).
This means that \(\triangle\) is \emph{closed} in \(\End(A)\).
The operator \(\triangle\) is said to be \emph{exact} if there exists a linear map \(\nabla \colon A \to A\) such that \(\triangle = [d, \nabla]\).
As usual, exactness implies closedness.
Therefore, the following notion, due to Guan--Muro \cite[Definition~4.5]{GuanMuro23} and of relevance to Poisson geometry, provides a refinement of the concept of \(\BV\)-algebra.

\begin{definition}
	An \defn{exact \(\BV\)-algebra} is a tuple \((A, d, \cdot, \nabla)\), where \((A, d, \cdot)\) is a dg commutative algebra and \(\nabla \colon A \to A\) is a linear map of degree \(-2\) and order at most \(2\), satisfying
	\begin{equation}\label{Eq:Nabla}
		[\nabla,[\nabla,d]] = 0.
	\end{equation}
\end{definition}

The following proposition, proven straighforwardly using Cartan calculus, shows that Poisson manifolds canonically define exact \(\BV\)-algebras.

\begin{proposition}[{\cite[Section~3]{Koszul85}}]\label{prop:eBVPoisson}
	Let \(M\) be a manifold with a section \(\pi \in \Gamma(\Lambda^2 TM)\) such that \([\pi, \pi] = 0\).
	Then the de Rham complex \(\Omega^{\bullet}(M)\) admits an exact \(\BV\)-algebra structure
	\[
	\big(\Omega^{\bullet}(M), d_{dR}, \wedge, \iota_\pi\big),
	\]
	where \(\iota_\pi\) denotes contraction with \(\pi\).
\end{proposition}

\medskip\noindent
Exact \(\BV\)-algebras are algebras over the (dg) operad
\[
\eBV \defeq \dfrac{\BV \vee \mathcal{T}(\nabla)}
{\big([\triangle, \nabla], \
	\nabla \ \text{second-order}\big)},
\]
whose differential is determined by
\begin{equation}\label{eq:exactness}
	\dif \, \nabla \defeq \triangle.
\end{equation}
In contrast to \(\BV\)-algebras, the operator \(\triangle\) is always \(0\) in the homology of \(\eBV\)-algebras.
As we will see next, this is also the case in the homology of \(\cBV\)-algebras that are not \(\BV\)-algebras, not because \(\triangle\) is exact, but because it is not closed.

\subsubsection{Coexact \(\BV\)-algebras and pseudo-Riemannian manifolds}

Recall that the operad \(\cBV\) is given by
\[
\dfrac{\BV \vee \mathcal{T}(\square)}
{\big([\triangle, \square], \ \square \ \text{second-order}\big)},
\]
with differential determined by
\[
\dif \, \triangle = \square.
\]
Here the generator \(\triangle\) is \emph{coexact}, explicitly \(\dif^*(\square) = \triangle\), where \(\dif^*\) is the adjoint of \(\dif\) with respect to the canonical inner product.\footnote{
	Explicitly, this inner product is defined by the canonical basis of \(\cBV(1)\), which in degree \(\bars{\triangle}\) is \(\set{\mathrm u_{m} =\triangle \circ \square^{\circ m} \mid m \geqslant 0}\).
	We can verify the claim \(\dif^*(\square) = \triangle\) by writing \(\dif^*(\square) = \sum_{m} \alpha_{m} \cdot \mathrm u_{m}\) with \(\alpha_{m} = \angles{\mathrm u_{m}, \dif^*(\square)} = \angles{\dif(\mathrm u_{m}), \square} = \angles{\square^{\circ m+1}, \square}\), and concluding that
	\[
	\alpha_{m} =
	\begin{cases}
		1 & m = 0,\\
		0 & \text{otherwise}.
	\end{cases}
	\]}
This motivates the terminology \defn{coexact \(\BV\)-algebras} and the notation \(\cBV\) for the operad controlling them.
We remind the reader that \(\BV\)-algebras are precisely coexact dg \(\BV\)-algebras where the operator corresponding to \(\square\), the obstruction \(n\), vanishes.

We have the following analogue of \cref{prop:eBVPoisson} stating that pseudo-Riemannian manifolds canonically define coexact \(\BV\)-algebras.

\begin{proposition}\label{prop:deRhamCoexact}
	Let \(M\) be a manifold equipped with a nowhere-degenerate section \(g \in \Gamma(S^2 T^*M)\).
	Then the de Rham complex \(\Omega^{\bullet}(M)\) admits a canonical coexact \(\BV\)-algebra structure
	\[
	\big(\Omega^{\bullet}(M), d_{dR}, \wedge, d^\star\big),
	\]
	where \(d^\star\) is the Hodge codifferential.
\end{proposition}

The parallel, or duality, observed here between Poisson and pseudo-Riemannian manifolds is best understood in the language of supergeometry, although we do not develop this perspective further.

\subsubsection{Relationship between \(\BV\)-type operads}

We conclude this subsection making explicit the relationship between the operads we have discussed so far.

\begin{theorem}\label{thm:Homology}
	The canonical maps induced by the identification of generators define the commutative diagram
	\[
	\begin{tikzcd}[row sep=small]
		\cBV \arrow[r, two heads] \arrow[rr, bend left=35, "\sim"] &
		\BV \arrow[r, hook] &
		\eBV \\
		& \Com \arrow[ul, hook', "\sim"', bend left] \arrow[ur, hook, "\sim", bend right] &
	\end{tikzcd}
	\]
	where \(\hookrightarrow\) denotes an injection, \(\twoheadrightarrow\) a surjection, and \(\xrightarrow{\sim}\) a quasi-isomorphism.
\end{theorem}

\begin{proof}
	The inclusion \(\Com \hookrightarrow \eBV\) is a quasi-isomorphism by \cite[Theorem~1.3]{GuanMuro23}.
	To analyze \(\Com \hookrightarrow \cBV\) we adapt the method of \cite[Section~2]{DrummondColeVallette13}, and we note that this argument also yields an alternative proof of the acyclicity of \(\Com \hookrightarrow \eBV\).

	The operad \(\cBV\) carries a weight grading defined by declaring that the total number of occurrences of \(\triangle\) and \(\square\) is the weight.
	Both the relations and the differential are homogeneous for this grading.
	Let \(\cBV^{[k]}\) denote the weight-\(k\) summand.
	Then \(\cBV^{[0]} \cong \Com\).

	Consider the assignment
	\[
	\cdot \mapsto 0,
	\qquad
	\triangle \mapsto 0,
	\qquad
	\square \mapsto \triangle.
	\]
	This is a morphism of quadratic data and therefore extends uniquely to a derivation \(h \colon \cBV \to \cBV\) preserving the weight.
	Define a degree \(1\) morphism of \(\Sy\)-modules \(H \colon \cBV \to \cBV\) by
	\[
	H \mid_{\cBV^{[0]}} \coloneqq 0
	\quad \text{and}\quad 
	H \mid_{\cBV^{[k]}} \coloneqq \frac{1}{k} \, h \mid_{\cBV^{[k]}} \ \ \text{for} \ \ 
	k \geqslant 1.
	\]

	Let
	\[
	\begin{tikzcd}[column sep=small]
		i \colon \Com \arrow[r, hook] & \cBV
	\end{tikzcd}
	\qquad\text{and}\qquad
	\begin{tikzcd}[column sep=small]
		p \colon \cBV \arrow[r, two heads] & \Com
	\end{tikzcd}
	\]
	denote the canonical inclusion and projection.
	Then \((i,p,H)\) is a deformation retract in the category of dg \(\Sy\)-modules, namely
	\[
	H \dif + \dif H
	=
	\id_{\cBV} - i p.
	\]

	The operator \(H \dif + \dif H\) is determined by its effect on generators:
	\[
	\cdot \mapsto 0,
	\qquad
	\triangle \mapsto \triangle,
	\qquad
	\square \mapsto \square.
	\]
	Thus \(i\) is a quasi-isomorphism.
	Finally, the 2-out-of-3 property of weak equivalences concludes the proof.
\end{proof}

\begin{remark}
	\cref{thm:Homology} shows that, for a pseudo-Riemannian manifold, the coexact BV-algebra structure on differential forms yields no homotopical invariants on de~Rham cohomology beyond the Massey products induced by the wedge product.
	The corresponding statement for the exact BV-algebra structure in the Poisson case was made in \cite{GuanMuro23}.
\end{remark}

%% file: sec/koszul_duals.tex

\subsection{Koszul duality}\label{ss:koszul_duals}


Let \(E\) be an \(\Sy\)-module and let \(\cT^{c}(E)\) denote the cofree conilpotent cooperad on \(E\).
Let \(R \subset \cT^{c}(E)\) be an \(\Sy\)-submodule.
We define \defn{\(\cC(E; R)\)} to be the smallest subcooperad of \(\cT^{c}(E)\) containing \(E\) and \(R\).

\subsubsection{Homogeneous quadratic case}

The \defn{Koszul dual cooperad} of an operad \(\mathrm{O} \cong \P(E, R)\) presented by quadratic data \((E, R)\) is
\[
\mathrm{O}^\ac \defeq \cC(sE; s^{2}R),
\]
where \(s\) denotes the homological suspension of \(\Sy\)-modules.
The \defn{Koszul dual operad} is the shifted linear dual of the cooperadic suspension of \(\mathrm{O}^\ac\):
\begin{equation}\label{eq:koszul_dual_operad}
	\mathrm{O}^! \defeq \mathrm{S} \otimes_H \big(\mathrm{O}^{\ac}\big)^*,
\end{equation}
where \(\mathrm{S} \defeq \mathrm{End}_{\KK s^{-1}}\) is the suspension operad.

\begin{example}\label{ex:Comac}
	The Koszul dual cooperad and operad associated to the quadratic presentation \((V, S)\) of the operad \(\Com\) given in \cref{def:com_operad} are
	\[
	\Com^\ac \cong \Lie^c_1
	\quad\text{and}\quad
	\Com^! \cong \Lie,
	\]
	where \(\Lie^c_1\) stands for the cooperad encoding Lie coalgebras with cobrackets of degree \(1\) and \(\Lie\) is the operad encoding Lie algebras.
\end{example}

\subsubsection{Inhomogeneous quadratic case}

To deal with operads \(\mathrm{O} \cong \P(E,R)\) presented by inhomogeneous quadratic data, like the operad \(\BV\), one considers the projection \(\q \colon \cT(E) \to \cT(E)^{(2)}\) and defines its \defn{analogue quadratic data} as \((E, \q R)\).
The \defn{quadratic analogue operad} of \(\rO\) is \(\q\mathrm{O} \defeq \P(E,\q R)\).

\begin{definition}
	The \defn{quadratic-linear conditions} for a quadratic-linear presentation \((E, R)\) are
	\begin{description}
		\item[\rm (\(ql_1\))] \(R \cap E = 0\)~,
		\item[\rm (\(ql_2\))] \(\big\{R \circ_{(1)} E + E \circ_{(1)} R\big\} \cap \cT(E)^{(2)} \subset R \cap \cT(E)^{(2)} = 0\).
	\end{description}
\end{definition}

Intuitively, these conditions express respectively the minimality of the generators and the maximality of the weight two part of the relations.

\begin{lemma}[{\cite[Lemma~37]{GCTV12}}]\leavevmode
	\begin{enumerate}
		\item Condition~(\(ql_1\)) implies that the space of relations \(R\) can be written as the graph of a map \(\varphi \colon \q R \to E\).
		\item Condition~\((ql_2)\) implies that there exists a square-zero coderivation \(d_\varphi\) on the Koszul dual cooperad of the quadratic analogue presentation extending
		\[
		\cC\big(sE, s^2 \q R\big) \to s^2 \q R \xra{s^{-1}\varphi} sE.
		\]
	\end{enumerate}
\end{lemma}

The \defn{Koszul dual cooperad} of an operad \(\mathrm{O} \cong \P(E, R)\) presented by an inhomogeneous quadratic data satisfying the quadratic-linear conditions is
\[
\mathrm{O}^\ac \defeq \big(\q \mathrm{O}^\ac, d_\varphi\big) = \big(\cC(sE, s^2 \q R), d_\varphi\big).
\]
The \defn{Koszul dual operad} \(\mathrm{O}^!\) is defined as in the homogeneous case \eqref{eq:koszul_dual_operad}.

\begin{example}\label{ex:qlCondition}
	The inhomogeneous quadratic presentation of the operad \(\BV\) given in \cref{lemma:presentationBV}
	satisfies the quadratic-linear conditions: the Jacobi relation and the derivation relation were included in this presentation in order to satisfy Condition~\((ql_2)\).
	In this case, the map
	\(\varphi \colon \q R \to E\) is given by
	\[
	\DM{1}{2} - \MDL{1}{2} - \MDR{1}{2} \mapsto \BB{1}{2}
	\]
	and by \(0\) otherwise.
\end{example}

\subsubsection{Inhomogeneous dg quadratic case}

In order to work with the operad \(\cBV\), we need to develop a Koszul duality theory for \emph{differential graded} quadratic-linear data \((E_\bullet, R_\bullet)\).
If this data satisfies the quadratic-linear conditions, then there exists a chain map \(\psi \colon (\q R_\bullet, \partial) \to (E_\bullet, \partial)\), whose graph coincides with the space \(R_\bullet\) of quadratic-linear relations.
This chain map induces a square-zero coderivation \(d_\psi\) on the Koszul dual cooperad of the quadratic analogue
\(\cC(sE_\bullet, s^2 \q R_\bullet)\),
which extends the map
\[
\cC(sE_\bullet, s^2 \q R_\bullet) \twoheadrightarrow s^2 \q R_\bullet \xra{s^{-1}\psi} sE_\bullet.
\]
The internal differential \(\partial\) of the dg quadratic-linear data \(\big(E_\bullet, R_\bullet\big)\) also induces a square-zero coderivation \(d_1\) on \(\cC(sE_\bullet, s^2 \q R_\bullet)\).

\begin{lemma}\label{lemma:commutingcodiff}
	For any dg quadratic-linear data \(\big(E_\bullet, R_\bullet\big)\) satisfying the quadratic-linear conditions,
	the two codifferentials \(d_1\) and \(d_\psi\) of the cooperad \(\cC(sE_\bullet, s^2 \q R_\bullet)\) anti-commute.
\end{lemma}

\begin{proof}
	The commutator \([d_1, d_\psi] \defeq d_1 d_\psi + d_\psi d_1\) is a coderivation of the quadratic
	cooperad \(\cC\big(sE_\bullet, \allowbreak s^2 \q R_\bullet\big)\), therefore it is completely characterised by its projection onto the
	space \(sE_\bullet\) of cogenerators.
	The coderivation \(d_1\) preserves the weight and the coderivation \(d_\psi\) lowers the weight by \(1\), so it is enough to compute
	\([d_1, d_\psi]= d_1 d_\psi + d_\psi d_1\) on \(s^2\q R_\bullet\).
	This relation is a direct consequence of the fact that the map \(\psi\) preserves the differential induced by \(\partial\) since it coincides with the following commutative diagram:
	\[
	\begin{tikzcd}[column sep=huge, row sep=large]
		s^2\q R_\bullet \arrow[r, "d_\psi=s^{-1}\psi"] \arrow[d, "d_1=s^2\partial"'] & s E_\bullet\ \arrow[d, "-d_1=s\partial"]\\
		s^2\q R_\bullet \arrow[r, "d_\psi=s^{-1}\psi"] & s E_\bullet.
	\end{tikzcd}
	\]
	This concludes the proof.
\end{proof}

\cref{lemma:commutingcodiff} allows us to introduce the following.

\begin{definition}
	The \defn{Koszul dual cooperad} of an operad \(\mathrm{O} \cong \P(E_\bullet, R_\bullet)\) presented by inhomogeneous dg quadratic data satisfying the quadratic-linear conditions is
	\[
	\mathrm{O}^\ac \defeq \big(\q \mathrm{O}^\ac, d_1+d_\psi\big) = \big(\cC(sE_\bullet, s^2 \q R_\bullet), d_1+d_\psi\big).
	\]
	The \defn{Koszul dual operad} \(\mathrm{O}^!\) is defined by the same formula \eqref{eq:koszul_dual_operad} as in the homogeneous case.
\end{definition}

%% file: sec/cofibrant_operads.tex

\subsubsection{Koszul resolution}\label{ss:cofibrant_operads}

Homotopy \(\rO\)-algebras are defined as algebras over a \defn{cofibrant replacement} of the operad \(\rO\), that is, a cofibrant operad \(\rQ\rO\) with an acyclic fibration
\(
\begin{tikzcd}[column sep=15pt]
	\rQ\rO \arrow[r, ->>, "\sim\ "] & \rO.
\end{tikzcd}
\)
A \defn{Koszul resolution} is a cofibrant replacement of the specific form
\begin{equation}\label{eq:Omega Oac->O}
	\begin{tikzcd}[column sep=22pt]
		\Omega\rO^{\ac} \arrow[r, ->>, "\sim\ "] & \rO,
	\end{tikzcd}
\end{equation}
where \(\rO^{\ac}\) is the Koszul dual cooperad associated to a chosen presentation of \(\rO\), \(\Omega\) is the cobar construction, and the morphism \eqref{eq:Omega Oac->O} is induced by the canonical \textit{twisting morphism} \(\rO^{\ac} \to \rO\).

\subsubsection{Twisting morphisms}

A \defn{twisting morphism} between a cooperad \(\rC\) and a operad \(\rO\) is a degree \(-1\) map of \(\Sy\)-modules \(\tau \colon \rC \to \rO\) satisfying the Maurer--Cartan equation in the convolution dg pre-Lie algebra \(\Hom_{\Sy}(\rC,\rO)\):
\[
\partial_{\rO} \tau + \tau \partial_{\rC} + \tau \star \tau = 0.
\]

\begin{lemma}\label{lemma:KDdgInhomo}
	Let \(\big(E_\bullet, R_\bullet\big)\) be dg quadratic-linear data satisfying the quadratic-linear conditions.
	Then, the canonical \(\Sy\)-module map
	\[
	\kappa \colon \mathcal{C}\big(sE_\bullet, s^2 qR_\bullet\big) \twoheadrightarrow sE_\bullet
	\xrightarrow{s^{-1}} E_\bullet \to \mathcal{P}\big(E_\bullet, R_\bullet\big)
	\]
	is a twisting morphism.
	Additionally, \(\kappa\) induces a twisting morphism
	\[
	\bar{\kappa} \colon \big(H_\bullet(\rO^{\ac}, d_1), \bar{d}_\psi\big) \to \big(H_\bullet(\rO, \partial), 0\big).
	\]
	where \(\bar{d}_\psi\) is the differential in the \(d_1\)-homology of \(\rO^\ac\) induced by \(d_\psi\).
\end{lemma}

\begin{proof}
	We have to show that the assignment \(\kappa\) satisfies the Maurer--Cartan equation
	\[
	\dif \kappa + \kappa\,(d_1 + d_\psi) + \kappa \star \kappa = 0.
	\]
	By definition of the various maps, the left-hand side of this equation vanishes on elements of weight \(0\) and of weight greater or equal to \(3\).
	For elements \(sE_\bullet\) of weight one, the only possibly non-trivial terms are the first two, which are equal to
	\[
	\partial s^{-1} + s^{-1}(-\partial) = 0.
	\]
	For elements of \(s^2 qR_\bullet\) that are also in \(s^2 R_\bullet\), that is homogeneous quadratic relations, their image under \(\kappa \star \kappa\) vanishes in \(\rO\) since they are initial homogeneous quadratic relations there.
	Their image under \(\kappa d_\psi\) vanishes by definition of \(\psi\), and their image under \(\dif \kappa + \kappa d_1\) vanishes for weight reasons.

	For elements \(s^2 \rho\) of \(s^2 qR_\bullet\) that are not in \(s^2 R_\bullet\), that is inhomogeneous quadratic relations, there exists a non-trivial \(\varepsilon \in E_\bullet\) such that \(\rho + \varepsilon \in R_\bullet\).
	In this case, the image of \(s^2 \rho\) under \(\dif \kappa + \kappa d_1\) still vanishes for weight reasons.
	The remaining non-trivial part of the Maurer--Cartan equation is equal to
	\[
	\kappa\big(d_\psi(s^2\rho)\big) + (\kappa \star \kappa)(s^2\rho)
	=
	\kappa(-s\varepsilon) + \varepsilon
	=
	-\varepsilon + \varepsilon
	=
	0
	\]
	in \(\rO\).

	Finally, since the internal codifferential \(d_1\) preserves the weight grading, the map \(\bar{\kappa}\) is well defined and equal to
	\[
	\bar{\kappa} \colon H_\bullet\big(\rO^{\ac}, d_1\big) \twoheadrightarrow H_\bullet\big(s E_\bullet,-\partial\big)
	\xrightarrow{s^{-1}} H_\bullet\big(E_\bullet,\partial\big) \to H_\bullet\big(\rO, \partial\big).
	\]
	The Maurer--Cartan equation on the chain level satisfied by \(\kappa\) implies the Maurer--Cartan equation on homology satisfied by \(\bar{\kappa}\).
\end{proof}


\subsubsection{Koszul property}

A twisting morphism \(\tau \colon \rC \to \rO\) is called \defn{Koszul} when its associated operad morphism \(\Omega\rC \to \rO\), which is always a fibration, is a weak equivalence.
This morphism is the unique operad map whose restriction to the generators \(s^{-1}\overline{\rC} \subset \Omega\rC\) sends \(s^{-1}c\) to \(\tau(c)\) for each \(c \in \overline{\rC}\).

\begin{definition}
	An operad \(\rO = \mathcal{P}\big(E_\bullet, R_\bullet\big)\) given by dg quadratic-linear data satisfying the
	quadratic-linear conditions is called \defn{Koszul} when \(\bar{\kappa}\) is a Koszul morphism.
\end{definition}

\begin{theorem}\label{t:koszul_main}
	The cobar construction \(\Omega\rO^\ac\) of the Koszul dual cooperad of a Koszul operad \(\rO = \mathcal{P}\big(E_\bullet, R_\bullet\big)\) is cofibrant and the canonical map
	\[
	\begin{tikzcd}[column sep=0.7cm]
		\rO_\infty \defeq \Omega \rO^{\ac}
		\arrow[r, ->>, "\sim"] & \rO
	\end{tikzcd}
	\]
	is an acyclic fibration.
\end{theorem}

\begin{proof}
	The twisting morphism \(\kappa\) of \cref{lemma:KDdgInhomo} induces a morphism of operads
	\[
	g_\kappa \colon \rO_\infty = \Omega \rO^{\ac} \twoheadrightarrow \rO,
	\]
	which is surjective since it reaches all the generators of \(\rO\), so it is a fibration.
	Let us recall that the cobar construction of the Koszul dual cooperad is given by
	\[
	\Omega \rO^{\ac} =
	\Big(
	\mathcal{T}\big(s^{-1}\overline{\mathcal{C}}\big(sE_\bullet, s^2 qR_\bullet\big)\big),
	d_1 + d_\psi + d_2
	\Big),
	\]
	where we use the same notation for the differentials induced by \(d_1\) and \(d_\psi\), by a slight abuse of notation.
	In order to prove that the map \(g_\kappa\) is a quasi-isomorphism, we consider, on the left-hand side, the increasing filtration
	\[
	\mathcal{F}_k \Omega \rO^{\ac} \defeq \bigoplus_{l \leqslant k} \Omega_l \rO^{\ac},
	\qquad
	0 = \mathcal{F}_{-1} \subset \mathcal{F}_0 \subset \mathcal{F}_1 \subset \mathcal{F}_2 \subset \cdots \subset \mathcal{F}_k \subset \mathcal{F}_{k+1} \subset \cdots \subset \Omega\rO^{\ac}
	\]
	given by the syzygy degree, which is defined for every element of \(\Omega\rO^{\ac}\) by the sum over the vertices of the weight of the labeling elements minus \(1\), see \cite[Section~7.3]{LodayVallette12}.
	The differential \(d_1\) preserves this filtration and the two differentials \(d_\psi\) and \(d_2\) lower it by \(1\).
	On the right-hand side, we consider the trivial filtration
	\[
	0 = \mathcal{G}_{-1} \subset \mathcal{G}_0 = \rO = \mathcal{G}_1 = \cdots = \mathcal{G}_k = \cdots = \rO.
	\]
	The assignment \(g_\kappa\) obviously preserves the respective filtrations.
	On the left-hand side, the first page of the associated spectral sequence is isomorphic to
	\[
	\big(E^0 \Omega \rO^{\ac}, d^0\big) \cong \Big(\mathcal{T}\big(s^{-1}\overline{\mathcal{C}}\big(sE_\bullet,
	s^2 qR_\bullet\big)\big), d_1\Big).
	\]
	The operadic K\"unneth theorem \cite[Proposition~6.2.3]{LodayVallette12} implies that
	the second page of this spectral sequence is isomorphic to the cobar construction of the homology of the Koszul dual cooperad equipped with the codifferential \(\bar{d}_\psi\):
	\[
	\big(E^1 \Omega \rO^{\ac}, d^1\big) \cong \Omega H_\bullet\big(\rO^\ac, d_1\big).
	\]
	On the right-hand side, the first page of the associated spectral sequence is isomorphic to
	\[
	\big(E^0 \rO, d^0\big) \cong \big(\rO, \dif\big),
	\]
	so its second page is isomorphic to the homology operad equipped with trivial differential:
	\[
	\big(E^1 \rO, d^1\big) \cong \big(H_\bullet(\rO, \partial), 0\big).
	\]
	The induced map
	\[
	\begin{tikzcd}
		E^1 g_\kappa = g_{\bar\kappa} \colon \Omega H_\bullet\big(\rO^\ac, d_1\big)
		\arrow[r, "\sim"] & H_\bullet(\rO, \partial)
	\end{tikzcd}
	\]
	is the quasi-isomorphism of operads corresponding to the Koszul morphism
	\(\bar{\kappa} \colon H_\bullet\big(\rO^{\ac}, d_1\big) \to H_\bullet\big(\rO,\partial\big)\).
	Since these two filtrations are bounded below and exhaustive, their convergence shows that the morphism \(g_\kappa\) is a quasi-isomorphism by \cite[Theorem~5.2.12]{WeibelBook}.

	This Koszul type resolution
	\[
	\begin{tikzcd}
		\rO_\infty = \Omega \rO^{\ac}
		=
		\Big(
		\mathcal{T}\big(s^{-1}\overline{\mathcal{C}}\big(sE_\bullet, s^2 qR_\bullet\big)\big), d_1 + d_\psi + d_2
		\Big)
		\arrow[r, ->>, "\sim"] & \rO
	\end{tikzcd}
	\]
	is quasi-free on a space of generators non-negatively graded by the syzygy degree, so it is cofibrant.
\end{proof}

%% file: sec/cbv.tex

\subsection{Application to the operad \(\cBV\)}\label{ss:cbv}

\subsubsection{Quadratic-linear conditions for \(\cBV\)}\label{subsubsec:QLcBV}

\begin{lemma}\label{lemme:qlsforcBV}
	The inhomogeneous quadratic presentation \((E_\bullet, R_\bullet)\) of the operad \(\cBV\) satisfies the quadratic-linear conditions.
\end{lemma}

\begin{proof}
	Condition \((ql_1)\) is trivially satisfied.
	In order to establish Condition \((ql_2)\), we have to compute
	\[
	\{R_\bullet \circ_{(1)} E_\bullet + E_\bullet \circ_{(1)} R_\bullet\} \cap \cT(E_\bullet)^{(2)};
	\]
	the only way to possibly get non-trivial elements is to consider either the inhomogeneous quadratic relations \textsc{bracket} \(\triangle\) or the \textsc{bracket} \(\square\) in \(R_\bullet\) composed above or below with the generators \(\m\), \(\b\), \(\triangle\), \(\c\), or \(\square\) in \(E_\bullet\).
	This gives the following 8 possibilities that are all straightforward to check: in the table below the \emph{inhomogeneous relation} is composed with a \emph{generator} to produce a \emph{combination} which either produces \emph{no relation}, is equal to a \emph{homogeneous relation}, or is a \emph{direct consequence} of a homogeneous relation.
	\begin{center}
		\renewcommand{\arraystretch}{1.15}
		\begin{tabular}{c l l l l }
			\hline
			& Inhomogeneous Relation & Generator & Combination & Homogeneous Relation \\
			\hline
			1 & \textsc{bracket \(\triangle\)} & \(\m\) & \begin{minipage}[t]{3.1cm} \(\b \circ_1 \m + \m \circ_1 \b=\) \\
				\(\big(\b \circ_1 \m + \m \circ_1 \b\big)^{(123)}\) \end{minipage} & Consequence of \textsc{leibniz \(\triangle\)} \\
			2 & \textsc{bracket \(\triangle\)} & \(\b\) & \textsc{jacobi \(\triangle\)} & \textsc{jacobi \(\triangle\)} \\
			3 & \textsc{bracket \(\triangle\)} & \(\triangle\) & \textsc{derivation} & \textsc{derivation} \\
			4 & \textsc{bracket \(\square\)} & \(\m\) & \begin{minipage}[t]{3.1cm} \(\c \circ_1 \m + \m \circ_1 \c=\) \\
				\(\big(\c \circ_1 \m + \m \circ_1 \c\big)^{(123)}\) \end{minipage} & Consequence of \textsc{leibniz \(\square\)} \\
			5 & \textsc{bracket \(\triangle\)} and \textsc{bracket \(\square\)} & \(\c\) and \(\b\) & \textsc{compatibility 2} & \textsc{compatibility 2} \\
			6 & \textsc{bracket \(\triangle\)} and \textsc{bracket \(\square\)} & \(\square\) and \(\triangle\) & \textsc{compatibility 3} & \textsc{compatibility 3} \\
			7 & \textsc{bracket \(\square\)} & \(\c\) & No relation & No relation \\
			8 & \textsc{bracket \(\square\)} & \(\square\) & No relation & No relation \\
			\hline
		\end{tabular}
	\end{center}
	Let us make explicit two types of computations to show how it works.
	In the first case of the above list, one comes up with the following six elements of \(\cT(E_\bullet)\):
	\begin{align*}
		\Theta(a,b,c) &\defeq \LBM{a}{b}{c}-\DLMM{a}{b}{c}+\LMDM{a}{b}{c}+\LMMBBB{a}{b}{c}~,\\
		\Xi(a,b,c) &\defeq \LMB{a}{b}{c}-\LMDM{a}{b}{c}+\LMMB{a}{b}{c}+\LMMBB{a}{b}{c}~,
	\end{align*}
	where \((a,b,c)\) is either equal to \((1,2,3)\), \((2,3,1)\), or \((3,1,2)\).
	The only cubical trees (i.e.\ with three vertices) on the right-hand sides which appear in a composite of a homogeneous quadratic relation (associativity here) in \(R_\bullet\) with a generator (\(\triangle\) here) in \(E_\bullet\) are the ones where the operator \(\triangle\) sits at the top or at the bottom of the composite of two products \(\m\).
	So, up to the associativity relation composed with one \(\triangle\), the only way to get ony quadratic terms from a linear combination of the aforementioned six terms is
	\begin{align*}
		\Theta(1,2,3)+\Xi(1,2,3)-\Theta(2,3,1)+\Xi(2,3,1)
		&=
		\LBM{1}{2}{3}+\LMB{1}{2}{3}-\LBM{2}{3}{1}-\LMB{2}{3}{1},
	\end{align*}
	and the similar elements obtained by cyclic permutations.
	This homogeneous quadratic relation is equal to the difference of the Leibniz \(\triangle\) relation with inputs \((1,2,3)\) and the Leibniz \(\triangle\) relation with inputs \((2,3,1)\), so it is not new.

	In the last case of the above list, one comes up with the following three elements of \(\cT(E_\bullet)\):
	\begin{align*}
		&\CBoxL{1}{2}-\BoxMBoxL{1}{2}+\MBoxBoxL{1}{2}+\MBoxBox{1}{2}~,\\
		&\CBoxR{1}{2}-\BoxMBoxR{1}{2}+\MBoxBox{1}{2}+\MBoxBoxR{1}{2}~,\\
		&\BoxC{1}{2}-\BoxBoxM{1}{2}+\BoxMBoxL{1}{2}+\BoxMBoxR{1}{2}~.
	\end{align*}
	None of the cubical trees on the right-hand sides appear in any composite of a homogeneous quadratic relation in \(R_\bullet\) with a generator in \(E_\bullet\).
	Then, the presence of the three different cubical trees with two adjacent generators \(\square\) prevents us from producing a non-trivial homogeneous quadratic relation from a linear combination of these above three terms.
\end{proof}

	\cref{lemme:qlsforcBV} implies that the quadratic-linear relations of the operad \(\cBV\) are equal to
	the graph of the map \(\psi\) given by
	\[
	\DM{1}{2} - \MDL{1}{2} - \MDR{1}{2} \mapsto \BB{1}{2} \qquad \text{and} \qquad
	\BoxM{1}{2} - \MBoxL{1}{2} - \MBoxR{1}{2} \mapsto \CC{1}{2},
	\]
	and by \(0\) otherwise.

\subsubsection{The cooperad \(\cBV^{\ac}\)}

In order to develop an effective theory of \(\cBV_{\!\infty}\)-algebras, we will make explicit the cooperad
\[
\cBV^{\ac} \defeq \big(\q\cBV^{\ac}, d_1 + d_{\psi}\big)
= \big(\cC(sE_\bullet, s^2 \q R_\bullet), d_1 + d_{\psi}\big).
\]
To do so we start with an explicit description of the operad \(\cBV^! \).

\begin{lemma}\label{prop:cBV!}
	The Koszul dual operad \(\cBV^!\) admits the following presentation, where the symmetric action on all generators
	is the sign representation.

	\medskip\noindent \(\diamond\) The generating \(\Sy\)-module is given by
	\[
	\Big(\KK\set{\m'} \oplus \KK\set{\c'} \oplus \KK\set{\triangle'} \oplus \KK\set{\b'} \oplus \KK\set{\square'}\ ,\ (d_1)^* + (d_\psi)^*\Big),
	\]
	where
	\[
	\m' \defeq \Mprime,\quad
	\b' \defeq \Bprime,\quad
	\triangle' \defeq \Dprime,\quad
	\c' \defeq \Cprime,\quad
	\square' \defeq \BOXprime,
	\]
	with
	\[
	|\m'| = |\c'| = 0, \quad
	|\b'| = |\square'| = -1, \quad
	|\triangle'| = -2.
	\]
	The differential \((d_1)^*\) is the unique derivation extending the assignment
	\[
	\Mprime \mapsto 0,\quad
	\Cprime \mapsto -\Bprime \mapsto 0,\quad
	\BOXprime \mapsto -\Dprime \mapsto 0,
	\]
	and the differential \((d_\psi)^*\) is the unique derivation extending the assignment
	\[
	\Bprime \mapsto -\DMprime{1}{2},\quad
	\Cprime \mapsto -\BoxMprime{1}{2}.
	\]

	\medskip\noindent \(\diamond\) The \(\Sy\)-module of relations is spanned by the following 11 types of relators:

	\noindent
	\begin{tabular}{ll}
		\textnormal{\textsc{jacobi:}}
		\(\LLprime{1}{2}{3}+\LLprime{2}{3}{1}+\LLprime{3}{1}{2}\) &
		\textnormal{\textsc{associativity:}}
		\(\BBBprime{1}{2}{3}-\BBRprime{1}{2}{3}\) \\[4pt]

		\textnormal{\textsc{leibniz \(\m'\text{--}\b'\):}}
		\(\LMBprime{1}{2}{3}-\LBMprime{1}{3}{2}-\LBMprime{2}{3}{1}\) & \\[4pt]

		\textnormal{\textsc{commutativity \(\m'\text{--}\triangle'\):}}
		\(\DMprime{1}{2}-\MDLprime{1}{2}\) &
		\textnormal{\textsc{commutativity \(\b'\text{--}\triangle'\):}}
		\(\DBprime{1}{2}-\BDLprime{1}{2}\) \\[4pt]

		\textnormal{\textsc{nilpotency:}}
		\(\BoxBoxprime,\ \BoxCprime{1}{2},\ \CBoxLprime{1}{2},\ \LCCprime{1}{2}{3}\) &
		\textnormal{\textsc{commutativity \(\triangle'\text{--}\square'\):}}
		\(\DBoxprime - \BoxDprime\) \\[4pt]

		\textnormal{\textsc{leibniz \(\m'\text{--}\square'\):}}
		\(\LMBoxprime{1}{2}{3}+\LBoxMprime{1}{3}{2}+\LBoxMprime{2}{3}{1}\) &
		\textnormal{\textsc{commutativity \(\m'\text{--}\square'\):}}
		\(\BoxMprime{1}{2}+\MBoxLprime{1}{2}\)
	\end{tabular}

	\noindent
	\begin{tabular}{l}
		\textnormal{\textsc{commutativity \(\b'\text{--}\c'\):}}
		\(\CBprime{1}{2}{3}-\CBprime{2}{3}{1},\
		\BCprime{1}{2}{3}-\BCprime{2}{3}{1},\
		\CBprime{1}{2}{3}-\BCprime{1}{2}{3}\) \\[4pt]

		\textnormal{\textsc{compatibility \(\b'\square' \text{--} \c'\triangle'\):}}
		\(\BoxBprime{1}{2}-\DCprime{1}{2},\
		\DCprime{1}{2}+\CDLprime{1}{2},\
		\BoxBprime{1}{2}+\BBoxLprime{1}{2}\)~.
	\end{tabular}
\end{lemma}

\begin{proof}
	By \cite[Proposition~7.2.1]{LodayVallette12}, the Koszul dual operad $\cBV^!$ admits the following presentation
	\[
	\mathrm{S} \otimes_H\big(\q\cBV^{\ac}\big)^* \cong \mathcal{P}\big(s^{-1} \mathrm{S} \otimes_H(E_\bullet)^*, s^{-2}
	\mathrm{S} \otimes_H(\q R_\bullet)^{\perp}\big).
	\]
	Here the notation with a prime stands for the linear duals of the binary generators, like \(\m' \defeq \m^*\) for instance,
	and the desuspensions of the linear duals of the arity-one generators; like \(\triangle' \defeq s^{-1}\triangle^*\), for instance.
	The rest is a straightforward computation of the orthogonal space of quadratic relations.
	The dual differentials are direct consequences of the definitions of the coderivations on the Koszul dual cooperad; the minus sign
	comes from the Koszul convention \(\big(s^{\otimes 2}\big)^* = -s^* \otimes s^*\).
\end{proof}

The translation of this operadic statement in terms of algebraic structure gives the following description for \(\cBV^!\)-algebras.

\begin{proposition}\label{prop:DualcBValg}
	The desuspension of a differential graded \(\cBV^!\)-algebra structure on \(sA\) is
		a Gerstenhaber \(\mathbb{K}[\hbar, \varepsilon]\)-algebra \((A, m ,\{\, , \,\} )\) on a chain
	\(\mathbb{K}[\hbar]\)-complex \((A,d)\), with \(|\hbar|=-2\) and \(|\varepsilon|=-1\), such that the differential \(d\) is a derivation
	with respect to \(\{\, , \,\}\) and
	\begin{equation}\label{Eq:derivvareps}
	d(\varepsilon a) + \varepsilon d(a)=-\hbar a.
	\end{equation}
	Moreover, it is equipped with a binary symmetric degree \(1\) nilpotent \(\mathbb{K}[\hbar, \varepsilon]\)-linear product \(\mu\) satisfying
	\begin{itemize}
		\item[$\diamond$] the relation \(\varepsilon\mu(a,b)=0\),

		\item[$\diamond$] the two products \(\mu\) and \(m\) commute, i.e. \(\mu \circ_i m = m \circ_i \mu\), for \(i=1,2\),
		and their composite is fully symmetric, i.e. \(\mu \circ_1 m = \left(\mu \circ_1 m\right)^{(23)}\),
		\item[$\diamond$] the product \(\mu\) and bracket \(\{\; , \, \}\) satisfy the Leibniz relation,
		\item[$\diamond$] the relation \(\varepsilon m(a,b)=\hbar\mu(a,b)\) holds,
		\item[$\diamond$] the derivative of the product \(\mu\) satisfies
		\begin{equation}\label{Eq:derivmu}
		d(\mu(a,b))=-\mu(d(a), b) - (-1)^{|a|}\mu(a, d(b)) + m(a,b) + \varepsilon\{a,b\} ~.
		\end{equation}
	\end{itemize}
\end{proposition}

\begin{proof}
	This is a direct corollary of \cref{prop:cBV!}. On the desuspension \(A\) of \(sA\), the suspensions of the linear duals \(s\b^*\)
	and \(s\m^*\) act and induce
	symmetric binary operations \(m\) and \(\{\, , \,\}\) of
	respective degrees \(0\) and \(1\) and the desuspension of the linear dual \(s^{-1}\triangle^*\) induces a degree \(-2\) linear
	operator \(\tau\).
	The relations orthogonal to the quadratic analogues of the relations \(R\) coincide with the first quadratic relations given above.
	They assert that \((A, m, \{\, , \,\})\) is a
	Gerstenhaber algebra with which \(\tau\) commutes, that is a Gerstenhaber \(\mathbb{K}[\hbar]\)-algebra with \(|\hbar|=-2\).

	The desuspension of the linear dual \(s^{-1}\square^*\) induces a degree \(-1\) nilpotent operator \(\theta\) whose action is
	encoded by the algebra of dual numbers \(\mathbb{K}[\varepsilon]\).
	The suspension of the linear dual \(s\c^*\) induces a degree \(1\) symmetric nilpotent binary product \(\mu\).
	The rest of the relations given in \cref{prop:cBV!} produce the relations satisfied by \(\theta\) and \(\mu\).

	Relation~\eqref{Eq:derivvareps}, that \([d, \theta] = \tau\), is a direct consequence of the definition of the internal differential
	\(d_1\), which implies
	\(d_1^*\left(\square'\right)=-\triangle'\) in \(\mathrm{S} \otimes_H\left(\q\cBV^{\ac}\right)^*\).
	The induced differential \((d_\psi)^*\) on the suspension of the linear dual operad \(\mathrm{S} \otimes_H\left(\q\cBV^{\ac}\right)^*\)
	 is the unique derivation which
	sends \(\c'\) to \(-\square' \circ \m'\) and
	\(\b'\) to \(-\triangle' \circ \m'\) and the other generators to \(0\). The first value of the the derivation \((d_\psi)^*\) produces
	 Relation~\eqref{Eq:derivmu} and
	its second value gives a relation (\eqref{eq:dKoszuldual}) which is obtained from \eqref{Eq:derivmu} under \(d\).
	The other values of the derivation \((d_\psi)^*\) induce the fact that the differential \(d\)
	 commutes with the action of \(\hbar\) and with
	the commutative product.
	The sign in Relation~\eqref{Eq:derivmu} comes from the fact that the differential on \(A\) is equal to the opposite of the
	differential on \(sA\).
\end{proof}

	Let us recall the following description of the underlying \(\Sy\)-module of the cooperad \(\BV^\ac\).

\begin{lemma}[{\cite[Proposition~3]{GCTV12}}]\label{lem:qBVac}
	The underlying \(\Sy\)-module of the Koszul dual cooperad \(\q\BV^{\ac}\) is given by
	\[
	\q\BV^{\ac} \cong \mathrm{D}^{\ac} \circ \Lie_1^{\ac} \circ \Com^{\ac}
	\cong T^c(\delta) \circ \Com^c_{2} \circ \Lie_{1}^c,
	\]
	where \(\mathrm{D} \defeq T(\triangle)/\big(\triangle^2\big)\) is the algebra of dual numbers,
	\(\Lie_1\) is the operad of Lie algebras with bracket of degree \(1\),
	\(\Lie_{1}^c\) is the cooperad of Lie coalgebras with cobracket of degree \(1\),
	\(\Com^c_{2}\) the cooperad of cocommutative coalgebras with coproduct of degree \(2\), and
	\(T^c(\delta)\) is the cofree conilpotent coalgebra on a degree \(2\) generator.
\end{lemma}

\begin{proof}
	The quadratic analogue \(\q\BV\) of the operad \(\BV\) is given by a distributive law, see \cite[Section~1.2]{GCTV12}.
	This implies that its underlying \(\Sy\)-module is isomorphic to
	\[
	\q\BV \cong \Com \circ \Lie_1 \circ \mathrm{D}.
	\]
	Then it implies the isomorphism for the underlying \(\Sy\)-module of the Koszul dual cooperad \(\q\BV^{\ac}\).
\end{proof}

\begin{theorem}\label{thm:FormcBVac}
	The cooperads \(\BV^\ac\) and \(\cBV^\ac\) admit the following decompositions as cooperads:
	\[
	\begin{tikzcd}[column sep=-2.5pt]
		\BV^\ac & \cong & \Com^\ac & \oplus &
		\rM^*
		\arrow[loop below, in=240, out=300, distance=2.3em, "d_\varphi"]
	\end{tikzcd}
	\quad\text{and}\quad
	\begin{tikzcd}[column sep=-2.5pt]
		\cBV^\ac & \cong & \Com^\ac & \oplus &
		\rM^*
		\arrow[rr, bend left=45, "\quad d_1 = s^{-1}", out=80, in=100, distance=1.1em]
		\arrow[loop below, in=240, out=300, distance=2.3em, "d_\varphi"] &
		\oplus & s^{-1}\rM^* \arrow[loop below, in=245, out=295, distance=2em, "-d_\varphi"]
	\end{tikzcd}
	\]
	where
	\begin{itemize}[label=\(\diamond\), leftmargin=*]
		\item \(\Com^\ac\) is the sub-cooperad spanned by the cogenerator \(s\m\).
		\item \(\rM^*\) is the coideal of \(\BV^\ac\) spanned by the cogenerators \(s\b\) and \(s\triangle\).
		\item \(s^{-1}\rM^*\) is the coideal of \(\cBV^\ac\) spanned by the cogenerators \(s\c\) and \(s\square\).
	\end{itemize}
\end{theorem}

\begin{proof}
	\cref{lem:qBVac} gives the underlying \(\Sy\)-module of the Koszul dual cooperad \(\q\BV^{\ac}\)
	\[
	\q\BV^{\ac} \cong \mathrm{D}^{\ac} \circ \Lie_1^{\ac} \circ \Com^{\ac}.
	\]
	So the linear dual operad \(\big(\q\BV^{\ac}\big)^*\) splits as
	\[
	\big(\q\BV^{\ac}\big)^* \cong \Lie_{-1} \oplus \rM,
	\]
	where \(\Lie_{-1}\) is the operad of Lie algebras with bracket of degree \(-1\), viewed here as
	the sub-operad spanned by the generator \(s^{-1}\m^*\) and where \(\rM\) is the \(\Lie_{-1}\)-bimodule given by the ideal of \(\big(\q\BV^{\ac}\big)^*\) spanned by the two generators \(s^{-1}\b^*\) and \(s^{-1}\triangle^*\).

	The presentation of the operad \(\mathrm{S} \otimes_H\big(\q\cBV^{\ac}\big)^*\) given in the above \cref{prop:cBV!} shows that the pair of nilpotent generators \(\c'\) and \(\square'\) behave in the same way as the pair of generators \(\b'\) and \(\triangle'\) up to suspension.
	This implies the following isomorphism of operads
	\[
	\big(\q\cBV^{\ac}\big)^* \cong \Lie_{-1} \oplus \rM \oplus s\rM,
	\]
	where \(s\rM\) is isomorphic to the ideal of \(\big(\q\cBV^{\ac}\big)^*\) spanned by
	the two generators \(s^{-1}\c^*\) and \(s^{-1}\square^*\).
	Dually, we get the isomorphism of cooperads
	\[
	\q\cBV^{\ac} \cong \Com^\ac \oplus \rM^* \oplus s^{-1}\rM^*.
	\]
	Under this isomorphism, the codifferential \(d_1\) is equal to the desuspension isomorphism \(s^{-1} \colon \rM^* \to s^{-1}\rM^*\).
	Since the total number of cogenerators \(s\c\) and \(s\square\) is preserved under the codifferential \(d_\psi\), both coideals \(\rM^*\) cogenerated by \(s\b\) and \(s\triangle\) and \(s^{-1}\rM^*\) cogenerated by \(s\c\) and \(s\square\) are stable under \(d_\psi\), where they are respectively equal to \(d_\varphi\) and \(-d_\varphi\).
\end{proof}

\subsubsection{Koszul property}

\begin{theorem}\label{prop:cBVKoszul}
	The operad \(\cBV\) is Koszul.
\end{theorem}

\begin{proof}
	We have already seen in \cref{lemme:qlsforcBV} that the dg inhomogeneous quadratic data \(\big(E_\bullet, R_\bullet\big)\)
	introduced in \cref{lemma:presentationcBV} for the operad \(\cBV\) satisfies the quadratic conditions.
	It remains to study its homology twisting morphism \(\bar{\kappa}\).
	By \cref{thm:FormcBVac}, the homology of the Koszul dual cooperad \(\cBV^\ac\) with respect to the internal codifferential
	\(d_1\) is isomorphic to the Koszul dual cooperad \(\Com^\ac\):
	\[
	H_\bullet\big(\cBV^{\ac}, d_1\big) \cong \Com^\ac.
	\]
	\cref{thm:FormcBVac} also shows that the induced codifferential \(\bar{d}_\psi = 0\) is trivial in the case of the operad \(\cBV\).
	\cref{thm:Homology} computes the homology operad of \(\cBV\):
	\[
	H_\bullet\big(\cBV, \partial\big) \cong \Com.
	\]
	In the end, the homology twisting morphism \(\bar{\kappa}\) is equal to the canonical twisting morphism
	\(\Com^\ac \to \Com\) of the operad \(\Com\), which is known to be Koszul, see \cite[Proposition~13.1.2]{LodayVallette12}.
\end{proof}

\begin{remark}
	In the present case of the operad \(\cBV\), an even stronger phenomenon appears, which does not hold in full generality: the homology of the operad \(\cBV\) and the homology of the Koszul dual cooperad \(\cBV^\ac\) are given by the homology of their presentations
	\begin{align*}
		&H_\bullet\big(\cBV, \partial\big) \cong
		\mathcal{P}\big(H_\bullet(E_\bullet, \partial), H_\bullet(R_\bullet, \partial)\big) \cong \Com
		\quad \text{and} \quad \\
		&H_\bullet\big(\cBV^\ac, d_1\big) \cong
		\mathcal{C}\big(sH_\bullet(E_\bullet, \partial), s^2 H_\bullet(qR_\bullet, \partial)\big) \cong \Com^\ac~.
	\end{align*}
\end{remark}

%% file: sec/hopf.tex

\subsection{Tensor products and operadic diagonals}\label{ss:hopf}

The universal framework for discussing tensor products of algebras over an operad is the \defn{Hadamard product} of operads.
Given two symmetric operads \(\rO_1\) and \(\rO_2\), their \defn{Hadamard product} is the \(\Sy\)-module
\[
\rO_1 \,\ot_H\, \rO_2
\qquad\text{with components}\qquad
(\rO_1 \ot_H \rO_2)(n) \coloneqq \rO_1(n)\otimes \rO_2(n),
\]
equipped with the operad structure inherited aritywise from \(\rO_1\) and \(\rO_2\).

An \defn{operadic diagonal} on an operad \(\rO\) is a morphism of operads
\[
\Delta \colon \rO \longrightarrow \rO \ot_H \rO.
\]
Such a morphism encodes, in a universal way, the data needed to equip the tensor product of two \(\rO\)-algebras with a new \(\rO\)-algebra structure:
given \(\rO\)-algebras \(A\) and \(B\), the composite
\[
\rO \xrightarrow{\;\Delta\;} \rO \ot_H \rO
\longrightarrow \End_A \ot_H \End_B
\cong \End_{A \otimes B}
\]
determines a canonical \(\rO\)-algebra structure on \(A \otimes B\).

\medskip

The operad \(\Com\) carries the operadic diagonal determined by
\[
\Delta(\m) = \m \ot \m,
\]
which gives the usual tensor product of commutative algebras.
Similarly, the operad \(\BV\) admits an operadic diagonal determined by
\[
\Delta(\m) = \m \ot \m,
\qquad
\Delta(\triangle) = \triangle \ot \id + \id \ot \triangle,
\]
and the operad \(\cBV\) carries the operadic diagonal determined by
\[
\Delta(\m) = \m \ot \m,
\qquad
\Delta(\triangle) = \triangle \ot \id + \id \ot \, \triangle,
\qquad
\Delta(\square) = \square \ot \id + \id \ot \, \square.
\]

The following is a well known fact obtained using a standard lifting argument and the fact that the Hadamard product of acyclic fibrations is an acyclic fibration.

\begin{proposition}\label{prop:Hopf-lifting}
	Let \(\rO\) be an operad with an operadic diagonal \(\Delta\).
	Let \(\rO_\infty \twoheadrightarrow \rO\) be any cofibrant resolution.
	There exists an operadic diagonal \(\Delta_\infty\) on \(\rO_\infty\), unique up to homotopy, fitting in the diagram
		\[
	\begin{tikzcd}[column sep=large, row sep=large]
		\rO_\infty \arrow[r, "\Delta_\infty", dashed] \arrow[d, ->>, "\sim"] &
		\rO_\infty \ot_H \rO_\infty \arrow[d, ->>, "\sim"] \\
		\rO \arrow[r, "\Delta"] &
		\rO \ot_H \rO .
	\end{tikzcd}
	\]
\end{proposition}

Since \(\cBV\) admits an operadic diagonal, its cofibrant resolution \(\cBV_{\!\infty} = \Omega \cBV^\ac\) inherits one as well, unique up to homotopy, by \cref{prop:Hopf-lifting}.
Consequently, the tensor product of two \(\cBV_{\!\infty}\)-algebras carries a natural \(\cBV_{\!\infty}\)-algebra structure.
While this establishes the existence of such a structure, it does not provide an explicit formula for the operadic diagonal on \(\cBV_{\!\infty}\).
Since tensor products of \(\cBV_{\!\infty}\)-algebras are expected to play a central role in double copy constructions, we plan to develop explicit formulas for a chosen lift in future work.

We note that the same lifting argument applies to both \(\rC_\infty\)- and \(\BV_{\!\infty}\)-algebras, endowing them with tensor product structures.
However, explicit formulas for the operadic diagonals remain unknown in both cases.
As explained in \cite{medina2023dennis}, such formulas for \(\Com_\infty\) would resolve a question posed by D.~Sullivan in \cite[p.231]{lawrence2014interval}.

%% file: sec/homotopy_fibre.tex
\subsection{Model categorical relationship}\label{ss:homotopy_fibre}

We start from the canonical diagram of dg quadratic-linear data
\[
\begin{tikzcd}[column sep=small]
	(V, S) \rar & (E_\bullet, R_\bullet) \rar & (E, R)
\end{tikzcd}
\]
defining the diagram \(\Com \to \cBV \to \BV\) studied in \cref{thm:Homology}.
Passing to the associated Koszul dual cooperads we obtain a diagram of cooperads
\begin{equation*}\label{eq:diagram_cooperads}
	\begin{tikzcd}[column sep=small]
		\Com^{\ac} \rar & \cBV^{\ac} \rar & \BV^{\ac}.
	\end{tikzcd}
\end{equation*}
Applying the cobar construction to it yields a diagram of operads
\[
\begin{tikzcd}[column sep=small]
	\Omega\Com^{\ac} \rar & \Omega\cBV^{\ac} \rar & \Omega\BV^{\ac}.
\end{tikzcd}
\]

\begin{theorem}\label{prop:FactAcyCofFib}
	The diagram
	\[
	\begin{tikzcd}[column sep=20pt]
		\Com_\infty \arrow[r, >->, "\sim"] & \cBV_{\!\infty} \arrow[r,->>] & \BV_{\!\infty}
	\end{tikzcd}
	\]
	is made up of an acyclic cofibration followed by a fibration.
\end{theorem}

\begin{proof}
	The above diagram of cooperads extends to a short exact sequence.
	\[
	\begin{tikzcd}[column sep=small]
		0 \rar & \Com^{\ac} \rar & \cBV^{\ac} \rar & \BV^{\ac} \rar & 0
	\end{tikzcd}
	\]
	Forgetting the differentials, it corresponds to
	\[
	\begin{tikzcd}[column sep=small]
		\Com^{\ac} \rar & \Com^\ac \oplus \rM^* \oplus s^{-1}\rM^* \rar & \Com^\ac \oplus \rM^*
	\end{tikzcd}
	\]
	by \cref{thm:FormcBVac}, with the canonical inclusion and projection defined by the direct summands.
	Since the cobar construction is a left adjoint, we get the exact sequence
	\[
	\begin{tikzcd}[column sep=small]
		\Com_\infty \rar & \cBV_{\!\infty} \rar & \BV_{\!\infty} \rar & 0
	\end{tikzcd}
	\]
	and this implies that the map \(\cBV_{\!\infty} \to \BV_{\!\infty}\) is a fibration.
	Let us now focus on the first map.
	Forgetting the differentials, it is the inclusion into the coproduct
	\[
	\Com_\infty \to \Com_\infty \vee \mathcal{T}(s^{-1}\rM^* \oplus s^{-2}\rM^*).
	\]
	Both operads \(\Com_\infty\) and \(\cBV_{\!\infty}\) are non-negatively graded.
	The degree \(0\) summand of the generating space \(s^{-1}\rM^* \oplus s^{-2}\rM^*\) is spanned by \(s^{-1}(s \c) \cong \c\) and \(s^{-1}(s \square) \cong \square\), whose differential vanishes.
	Defining the increasing and exhaustive filtration
	\[
	\mathrm{S}_k \coloneqq \bigoplus_{l=0}^k (s^{-1}\rM^* \oplus s^{-2}\rM^*)_l,
	\]
	we obtain an increasing and exhaustive filtration \(\Com_\infty \vee \mathcal{T}(\mathrm{S}_k)\) of the operad \(\cBV_{\!\infty}\) satisfying
	\[
	(d_1 + d_\psi + d_2)(S_k) \subset \Com_\infty \vee \mathcal{T}(\mathrm{S}_{k-1})
	\]
	for all \(k \geqslant 0\).
	This proves that the map \(\Com_\infty \to \cBV_{\!\infty}\) is a cofibration, which is acyclic by the arguments of \cref{prop:cBVKoszul}.
\end{proof}

A straightforward consequence of this theorem is that the kernel of the map \(\cBV_{\!\infty} \to \BV_{\!\infty}\) is a model for the \textit{homotopy fibre} of the map \(\Com_\infty \to \BV_{\!\infty}\).
By \cref{thm:FormcBVac}, this kernel is the ideal of \(\cBV_{\!\infty}\) generated by \(s^{-2}\rM^*\).
It is explicitly given by the linear span of trees with vertices labelled by elements of \(s^{-1}\overline{\cBV}^{\ac}\) such that at least one vertex is labelled by an element of \(s^{-2}\rM^*\).

%
%


%% file: sec/homotopy_algebras.tex

\subsection{Homotopy algebras}\label{ss:homotopy_algebras}

We verify that algebras over the operad \(\cBV_{\!\infty}\) are precisely \(\cBV_{\!\infty}\)-algebras as defined in \cref{ss:generating_maps} and compare them to the \(\BV^\square_\infty\)-algebras introduced by M. Reiterer \cite{Reiterer2020HomotopyBVYMCK}.

\medskip

The main technical tool we will use is the following.

\begin{lemma}[{\cite[Proposition~10.1.1]{LodayVallette12}}]\label{lem:TwCAlg}
	For any conilpotent cooperad \(\rC\),
	an \(\Omega\rC\)-algebra structure \(\Omega\rC \to \End_A\) on a chain complex \(A\) is equivalent to a twisting morphism \(\rC \to \End_A\).
\end{lemma}

\subsubsection{\(\rC_\infty\)-algebras}

Let us start with the case of \(\rC_\infty\)-algebras to fix the general method.

\begin{proposition}
	Algebras over the Koszul resolution \(\Com_\infty \defeq \Omega \Com^{\ac}\) are precisely the \(\rC_\infty\)-algebras defined in \cref{ss:generating_maps}.
\end{proposition}

\begin{proof}
	By \cref{lem:TwCAlg}, consider a twisting morphism \(\alpha \colon \Com^{\ac} \to \End_A\), that is, a degree \(-1\) map with \(\alpha(\id)=0\) satisfying the Maurer--Cartan equation
	\begin{equation}\label{EQ:MCCOM}
		\partial_A \alpha + \alpha \star \alpha = 0,
	\end{equation}
	where \(\partial_A\) denotes the differential of \(\End_A\).
	As recalled in \cref{ex:Comac}, the Koszul dual cooperad \(\Com^\ac \cong \Lie^c_1\) is the cooperad of shifted Lie coalgebras.
	The relationship with the Koszul dual cooperad of associative algebras \(\Ass^\ac \cong \Ass^c_1\) is given by Ree's theorem
	 \cite[Theorem~1.3.6]{LodayVallette12}: the kernel of the canonical morphism of cooperads \(\Ass^c_1 \to \Lie^c_1\) is spanned
	 by the signed sums of shuffles.
	 Therefore the data of a Maurer--Cartan element \(\alpha \colon \Lie^c_1 \to \End_A\) is equivalent to the data of a
	 Maurer--Cartan element \(\alpha \colon \Ass^c_1 \to \End_A\) vanishing on the signed sum of shuffles. The cooperad \(\Ass^c_1\)
	 is isomorphic to the Koszul dual cooperad \(\Ass^{\ac}\) of the operad \(\Ass\) encoding associative algebras. In the end, a
	 \(\Com_\infty\)-algebra is an \(\Ass_\infty\)-algebra, also called \(\rA_\infty\)-algebra, whose structure operations vanish on the
	 signed sum of shuffles.
\end{proof}

\subsubsection{\(\BV_{\!\infty}\)-algebras}\label{ss:BVinfty}

The case of \(\BV_{\!\infty}\)-algebras is more involved due to the higher complexity of the Koszul dual cooperad \(\BV^\ac\), but it is mandatory on the way to \(\cBV_{\!\infty}\)-algebras.

\begin{theorem}[{\cite[Theorem~20]{GCTV12}}]\label{thm:BVinfty}
	Algebras over the Koszul resolution \(\BV_{\!\infty} \defeq \Omega\BV^{\ac}\) are precisely the \(\BV_{\!\infty}\)-algebras defined in
	 \cref{BVinfinity}, that is chain complexes \((A,d)\) equipped with a collection of maps \(\{m_{p_1,\dots,p_k}^t\}\) satisfying the 		block and shuffle symmetries, with \(m_1^0=d\), and satisfying the relations
	\begin{align*}\tag{\(\mathsf{M}^t_{p_1,\dots,p_k}\)}\label{RelationsM}
		&\quad\ 	\sum_{\mathclap{\begin{array}{c}
			\scriptstyle 0 \leqslant s \leqslant t
	\end{array}}}
	\hspace{25pt}
	\sum_{\hspace*{10pt}\mathclap{\begin{array}{l}
				\scriptstyle I \,\sqcup J = \{1, \dots, k\} \\[-2pt]
				\scriptstyle I = \{i_1, \dots, i_a\} \neq \emptyset \\[-2pt]
				\scriptstyle J = \{j_1, \dots, j_b\}
	\end{array}}}
	\hspace*{35pt} 
	\sum_{\hspace*{0pt}\mathclap{\begin{array}{c}
				\scriptstyle 1 \leqslant q_1 \leqslant p_1 \\[-2pt]
				\scalebox{0.5}{\(\vdots\)} \\[-2pt]
				\scriptstyle 1 \leqslant q_a \leqslant p_a
	\end{array}}}
	\hspace*{3pt}	
	\pm \ m^{s}_{r,\, p_{j_1}, \dots, p_{j_b}}
	\Big(m^{t-s}_{\frac{p_1, \dots, p_{i_a}}{q_1, \dots, q_a}}
	(w_{i_1} \!\ot\dotsb\ot w_{i_a}) \ot w_{j_1} \!\ot\dotsb\ot w_{j_b}\Big)
\\
		&
		-\sum_{1 \leqslant i \leqslant k}
	\hspace*{15pt}
	\sum_{\mathclap{1 \leqslant j \leqslant p_i-1}}
	\ \pm \, m^{t-1}_{p_1, \dots, j, p_i-j, \dots, p_k}
	\Big(w_1 \ot\dotsb\ot w_i^{(1)} \ot w_i^{(2)} \ot\dotsb\ot w_k\Big)=0.
	\end{align*}
\end{theorem}

\begin{proof}
	By \cref{lem:qBVac}, the underlying graded \(\Sy\)-modules of the Koszul dual cooperad is given by
		\[
	\q\BV^{\ac} \cong \mathrm{D}^{\ac} \circ \Lie_1^{\ac} \circ \Com^{\ac}
	\cong T^c(\delta) \circ \Com^c_{2} \circ \Lie_{1}^c.
	\]
	These latter three cooperads admits the following underlying graded \(\Sy\)-modules:
	\[
	T^c(\delta)\cong \KK[\delta], \quad
	\Com_{2}^c(n) \cong s^{2n-2} \KK_n,
	\qquad \text{and} \qquad
	\Lie_{1}^c(n)\cong s^{n-1}\mathrm{sign}_{n} \otimes \Lie^*(n),
	\]
	where \(\KK_n\) stands for the trivial representation of \(\Sy_n\) and \(\mathrm{sign}_n\) stands for the sign representation
	of \(\Sy_n\).
	Therefore, the image of any morphism \(\alpha \colon \q\BV^\ac \to \mathrm{End}_A\) of \(\Sy\)-modules is a collection of
	maps \(\{m_{p_1,\dots,p_k}^t\}\) having the degrees and symmetries described in \cref{ss:generating_maps}.

	Then the unique coderivation \(d_\varphi\) of the cooperad \(\q\BV^{\ac}\) which extends the map \(\varphi\) is explicitly given by
	\[
	d_\varphi\big(\delta^t\otimes L_1 \odot \cdots \odot L_k\big) =
	\sum_{i=1}^k (-1)^{|L_1|+\cdots+|L_{i-1}|}\, \delta^{t-1}
	\otimes L_1\odot \cdots\odot L_i'\odot L_i''\odot \cdots
	\odot L_k,
	\]
	where \(\odot\) stands for the graded symmetric tensor product and where \(L_i' \odot L_i''\) is Sweedler-type notation for the image of \(L_i\) under the binary part
	\[
	\Lie^c_1 \to \Lie^c_1(2)\otimes_{\Sy_2} \big(\Lie^c_1 \otimes \Lie^c_1\big)
	\]
	of the decomposition map of the cooperad \(\Lie^c_1\).
	The image of \(d_\varphi\) is equal to \(0\) when \(t=0\) or on elements \(L_i \in \Lie^c_1(1)=\id\KK\).

	The Maurer--Cartan equation defining twisting morphisms from \(\BV^\ac=\left(\q\BV^\ac, d_\varphi\right)\) to \(\mathrm{End}_A\)
	is
	\begin{equation*}
		\partial_A \alpha + \alpha d_\varphi+ \alpha \star \alpha = 0.
	\end{equation*}
	Under the convention \(m_1^0=d\), the first and the third terms of this Maurer--Cartan equation coincide with the first term of
	 Relations~\eqref{RelationsM} and the second term of the Maurer--Cartan equation coincides with the second term of
	 Relations~\eqref{RelationsM}, see \cite[\S 2]{GCTV12} for a more details.
\end{proof}

\subsubsection{\(\cBV_{\!\infty}\)-algebras}

\begin{theorem}\label{thm:cBVinfty}
	Algebras over the Koszul resolution \(\cBV_{\!\infty} \defeq \Omega\cBV^{\ac}\) are precisely the \(\cBV_{\!\infty}\)-algebras defined in~\cref{ss:generating_maps}.
\end{theorem}

\begin{proof}
	By \cref{lem:TwCAlg}, \(\cBV_{\!\infty}\)-algebra structures are equivalent to twisting morphisms
	 \(\alpha \colon \cBV^\ac \to \mathrm{End}_A\), which are degree \(-1\) assignments satisfying the Maurer--Cartan equation
	\begin{equation}\label{eq:TwAlpha}
		\partial_A \alpha + \alpha (d_1+d_\psi) + \alpha \star \alpha = 0.
	\end{equation}
	\cref{thm:FormcBVac} shows that the Koszul dual cooperad \(\cBV^\ac\) satisfies the following decomposition:
	\[
	\begin{tikzcd}[column sep=-2.5pt]
		\cBV^\ac & \cong & \Com^\ac & \oplus &
		\rM^*
		\arrow[rr, bend left=45, "\quad d_1 = s^{-1}", out=80, in=100, distance=1.1em]
		\arrow[loop below, in=240, out=300, distance=2.3em, "d_\varphi"] &
		\oplus & s^{-1}\rM^*, \arrow[loop below, in=245, out=295, distance=2em, "-d_\varphi"]
	\end{tikzcd}
	\]
	where the dg sub-cooperad \(\big(\Com^\ac \oplus \rM^*, d_\varphi\big)\) is isomorphic to the cooperad \(\BV^\ac\).
	Therefore its image under the twisting morphism \(\alpha\) produces operations \(\{m_{p_1,\dots,p_k}^t\}\) for \(t \geqslant 0\), \(k \geqslant 1\), and \(p_1,\ldots,p_k \geqslant 1\), having the same degrees and symmetries as the ones of \(\BV_{\!\infty}\)-algebras described above in \cref{thm:BVinfty}.
	The image of the last summand \(s^{-1}\rM^*\cong \q \BV^\ac/\Com^\ac\) under the twisting morphism \(\alpha\) produces operations \(\{-n_{p_1,\dots,p_k}^t\}\) for \(t\geqslant 0\), \(k\geqslant 1\), and \(p_1,\ldots,p_k\geqslant 1\), with \(t+k\geqslant 2\), satisfying the same symmetries as the operations \(\{m_{p_1,\dots,p_k}^t\}\) but with degree shifted by \(-1\).

	\medskip

	The restriction of the Maurer--Cartan equation \eqref{eq:TwAlpha} to the summand \(\Com^\ac\) is simply \(\partial_A \alpha\!\!\mid_{\Com^\ac} + \alpha\!\!\mid_{\Com^\ac} \star \alpha\!\!\mid_{\Com^\ac} = 0\), which is the Maurer--Cartan equation \eqref{EQ:MCCOM} satisfied by the operations \(\{m^0_{p}\}_{p\geqslant 1}\), which therefore form a \(\rC_\infty\)-algebra.
	The evaluation of the Maurer--Cartan equation \eqref{eq:TwAlpha} on any element \(\mu^t_{p_1, \ldots, p_k}\in \rM^*\), simply denoted by \(\mu\), is equal to
	\[
	\partial_A \alpha(\mu) + \alpha (d_1+d_\psi)(\mu) + (\alpha \star \alpha)(\mu) = 0,
	\]
	which is equivalent to
	\[
	-\alpha \big(s^{-1}\mu\big)= \partial_A \alpha(\mu) +\alpha d_\varphi(\mu) + (\alpha \star \alpha)(\mu).
	\]
	The right-hand side of this equation is the Maurer--Cartan equation defining the relations \eqref{RelationsM} satisfied by the operations
	\(\{m_{p_1,\dots,p_k}^t\}\) of a \(\BV_{\!\infty}\)-algebra.
	Since the left-hand side is equal to the obstruction map \(n_{p_1,\dots,p_k}^t\),
	it completely prescribes their values in terms of the generating maps \(\{m_{p_1,\dots,p_k}^t\}\); this actually recovers
	their definition given in \cref{sssec:Obstructions}.

	It remains to prove that the obstruction maps \(n_{p_1,\dots,p_k}^t = -\alpha \big(s^{-1}\mu^t_{p_1,\ldots, p_k}\big)\)
	automatically satisfy the evaluation of the Maurer--Cartan equation \eqref{eq:TwAlpha} on the last summand \(s^{-1}\rM^*\), that is
	\[\partial_A \alpha \big(s^{-1}\mu\big) -\alpha\big(s^{-1}d_\varphi(\mu)\big)+(\alpha \star \alpha)\big(s^{-1}\mu\big)=0.\]
	The image under \(\partial_A\) of Equation~\eqref{eq:TwAlpha} evaluated on any \(\mu\in \rM^*\) gives
	\[
	\partial_A \alpha (d_1+d_\psi)(\mu) + \partial_A(\alpha \star \alpha)(\mu)=
	\partial_A \alpha (d_1+d_\psi)(\mu) + \big((\partial_A\alpha) \star \alpha\big)(\mu)
	-\big(\alpha \star (\partial_A\alpha)\big)(\mu)=0,
	\]
	since \(\partial_A\) is a derivation of the operad \(\mathrm{End}_A\).
	Since both sides of the product \(\star\) apply to strictly lower weight elements of \(\Com^\ac \oplus \rM^*\), we can apply
	Equation~\eqref{eq:TwAlpha} to them and obtain
	\begin{multline*}
		\partial_A \alpha (d_1+d_\psi)(\mu) -
		\big((\alpha\star\alpha) \star \alpha\big)(\mu)
		-\big((\alpha(d_1+d_\psi) \star \alpha\big)(\mu)
		- \big(\alpha\star (\alpha\star\alpha)\big)(\mu)\\
		- \big(\alpha\star(\alpha(d_1+d_\psi)\big)(\mu)=
		0.
	\end{multline*}
	The pre-Lie relation implies that the degree \(-1\) element \(\alpha\) satisfies \((\alpha\star\alpha) \star \alpha=\alpha\star (\alpha\star\alpha)\).
	Since \(d_1+d_\psi\) is a coderivation of the cooperad \(\cBV^\ac\), we get
	\begin{multline*}
		\partial_A \alpha d_1(\mu) +
		\partial_A \alpha d_\psi(\mu)
		+\big(\alpha\star\alpha\big)(d_1+d_\psi)(\mu)=
		\partial_A \alpha \big(s^{-1}\mu\big) +
		\big(\alpha \star \alpha\big)\big(s^{-1}\mu\big)\\
		+
		\big(\partial_A \alpha + \alpha \star\alpha\big)\big(d_\psi(\mu)\big)=
		0.
	\end{multline*}
	Since \(d_\psi\) lowers the weight by one on elements of \(\rM^*\), we can use Equation~\eqref{eq:TwAlpha} evaluated on \(d_\psi(\mu)\) to finally get
	\[
		\big(\partial_A \alpha +
		\alpha \star \alpha\big)\big(s^{-1}\mu\big)-
		\alpha(d_1+d_\psi)d_\psi(\mu)
		=
		\partial_A \alpha \big(s^{-1}\mu\big) -
		\alpha\big(s^{-1}d_\varphi(\mu)\big)
		+(\alpha \star \alpha)\big(s^{-1}\mu\big)
		=0. \ \qedhere
	\]
\end{proof}

\begin{corollary}\label{cor:RelationsN}
The obstruction maps of any \(\cBV_{\!\infty}\)-algebra satisfy the relation
	\begin{align*}\tag{\(\mathsf{N}^t_{p_1,\dots,p_k}\)}\label{RelationsN}
		&\quad\ 	\sum_{\mathclap{\begin{array}{c}
			\scriptstyle 0 \leqslant s \leqslant t
	\end{array}}}
	\hspace{25pt}
	\sum_{\hspace*{10pt}\mathclap{\begin{array}{l}
				\scriptstyle I \,\sqcup J = \{1, \dots, k\} \\[-2pt]
				\scriptstyle I = \{i_1, \dots, i_a\} \neq \emptyset \\[-2pt]
				\scriptstyle J = \{j_1, \dots, j_b\}
	\end{array}}}
	\hspace*{35pt} 
	\sum_{\hspace*{0pt}\mathclap{\begin{array}{c}
				\scriptstyle 1 \leqslant q_1 \leqslant p_1 \\[-2pt]
				\scalebox{0.5}{\(\vdots\)} \\[-2pt]
				\scriptstyle 1 \leqslant q_a \leqslant p_a
	\end{array}}}
	\hspace*{3pt}	
	\pm \ m^{s}_{r,\, p_{j_1}, \dots, p_{j_b}}
	\Big(n^{t-s}_{\frac{p_1, \dots, p_{i_a}}{q_1, \dots, q_a}}
	(w_{i_1} \!\ot\dotsb\ot w_{i_a}) \ot w_{j_1} \!\ot\dotsb\ot w_{j_b}\Big)
\\&\quad 	-\sum_{\mathclap{\begin{array}{c}
			\scriptstyle 0 \leqslant s \leqslant t
	\end{array}}}
	\hspace{25pt}
	\sum_{\hspace*{10pt}\mathclap{\begin{array}{l}
				\scriptstyle I \,\sqcup J = \{1, \dots, k\} \\[-2pt]
				\scriptstyle I = \{i_1, \dots, i_a\} \neq \emptyset \\[-2pt]
				\scriptstyle J = \{j_1, \dots, j_b\}
	\end{array}}}
	\hspace*{35pt} 
	\sum_{\hspace*{0pt}\mathclap{\begin{array}{c}
				\scriptstyle 1 \leqslant q_1 \leqslant p_1 \\[-2pt]
				\scalebox{0.5}{\(\vdots\)} \\[-2pt]
				\scriptstyle 1 \leqslant q_a \leqslant p_a
	\end{array}}}
	\hspace*{3pt}	
	\pm \ n^{s}_{r,\, p_{j_1}, \dots, p_{j_b}}
	\Big(m^{t-s}_{\frac{p_1, \dots, p_{i_a}}{q_1, \dots, q_a}}
	(w_{i_1} \!\ot\dotsb\ot w_{i_a}) \ot w_{j_1} \!\ot\dotsb\ot w_{j_b}\Big)
\\
		&\quad
		+\sum_{1 \leqslant i \leqslant k}
	\hspace*{15pt}
	\sum_{\mathclap{1 \leqslant j \leqslant p_i-1}}
	\ \pm \, n^{t-1}_{p_1, \dots, j, p_i-j, \dots, p_k}
	\Big(w_1 \ot\dotsb\ot w_i^{(1)} \ot w_i^{(2)} \ot\dotsb\ot w_k\Big)=0.
	\end{align*}
	for \(t \geqslant 0\), \(k \geqslant 1\), and \(p_1,\ldots,p_k \geqslant 1\) such that \(t+k\geqslant 2\), where the sign of the second sum is equal to the sign of the first sum of Relations~\eqref{RelationsM}, where the sign of the first sum is the same but with the extra term
	\[
	(-1)^{p_{j_1}+\cdots+p_{j_b}+r-1},
	\]
	and where the sign of the third sum is equal to the sign of the second sum of \eqref{RelationsM}.
\end{corollary}

\begin{proof}
	These relations coincide to the Maurer--Cartan equation
		\[-\partial_A \alpha \big(s^{-1}\mu\big) +\alpha\big(s^{-1}d_\varphi(\mu)\big)-(\alpha \star \alpha)\big(s^{-1}\mu\big)=0\]
		established above in the proof of \cref{thm:cBVinfty}:
	under the convention \(m_1^0=d\), the first and the third terms of this Maurer--Cartan equation coincide with the first two terms of
	 Relations~\eqref{RelationsN} and the second term of the Maurer--Cartan equation coincides with the last term of
	 Relations~\eqref{RelationsN}.
\end{proof}

\subsubsection{\(\BV^\square_{\!\infty}\)-algebras}\label{ss:reiterer}

The notion of \(\cBV_{\!\infty}\)-algebra introduced in this paper is close to the notion of \(\BV^\square_\infty\)-algebra introduced by M. Reiterer in \cite{Reiterer2020HomotopyBVYMCK}.
His definition can be unravelled to provide a collection of operations \(\{m^t_{p_1,\ldots, p_k}\}\) with the same degrees and symmetries as ours and, when they are unobstructed, they also give rise to \(\BV_{\!\infty}\)-algebra structures.

In a \(\BV^\square_\infty\)-algebra, there is a linear operator of degree \(0\) given by the action of a certain element of a Hopf algebra; such an operator can be encoded by the general operator \(n^1_1\) in a \(\cBV_{\!\infty}\)-algebra.
The relations satisfied by a \(\BV^\square_\infty\)-algebra are obtained by starting from the relations of a \(\BV_{\!\infty}\)-algebra and by adding an extra term, which is equal to \(0\) when \(k=1\), except for \(t=1\) and \(p_1=1\), where it is precisely the aforementioned operator.
Such obstructions are encoded in the present operations \(\{n^t_{p_1, \ldots, p_k}\}\), but in the case of \(\BV^\square_\infty\)-algebras these obstructions admit a formula involving only the operator \(n^1_1\) and the operations \(\{m^t_{p_1, \ldots, p_k}\}\).
Thus the notion of a \(\BV^\square_\infty\)-algebra is more restrictive than the present notion of a \(\cBV_{\!\infty}\)-algebra.

\medskip

The discrepancy between these two notions comes from the different approaches which lead to them.
M. Reiterer starts with an equivalent definition of \(\BV_{\!\infty}\)-algebras in terms of a series of (co)derivations on the (co)free Gerstenhaber algebra whose sum satisfies some Maurer--Cartan equation, see \cite[\S 4.1]{Reiterer2020HomotopyBVYMCK}, and then modifies this equation ``by hand'' in order to add obstructions with a prescribed form.
One limitation of this approach is that it produces the above-mentioned constraints, including the fact that there are no obstructions for \(t\geqslant 0\), \(p\geqslant 0\), except in the sole case \(t=p=1\).
The equivalent definition of a \(\cBV_{\!\infty}\)-algebra in terms of a square-zero coderivation of the cofree \(\cBV^{\ac}\)-coalgebra is more subtle and cannot be obtained so easily from the similar definition of a \(\BV_{\!\infty}\)-algebra, see \cref{Def:BarConstr} and \cref{prop:FormBar}.

\medskip

Even if the notion of \(\BV^\square_\infty\)-algebras could be encoded by an operad, we do not expect this operad to be a cofibrant replacement of the operad \(\cBV\); so \(\BV^\square_\infty\)-algebras seem not to be a model for homotopy \(\cBV\)-algebras as we understand them.

%% file: sec/homotopical_tools.tex

\subsection{Homotopical properties}\label{ss:tools}

The cofibrant replacement \(\cBV_{\!\infty} \defeq \Omega \cBV^{\ac}\) of the operad \(\cBV\) has the crucial feature of being the cobar construction of a cooperad, namely the Koszul dual \(\cBV^{\ac}\).
This places \(\cBV_{\!\infty}\)-algebras into a rich algebraic and homotopical framework that provides them with a \emph{deformation theory}, an \emph{obstruction theory}, \emph{\(\infty\)-morphisms}, a \emph{homotopy transfer} and a \emph{rectification theorem}.
We only sketch these properties here and refer to \cite[Chapters~10--12]{LodayVallette12} for detailed proofs.

\subsubsection{Deformation theory}

\begin{definition}[{\cite[Section~10.1]{LodayVallette12}}]
	For any chain complex \(A\), the \defn{deformation dg pre-Lie algebra} of \(\cBV_{\!\infty}\)-algebra structures on \(A\) is
	\[
	\g_{\cBV, A} \defeq
	\left(
	\Hom_{\Sy} \left({\cBV}^{\ac}, \End_A \right), \partial, \star
	\right).
	\]
\end{definition}

\begin{proposition}\label{prop:DefpreLie}
	The underlying module of the deformation dg pre-Lie algebra is isomorphic to
	\[
	\Hom_{\Sy}(\cBV^{\ac}, \End_A)
	\cong
	s^{2}\, \Hom\big(\mathrm{S}^c(s\mathrm{Lie}^c(sA))[\delta], A\big)
	\oplus
	s^{3}\, \Hom\big(\mathrm{S}^c(s\mathrm{Lie}^c(sA))[\delta]/s\mathrm{Lie}^c(sA), A\big),
	\]
	where \(\delta\) has degree \(2\), \(\mathrm{S}^c(-)\) is the conilpotent cofree cocommutative coalgebra, and \(\mathrm{Lie}^c(-)\) is the conilpotent cofree Lie coalgebra.
\end{proposition}

\begin{proof}
	This follows from the decomposition of the Koszul dual cooperad given in \cref{thm:FormcBVac}
	\[
	\cBV^{\ac}
	\cong \underbrace{\Com^{\ac} \oplus \rM^*}_{\cong \BV^{\ac}} \oplus s^{-1}\rM^*
	\cong \BV^{\ac} \oplus s^{-1}\BV^{\ac}/\Com^{\ac},
	\]
	which yields
	\begin{align*}
		\cBV^{\ac}(A) &\cong \BV^{\ac}(A) \oplus s^{-1}(\BV^{\ac}/\Com^{\ac})(A) \\
		&\cong s^{-2}\,\mathrm{S}^c(s\mathrm{Lie}^c(sA))[\delta] \oplus
		s^{-3}\,\mathrm{S}^c(s\mathrm{Lie}^c(sA))[\delta]/s^{-2}\mathrm{Lie}^c(sA),
	\end{align*}
	since \(\Com^{\ac}(A) \cong s^{-1}\mathrm{Lie}^c(sA)\).
\end{proof}

Maurer--Cartan elements, i.e.\ degree \(-1\) elements \(\alpha\) satisfying
\[
\partial\alpha + \alpha \star \alpha
= \partial_A \alpha + \alpha(d_1+d_\psi) + \alpha \star \alpha = 0,
\]
of the deformation dg pre-Lie algebra \(\g_{\cBV, A}\) are in one-to-one correspondence with \(\cBV_{\!\infty}\)-algebra structures on \(A\), see \cite[Proposition~10.1.1]{LodayVallette12}.
One may start with a graded module \(A\) whose differential is encoded in the first component of \(\alpha\), or with a dg module \((A, d)\) and a Maurer--Cartan element having trivial component on \(A\).
The pre-Lie product admits explicit integration formulas by \cite[Theorem~2]{DSV16}, yielding a complete topological group structure on the degree zero part: this is the \defn{gauge group of symmetries} \(\mathrm{G}\).
Functoriality of the deformation pre-Lie algebra in the Koszul dual cooperad allows one to relate the deformation dg pre-Lie algebras of \(\cBV_{\!\infty}\)-algebras with those of \(\BV_{\!\infty}\)-algebras, \(\Gerst_\infty\)-algebras, and \(\Com_\infty\)-algebras.

\begin{definition}
	The \defn{cohomology groups} of a \(\cBV_{\!\infty}\)-algebra \(\alpha \in \mathrm{MC}(\g_{\cBV, A})\) are the groups defined by the chain complex
	\[
	\g_{\cBV, A}^\alpha \defeq
	\big(\Hom_{\Sy}(\cBV^{\ac}, \End_A), \partial^\alpha \defeq \partial + [\alpha, -]\big),
	\]
	with differential twisted by \(\alpha\).
\end{definition}

Since twisting by a Maurer--Cartan element produces a dg Lie algebra, this chain complex carries a Lie bracket, sometimes called the \defn{intrinsic Lie bracket}, which detects the deformations of the \(\cBV_{\!\infty}\)-algebra structure, see \cite[Section~12.2]{LodayVallette12}.

\begin{theorem}
	The space \(Z_{-1}\big(\g_{\cBV, A}^\alpha\big)\) of degree \(-1\) cycles is isomorphic to the tangent space at \(\alpha\) of the algebraic variety
	\(\mathrm{MC}\big(\g_{\cBV, A}\big)\) and the degree \(-1\) homology group
	\(H_{-1}\big(\g_{\cBV, A}^\alpha\big)\) is isomorphic to the tangent space at \(\alpha\) of the algebraic stack
	\(\mathcal{MC}\big(\g_{\cBV, A}\big)\defeq \mathrm{MC}\big(\g_{\cBV, A}\big)/\mathrm{G}\) defined as the moduli space of Maurer--Cartan elements up to gauge group action.
\end{theorem}

\begin{proof}
	This is a direct application of \cite[Theorem~12.2.14]{LodayVallette12}.
\end{proof}

\subsubsection{Obstruction theory}

The weight grading of the Koszul dual cooperad \(\cBV^{\ac}\) by its cogenerators \(s \m, s \b, s \triangle, s \c, s \square\) induces a weight grading on the deformation dg pre-Lie algebra
\[
\g_{\cBV, A} \cong \prod_{w \geqslant 0} \g_{\cBV, A}^{(w)}.
\]
Maurer--Cartan elements concentrated in weight one are in one-to-one correspondence with \(\cBV\)-algebra structures on \(A\).
This decomposition yields an effective obstruction theory for \(\cBV_{\!\infty}\)-algebras, as shown for instance by the next theorem.

\begin{theorem}
	If
	\[
	H_{-2}\big(\g_{\cBV, A}^{(w)}, \partial_0\big) \cong 0,
	\]
	for all \(w \geqslant 2\), where
	\[
	\partial_0(f) \defeq \partial_A f - (-1)^{|f|} f d_1,
	\]
	then any dg commutative algebra \((A, d, \cdot)\) equipped with a degree \(+1\) linear operator \(\triangle\) and a degree \(+1\) symmetric binary operation \(b\) extends to a \(\cBV_{\!\infty}\)-algebra such that
	\[
	m^0_1 = d,
	\qquad
	m^0_2 = \cdot,
	\qquad
	m^1_1 = \triangle,
	\qquad
	m^0_{1,1} = b,
	\qquad
	n^1_1 = [d, \triangle],
	\qquad
	n^0_{1,1} = [d, b].
	\]
\end{theorem}

\begin{proof}
	This is a direct application of \cite[Theorem~12.2.12]{LodayVallette12}.
\end{proof}

In \cite[Theorem~28]{GCTV12}, this kind of obstruction theory was used to
endow the BRST complex (off-shell) of a topological vertex operator algebra with an explicit \(\BV_{\!\infty}\)-algebra lifting the \(\BV\)-algebra on cohomology (on-shell), answering a conjecture of Lian--Zuckerman \cite{LianZuckerman93}.

\begin{proposition}
	For every \(w \geqslant 0\), the chain complex \(\big(\g_{\cBV, A}^{(w)}, \partial_0\big)\) is acyclic whenever \((\End_A, \partial_A)\) is acyclic, for example when \((A, d)\) is acyclic.
\end{proposition}

\begin{proof}
	For any \(w \geqslant 0\), the complex \(\big(\g_{\cBV, A}^{(w)}, \partial_0\big)\) has the following form:
	\[
	\begin{tikzcd}[column sep=tiny]
		\Hom_{\Sy}(\Com^\ac(w+1), \End_A)
		\arrow[loop below,distance=2em, "(\partial_A)_*"]
		\!\!\!&\!\!\! \oplus \!\!\!&\!\!\!
		\Hom_{\Sy}((\BV^{\ac}/\Com^{\ac})^{(w)}, \End_A)
		\arrow[rr,bend left=45, "(d_1)^* = s", "\cong"']
		\arrow[loop below,distance=2em, "(\partial_A)_*"]
		\!\!\!&\!\!\! \oplus \!\!\!&\!\!\!
		s\Hom_{\Sy}((\BV^{\ac}/\Com^{\ac})^{(w)}, \End_A)
		\arrow[loop below,distance=2em, "(\partial_A)_*"].
	\end{tikzcd}
	\]
	Filtering by the homological degree of \(\End_A\) yields a convergent spectral sequence whose first page has differential \((d_1)^*\), an isomorphism between the last two summands.
	Thus the second page is
	\[
	(E^1, d^1) \cong \big(\Hom_{\Sy}(\Com^\ac(w+1), \End_A), (\partial_A)_*\big)
	\cong \big((\Com^\ac(w+1))^* \otimes_{\Sy} \Hom(A^{\otimes w+1}, A), \id \otimes \partial_A\big),
	\]
	which is acyclic under the assumption.
\end{proof}

Since the Koszul dual cooperad \(\cBV^{\ac}\) contains the Koszul dual cooperads \(\BV^{\ac}\), \(\Gerst^{\ac}\), and \(\Com^{\ac}\), one can set up \emph{relative} obstruction theories detecting when a \(\BV_{\!\infty}\)-algebra, a \(\Gerst_\infty\)-algebra, or a \(\Com_\infty\)-algebra can be lifted to a \(\cBV_{\!\infty}\)-algebra.
This is obtained by adapting \cite[Section~3.4]{GCTV12}, where the relative weight grading is given by the number of cogenerators in \(s^{-1}\rM^*\), \(s\triangle\) and \(s^{-1}\rM^*\), or \(\rM^* \oplus s^{-1}\rM^*\), respectively.

\subsubsection{Homotopy theory}

The third equivalent definition of a \(\cBV_{\!\infty}\)-algebra identifies such a structure on \(A\) with a square-zero coderivation (a \emph{codifferential}) on the conilpotent cofree dg \(\cBV^{\ac}\)-coalgebra \(\cBV^{\ac}(A)\).
Interpreted as an assignment, this leads to the following definition.

\begin{definition}\label{Def:BarConstr}
	The \defn{bar construction} of a \(\cBV_{\!\infty}\)-algebra \((A, \alpha)\) is the dg \(\cBV^{\ac}\)-coalgebra
	\[
	\mathrm{B}_\iota A \defeq (\cBV^{\ac}(A), d_1 + d_\psi + d_\alpha),
	\]
	where \(d_\alpha\) is the unique coderivation extending the map \(\cBV^{\ac}(A) \to A\) defined by \(\alpha\).
\end{definition}

\begin{proposition}\label{prop:FormBar}
	The underlying module of the bar construction is isomorphic to
	\[
	\cBV^{\ac}(A) \cong
	s^{-2}\, \mathrm{S}^c(s\mathrm{Lie}^c(sA))[\delta]
	\oplus
	s^{-3}\, \mathrm{S}^c(s\mathrm{Lie}^c(sA))[\delta]/s\mathrm{Lie}^c(sA).
	\]
\end{proposition}

\begin{proof}
	This follows directly from the argument in \cref{prop:DefpreLie}.
\end{proof}

\begin{definition}
	An \defn{\(\infty\)-morphism} \(A \rightsquigarrow B\) of \(\cBV_{\!\infty}\)-algebras is a morphism of dg \(\cBV^{\ac}\)-coalgebras between the corresponding bar constructions \(\mathrm{B}_\iota A \to \mathrm{B}_\iota B\).
\end{definition}

The data of an \(\infty\)-morphism \((A, \alpha) \rightsquigarrow (B, \beta)\) is equivalent to a degree zero element \(f \in \Hom_{\Sy}\big(\cBV^{\ac}, \End^A_B\big)\) satisfying
\[
\partial f = f \star \alpha - \beta \circledcirc f,
\]
see \cite[Theorem~10.2.3]{LodayVallette12} and \cite[Section~5]{DSV16}.
So an \(\infty\)-morphism consists of two families of maps
\begin{align*}
	f_{p_1,\dots,p_k}^t &\colon
	A^{\otimes p_1}\otimes\dots\otimes A^{\otimes p_k} \longrightarrow A,
	\qquad t \geqslant 0,\ k \geqslant 1,\ p_i \geqslant 1,\\
	g_{p_1,\dots,p_k}^t &\colon
	A^{\otimes p_1}\otimes\dots\otimes A^{\otimes p_k} \longrightarrow A,
	\qquad t \geqslant 0,\ k \geqslant 1,\ p_i \geqslant 1,\ t+k \geqslant 2,
\end{align*}
of respective degrees
\[
\bigl|f_{p_1,\dots,p_k}^t\bigr| = p_1 + \dots + p_k + k + 2t - 2,
\qquad
\bigl|g_{p_1,\dots,p_k}^t\bigr| = p_1 + \dots + p_k + k + 2t - 3,
\]
satisfying the same symmetries as the generating and obstruction maps of a \(\cBV_{\!\infty}\)-algebra, see \cref{ss:generating_maps} and \cref{ss:obstruction_maps}.
The \(\infty\)-morphisms of \(\cBV_{\!\infty}\)-algebras satisfy exactly the same pattern of relations as the structural maps described and proved in \cref{thm:cBVinfty}.
The collection of maps \(\{f_{p}^0\}_{p \geqslant 1}\) forms an \(\infty\)-morphism of \(\rC_\infty\)-algebras.
The maps \(g_{p_1,\dots,p_k}^t\) are precisely equal to the obstructions to the relations of the maps \(f_{p_1,\dots,p_k}^t\) defining an \(\infty\)-morphism of \(\BV_{\!\infty}\)-algebras.
Then, altogether, the maps \(g_{p_1,\dots,p_k}^t\) and \(f_{p_1,\dots,p_k}^t\) satisfy relations which are automatic from this prescription.
The proof is similar to that of \cref{thm:cBVinfty} and is left to the reader.

\medskip

The isomorphisms in the category of \(\cBV_{\!\infty}\)-algebras are the \defn{\(\infty\)-isomorphisms}, namely those \(\infty\)-morphisms whose first component \(f^0_1 \colon A \to B\) is a dg module isomorphism, see \cite[Theorem~10.4.1]{LodayVallette12}.
The Deligne groupoid of the deformation dg pre-Lie algebra \(\g_{\cBV, A}\) has \(\cBV_{\!\infty}\)-algebra structures on \(A\) as objects and \(\infty\)-isotopies, i.e.\ \(\infty\)-isomorphisms with first component the identity, as morphisms, see \cite[Theorem~3]{DSV16}.

\medskip

There is a bar-cobar adjunction between dg \(\cBV\)-algebras and conilpotent dg \(\cBV^{\ac}\)-coalgebras
\[
\begin{tikzcd}
	\Omega_\kappa \colon \text{conil.\ dg } \cBV^{\ac}\text{-coalgebras} \arrow[r, shift left=1.75, harpoon, "\perp"']
	& \text{dg } \cBV\text{-algebras} \arrow[l, shift left=1.75, harpoon] \colon \mathrm{B}_\kappa,
\end{tikzcd}
\]
where the underlying constructions are given respectively by the free algebra and the cofree coalgebra functors, see \cite[Section~11.2]{LodayVallette12}.

\begin{theorem}[Rectification]
	Any \(\cBV_{\!\infty}\)-algebra \(A\) is naturally \(\infty\)-quasi-isomorphic to the canonical dg \(\cBV\)-algebra
	\[
	A \stackrel{\sim}{\rightsquigarrow} \Omega_\kappa \mathrm{B}_\iota A.
	\]
\end{theorem}

\begin{proof}
	This is a direct application of \cite[Theorem~11.4.4]{LodayVallette12} to the Koszul operad \(\cBV\).
\end{proof}

This rectification extends to \(\infty\)-morphisms by \cite[Proposition~11.4.5]{LodayVallette12}.

\medskip

Recall that a \defn{contraction} of a chain complex \((A, d_A)\) onto another chain complex \((H, d_H)\) is a homotopy \(h\) of degree \(1\) together with chain maps \(i\) and \(p\) such that
\[
\begin{tikzcd}
	(A, d_A) \arrow[r, shift left, "p"] \arrow[loop left, distance=1.5em, "h"]
	& (H, d_H) \arrow[l, shift left, "i"]
\end{tikzcd}
\]
and satisfying
\[
pi = \id_H,
\qquad ip - \id_A = d_A h + h d_A,
\qquad hi = 0,
\qquad ph = 0,
\qquad h^2 = 0.
\]

\begin{theorem}[Homotopy transfer theorem]
	For any \(\cBV_{\!\infty}\)-algebra structure \(\alpha\) on \((A, d_A)\) and any contraction onto a chain complex \((H, d_H)\), there exists a \(\cBV_{\!\infty}\)-algebra structure on \(H\) defined by the composite
	\[
	\begin{tikzcd}[column sep=large]
		\overline{\cBV}^{\ac} \arrow[r, "\Delta_{\cBV^{\ac}}"] & \mathcal{T}^c(\overline{\cBV}^{\ac})
		\arrow[r, "\mathcal{T}^c(s\alpha)"] & \mathcal{T}^c(s\End_A)
		\arrow[r, "\Psi"] & \End_H,
	\end{tikzcd}
	\]
	where \(\Delta_{\cBV^{\ac}}\) is the comonadic decomposition map of the cooperad \(\cBV^{\ac}\) and \(\Psi \colon \End_A \rightsquigarrow \End_H\) is the Van der Laan \(\infty\)-morphism.
	Both maps \(i\) and \(p\) extend to \(\infty\)-quasi-isomorphisms between \(\alpha\) and this transferred structure.
\end{theorem}

\begin{proof}
	This is a direct application of \cite[Theorem~10.3.1]{LodayVallette12}.
	Explicit formulas for the extensions to \(\infty\)-quasi-isomorphisms are given in \cite[Theorem~10.3.6]{LodayVallette12} and \cite[Proposition~10.3.9]{LodayVallette12}.
\end{proof}

The homotopy transfer theorem admits a conceptual interpretation via perturbation theory and the gauge group action, see \cite[Section~8]{DSV16}.
It is functorial with respect to the Koszul dual cooperad by \cite[Proposition~3.18]{HLV21}.
Applied to the homology \(H = H_\bullet(A)\) of a \(\cBV_{\!\infty}\)-algebra, it produces \defn{Massey products} for homotopy coexact Batalin--Vilkovisky algebras.
A \(\cBV_{\!\infty}\)-algebra with trivial differential \(m^0_1 = 0\) is called \defn{minimal}.

\begin{proposition}[Minimal model]
	Every \(\cBV_{\!\infty}\)-algebra is \(\infty\)-isomorphic to the product of a minimal \(\cBV_{\!\infty}\)-algebra, given by the transferred structure on its homology, with an acyclic trivial \(\cBV_{\!\infty}\)-algebra.
\end{proposition}

\begin{proof}
	This is a direct application of \cite[Theorem~10.4.3]{LodayVallette12}.
\end{proof}

This implies that any \(\infty\)-quasi-isomorphism of \(\cBV_{\!\infty}\)-algebras admits a homotopy inverse, see \cite[Theorem~10.4.4]{LodayVallette12}.
The \defn{Kaledin--Emprin classes} introduced in \cite{Emprin24} provide faithful obstructions for the formality of \(\cBV_{\!\infty}\)-algebras.
Finally, Theorem~0.25 of \cite{CPRN24} applied to the split morphism \(\Com_{(n)} \to \cBV_{(n)}\) shows that the formality of \(\cBV_{\!\infty}\)-algebras reduces to the formality of their underlying \(\Com_\infty\)-algebras.

%% file: sec/dual_algebras.tex

\section{Koszul dual Batalin--Vilkovisky algebras are Belinson--Drinfeld algebras}\label{ss:bv_bd}

In this appendix, we make explicit the notion of algebras over the (co)operad Koszul dual to the Batalin--Vilkovisky operad and we show that they coincide with the various notions of Belinson--Drinfeld algebras.

\begin{proposition}\label{prop:DualBValg}
	The desuspension of a differential graded \(\BV^!\)-algebra structure on \(sA\) is
	a Gerstenhaber \(\mathbb{K}[\hbar]\)-algebra \((A, m ,\{\, , \,\} )\) on a chain
	\(\mathbb{K}[\hbar]\)-complex \((A,d)\), with \(|\hbar|=-2\), such that the differential \(d\) is a derivation
	with respect to \(\{\, , \,\}\) and such that
	\begin{equation}\label{eq:dKoszuldual}
	d(m(a,b)) = m(d(a),b) + (-1)^{|a|} m(a,d(b)) + \hbar\{a,b\}.
	\end{equation}
\end{proposition}

\begin{proof}
	All the arguments are already contained in the proof of \cref{prop:DualcBValg} since \(\BV^!\) is the dg suboperad of \(\cBV^!\)
	spanned by the three generators \(\m'\), \(\triangle'\), and \(\b'\). As explained in the proof of \cref{prop:DualcBValg},
	the derivative of Relation~\eqref{Eq:derivmu} produces Relation~\eqref{eq:dKoszuldual}. The former does not hold in the
	 desuspension of a dg \(\BV^!\)-algebra but the latter does.
\end{proof}

In short plain words, the desuspension of a \(\mathbb{K}\)-algebra over the operad \(\BV^!\) is a differential graded Gerstenhaber
\(\mathbb{K}[\hbar]\)-algebra \((A, d, m, \{\, , \,\})\), with \(|\hbar|=-2\), except for the derivation relation with respect to the commutative product \(m\) which is replaced by
\[ [d, m]=\hbar \{\, , \,\}.\]

In the context of chiral algebras, Beilinson--Drinfeld introduced in \cite{BeilinsonDrinfeld04} an algebraic structure closely related to Koszul dual Batalin--Vilkovisky algebras. It admits at least two versions: one with a formal parameter \(\hbar\) in order to treat
 perturbative Quantum Field Theories \cite{CG16, CG21} and one with a polynomial parameter \(\hbar\) enough to treat nonpeturbatively free theories \cite{GH18}.

\begin{definition}\label{def:BDalg}
	A \emph{Beilinson--Drinfeld (BD) algebra} is a Gerstenhaber \(\mathbb{K}[\hbar]\)-algebra \((A, m ,\{\, , \,\} )\) on a chain
	\(\mathbb{K}[\hbar]\)-complex \((A,d)\), with \(|\hbar|=0\), such that the differential \(d\) is a derivation
	with respect to \(\{\, , \,\}\) and such that
	\begin{equation*}
	d(m(a,b)) = m(d(a),b) + (-1)^{|a|} m(a,d(b)) + \hbar\{a,b\}.
	\end{equation*}
	An \emph{absolute BD-algebra} is defined similarly but over the ring of power series \(\KK[\![\hbar]\!]\).
\end{definition}

To be self coherent, we used the homological degree convention in this definition, on the opposite to the abovementioned references where the Lie bracket has cohomological degree \(+1\). Let us denote by \(\BV_{-1}\) the operad encoding Batalin--Vilkovisky algebra with the operator \(\Delta\) and the Lie bracket placed in homological degree \(-1\).

\begin{theorem}
	BD-algebras are suspensions of \(\BV_{-1}^!\)-algebras and absoute BD-algebras are absolute \(\BV_{-1}^{\ac}\)-algebras.
\end{theorem}

\begin{proof}
	The first point is a direct consequence of \cref{prop:DualBValg} with modified homological degrees accordingly.
	For the notion of absolute algebras over a cooperad, we refer to \cite[Section~3]{RocaiLucio25}. The Koszul dual of the degree \(-1\) operator \(\Delta\) is a generating element of degree \(0\) and arity \(1\) in the Koszul dual cooperad \(\BV_{-1}^{\ac}\). The arguments of Example~4.6 of \emph{loc. cit.} show that its action on an absolute \(\BV_{-1}^{\ac}\)-algebra is equivalent to an action of the ring of power series \(\KK[\![\hbar]\!]\). One concludes with the same arguments as in Section~4.2 of \emph{loc. cit.}.
\end{proof}

So BV-algebras are Koszul dual to BD-algebras, up to a mild change of degree convention.
Therefore the toolkit provided by the operadic calculus \cite[Chapter~11]{LodayVallette12} and \cite{RocaiLucio25} settles a bar-cobar adjunctions between categories of BV-(co)algebras and BD-(co)algebras, conilpotent or absolute, which induce Quillen equivalences \cite{Vallette14}.